\documentclass[acmsmall]{acmart}

\newcommand{\doi}[1]{\textsc{doi}: \href{http://dx.doi.org/#1}{\nolinkurl{#1}}}

\AtBeginDocument{%
  }

\setcopyright{acmlicensed}
\copyrightyear{2024}
\acmYear{2024}
\acmDOI{XXXXXXX.XXXXXXX}
\usepackage{nicefrac} 
\usepackage{ebproof} 
\usepackage{stmaryrd}
\usepackage{caption}
\usepackage{subcaption}
\usepackage{xfrac}

\usepackage{thm-restate}

\usepackage{paralist}

\newcommand\citewiththeorem[2]{{\cite[#1]{#2}}} 

\newcommand\pair[2]{\ensuremath{\left( #1,#2\right)}}
\newcommand\tlet[4]{\ensuremath{\mathsf{let}\ \pair{#1}{#2}=#3\ \mathsf{in}\ #4}}
\newcommand\tif[3]{\ensuremath{\mathsf{if}\ #1\ \mathsf{then}\ #2\ \mathsf{else}\ #3}}

\newcommand\Ct{\ensuremath{\overline{\mathbb{Q}}}}
\newcommand\B{\ensuremath{\mathbb B}}

\newcommand\N{\ensuremath{\mathbb N}}

\newcommand{\lpar}[1]{\ensuremath{\llparenthesis\, #1 \,\rrparenthesis}}

\newcommand{\punq}{\textnormal{PUNQ}}
\newcommand{\dlal}{\textnormal{DLAL}}
\newcommand{\ortho}{\ensuremath{\mathsf{ORTHO}}}
\newcommand{\card}{\ensuremath{\textnormal{card}}}

\newcommand\T{\mathtt{T}}
\newcommand\SUP{\mathtt{S}}
\newcommand\BV{\mathtt{BV}}
\newcommand\V{\mathtt{V}}
\newcommand\CF{\text{CF}}

\newcommand\FBQP{\text{FBQP}}
\newcommand\BQP{\text{BQP}}

\newcommand{\CHK}{\textnormal{\textsf{CHK}}}

\newcommand\ls{Lambda-$\mathcal{S}_1$}

\newcommand\lmap[1]{\ensuremath{\lfloor #1 \rfloor}}

\newcommand\dom{\text{dom}}
\newcommand\subsetast{\ensuremath{\sqsubseteq}}
\newcommand\eqast{\ensuremath{\simeq}}
\newcommand{\Mod}[1]{\ (\mathrm{mod}\ #1)}

\usepackage{tikz}
\usetikzlibrary{quantikz2}
\renewcommand\bra[1]{{\langle{#1}|}}
\makeatletter
\renewcommand\ket[1]{%
  \@ifnextchar\braket{\k@t{#1}\!\!}{\k@t{#1}\!}%
}
\newcommand\k@t[1]{{|{#1}\rangle}}
\makeatother

\begin{document}
\title{A feasible and unitary quantum programming language} 
\thanks{Partially funded by the French-Argentinian IRP SINFIN. Díaz-Caro is also funded by PICT 2021-I-A-00090 and 2019-1272, and PIP 11220200100368CO}

\author{Alejandro Díaz-Caro}
\affiliation{%
\institution{Universidad de Buenos Aires-CONICET, ICC}
\city{Buenos Aires}
\country{Buenos Aires, Argentina}}
\affiliation{%
\institution{Universidad Nacional de Quilmes, DCyT}
\city{Bernal}
\country{Bernal, Argentina}}
\author{Emmanuel Hainry}
\author{Romain Péchoux}
\author{Mário Silva}
\affiliation{%
\institution{Université de Lorraine, CNRS, Inria, LORIA}
\city{Nancy}
\country{F54000 Nancy, France}
\postcode{F54000}}
\authorsaddresses{}

\begin{abstract}
  We introduce a novel quantum programming language featuring higher-order programs and quantum control flow which ensures that all qubit transformations are unitary. Our language boasts a type
  system guaranteeing both unitarity and polynomial-time normalization.
  Unitarity is achieved by using a special modality for superpositions
  while requiring orthogonality among superposed terms. Polynomial-time
  normalization is achieved  using a linear-logic-based type discipline employing Barber and Plotkin
  duality along with a specific modality to account for potential duplications. This type
  discipline also guarantees that derived values have polynomial size. Our language
  seamlessly combines the two modalities: quantum circuit programs uphold
  unitarity, and all programs are evaluated in polynomial time, ensuring their feasibility.
\end{abstract}

\maketitle

\renewcommand{\shortauthors}{A.~Díaz-Caro, E.~Hainry, R.~Péchoux, and M.~Silva}

\section{Introduction}
\subsection{Motivation}
\emph{Classical control vs quantum control.}
Quantum programming languages can be classified into two primary categories
based on their control flow handling. On the one hand, \emph{classical control}~\cite{S04} involves quantum operations
executed on a specialized device within a classical computer. The classical
computer manages program execution by instructing which quantum operation to
apply to specific qubits. This approach resembles circuit description languages
that use high-level operations on quantum circuits, with examples such as the
quantum lambda calculus~\cite{SV06}, Quipper~\cite{QUIPPER}, and
Qwire~\cite{QWIRE}. Ensuring physical implementability typically involves
constraints and linearity-based type systems on quantum data to maintain
essential quantum physics properties such as the no-cloning
theorem~\cite{WoottersZurek82}. Notably, this category accommodates models like
QRAM~\cite{K96}.

On the other hand, \emph{quantum control}~\cite{DiazcaroLSFA21} allows programming
quantum operations based on quantum data, a fundamental concept in quantum
computing. For example, the CNOT operation governs the application of a NOT
operation based on a control qubit. More advanced examples like the Quantum
Switch~\cite{QuantumSwitch15,QuantumSwitch17} have gained popularity, enabling
the application of two operations in different orders based on a control qubit.

While both categories offer equivalent computational power, quantum control
provides a more natural approach, allowing programs to be ``fully quantum''
with operational control flow over quantum data. Programming languages like
QML~\cite{AG05}, Lambda-$\mathcal{S}_1$~\cite{DCM22}, and Qunity~\cite{QUNITY}
fall into this category. These programs must overcome significant
constraints to ensure physical implementability, particularly in terms of
efficient compilation into quantum circuits or low-level models like QRAM.

\paragraph{Unitarity.}
One major issue in quantum control is ensuring \emph{unitarity}, a fundamental
property of quantum systems that maintains the total probability of all
possible outcomes over time. Quantum gates in circuits are represented by
unitary maps, preserving $\ell^2$-norm and orthogonality. However,
representing quantum data as linear combinations in a Hilbert space leads to a
problem: programs must have an $\ell^2$-norm of $1$ (i.e., must lie in the
unitary sphere) to be physically realizable, often requiring orthogonality
among program branches. This property is generally not preserved by program
semantics, leading to the need to restrict programs to those satisfying
unitarity. The question of ensuring this restriction through type systems has
been first explored by~\cite{AG05} to some extent, and more recently
by~\cite{DCGMV19}, characterizing superpositions and isometries with
realizability techniques~\cite{V08,M11}. In this realizability model, types are
interpreted as subsets of a vector space's unit sphere, with all typed terms
preserving the $\ell^2$-norm, and quantum data expressed as superpositions of
classical data types.  The corresponding typing discipline, that is in a way
the dual to Intuitionistic Linear Logic~\cite{GL87,HDP93},  has been introduced
in~\cite{DCM22} to delineate a  programming language for unitarity called
Lambda-$\mathcal{S}_1$.

\paragraph{Feasibility.}\label{ss:f}
While unitarity is essential for quantum program implementability, it is
insufficient. \emph{Feasibility}~\cite{G83} (or \emph{tractability}), the property that a program can be executed within reasonable time and space-constraints, is equally
crucial. Compiling quantum programs to low-level models, like quantum
circuits, requires imposing restrictions on qubit count, gate count, and
error rates. Achieving feasibility involves studying program computational
complexity, typically related to polynomial-depth uniform circuit families.
Yao's Theorem~\cite{Y93} links such families to Bounded-error Quantum
Polynomial time (\textsc{bqp})~\cite{BV97}, a quantum analogue of the probabilistic complexity class $\textsc{bpp}$.
Feasibility has deep roots in classical complexity, leading to fields like
descriptive complexity~\cite{I12} and implicit computational
complexity~\cite{P20} characterizing complexity classes logically and through
programming languages.

Previous work~\cite{DLMZ10}, based on light linear logic~\cite{G94},
characterizes polynomial time in quantum lambda calculus. However, no type
system currently ensures quantum program feasibility with quantum control.
Thus, a major problem in quantum computing is to develop programming languages
with quantum control, typed in such a way as to ensure both unitarity and
feasibility of the programmed functions, so that they have a physical implementation that does not break the laws of quantum mechanics, and can (at least in principle) be efficiently compiled
into a circuit.

\subsection{Contribution}
We present $\punq$, a typed quantum programming language with quantum control
addressing unitarity and feasibility. Quantum control employs a quantum conditional inspired
by QML~\cite{AG05}, producing superposed outputs from superposed inputs.
Unitarity is achieved via a variant of Lambda-$\mathcal{S}_1$~\cite{DCM22},
introducing a modality $\sharp$ for superpositions, addressing quantum
no-cloning. This modality marks non-duplicable types, treating superpositions
of type $\sharp A$ linearly. Only terms of types distinct from
$\sharp A$ can be duplicated. For example, $\sharp \B$ denotes superpositions of Booleans
(qubits), and realizers of $\sharp \B \multimap \sharp \B$ represent single-qubit quantum
gates.
Feasibility is ensured by a variant of the Dual Light Affine Logic type system
($\dlal$)~\cite{BT09}, employing a  modality $\S$ to account for potential
duplication and a duality à la~\cite{BP96}, with two arrows corresponding to a non-linear context and a linear context, respectively.  In this setting, linear arrows are strictly linear to preserve unitarity and non-linear arrows cannot be applied to superpositions, thus, ensuring the no-cloning principle of quantum mechanics.

The main contributions of this paper are:
\begin{itemize}
  \item a new typed programming language with quantum control that enjoys
    subject reduction (Theorem~\ref{thm:sr}) and progress (Theorem~\ref{thm:progress}),
  \item a soundness result (Theorem~\ref{thm:unitarity}) showing that typable linear maps over qubits encode isometries and unitary operators when dimensions match,
  \item a completeness result (Theorem~\ref{thm:isocomplete}) stating that any isometry can be encoded by a $\punq$ program, 
  \item a non-separability result (Theorem~\ref{thm:non-separability}): there is no linear map that can separate qubits,
  \item a feasibility result (Theorem~\ref{thm:soundness}), ensuring polynomial time normalization and polynomially bounded size of normal forms,
  \item a complexity result for type checking over three fragments of $\punq{}$ (Theorem~\ref{thm:chk-complexity}). We identify two fragments in which we can check orthogonality in polytime and one in which it is potentially undecidable (Lemma~\ref{lemma:orthos-complexity}), and give examples of the expressivity of each fragment.
  
\end{itemize}

In summary, $\punq$ can be viewed as the first feasible and physically
realistic quantum programming language with quantum control. We provide several simple examples to illustrate our results: towards completeness, we show that standard one and two-qubit gates can be simulated by programs (Examples~\ref{ex:hadamard}, \ref{ex:z}, and \ref{ex:controlled-not}), we provide the encoding of a simple quantum teleportation protocol (Example~\ref{ex:teleportation}) illustrating that soundness can be used to certify unitarity, we also provide a quantum random walk algorithm (Example~\ref{ex:random-walk}) illustrating polynomial time normalization.

\subsection{Related Work}
\emph{Quantum control and unitarity.} 
QML~\cite{AG05} was the pioneering language to introduce a quantum
if-statement, enabling superposition in branches based on the superposition in
the if-statement guard. The $\ell^2$-norm preservation is ensured through a
semantic notion of validity: an if-statement is valid when its branches are
orthogonal, effectively reducing the two branches to orthogonal values.
Another approach for handling quantum control and superpositions, which is less
semantics-driven, is the introduction of Lineal~\cite{AD17}. It is an untyped
lambda calculus extended with superpositions of terms. Lineal strictly covers
measurement-free quantum programs, as it does not involve orthogonality checks
and does not enforce superpositions to have norm $1$. However, it treats
function application linearly, providing a generalization of the QML
if-statement.  Lambda-$\mathcal S_1$~\cite{DCGMV19,DCM22}, as well as the
current work, can be seen as the $\ell^2$-norm preserving restriction of
Lineal. Lambda-$\mathcal S$~\cite{DiazcaroDowekRinaldiBIO19,DCM23}, a preliminary
version of this typing discipline, does not ensure unitarity but includes
measurements.  Qunity~\cite{QUNITY} enforces the no-cloning property using constrained sharing, allowing duplicated variables to produce only
entangled states in a basis-dependent way. While this ensures that well-typed programs in Qunity have a quantum circuit representation, the typing does not guarantee any resource bound, neither on the time complexity of the programs, nor on the size of the compiled circuit. 
There have also been attempts~\cite{Y16} to define a notion
of ``quantum alternation'' allowing measurements to be quantum-controlled, where
the resulting system is not monotone with respect to the Löwner order and,
hence, cannot be considered to be a physically feasible concept~\cite{BP15}.
An alternative approach in~\cite{Y23} characterizes unitarity using a model
based on injective semantics. Complementary approaches following the
ZX-calculus research line provide graphical languages for quantum control with
quantum tests~\cite{Chardonnet23,CDVVV22}.  Currently, none of the mentioned systems can
guarantee the feasibility of their programs in terms of complexity or resource
requirements.

\paragraph{Quantum complexity classes.}
Programming-language-based characterizations of well-known complexity classes
have been deeply studied in the field of Implicit Computational Complexity
(see~\cite{P20} for a survey). To mention a few of them, ~\cite{BC92} provided
the first implicit (i.e., where the complexity bound does not need to be explicited
by the programmer) characterization of polynomial time and~\cite{GMR08} is the
first lambda-calculus characterization of polynomial space.  Although a great
deal of work has been done in this field over the last three decades, only a
small number of it has focused on the quantum paradigm. The paper~\cite{DLMZ10}
characterizes $\BQP$ on the quantum lambda-calculus. However, due to the
presence of unrestricted measurement, this classically-controlled system cannot
guarantee unitarity. Finally,~\cite{HPS23,Y20} provide two characterizations
of $\FBQP$ with quantum control which are restricted to first-order.

\subsection{Illustrating Example: Grover's Quantum Search Algorithm}

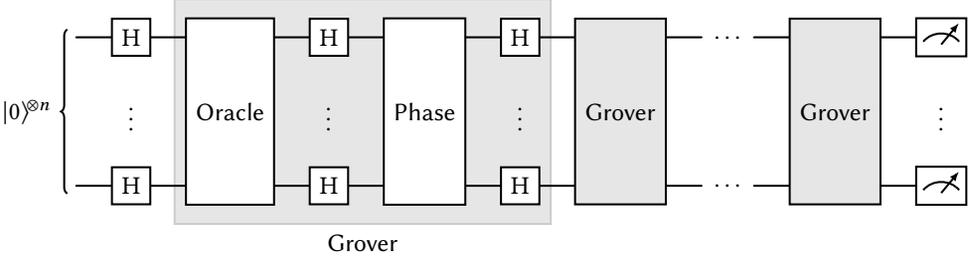
\begin{figure}
\[
\scalebox{0.95}{
   \begin{quantikz}[wire types={q,n,q}]
	\lstick[3]{$\ket{0}^{\otimes n}$}  & \gate{\text{H}} &\gate[3]{\mathsf{Oracle}}\gategroup[3,steps=4,style={draw=gray!40,fill=gray!20, inner	xsep=1pt},background,label style={label position=below,anchor=north,yshift=-0.2cm}]{$\mathsf{Grover}$} & \gate{\text{H}} & \gate[3]{\mathsf{Phase}}  & \gate{\text{H}} & \gate[3, style={fill=gray!20}]{\mathsf{Grover}} &\ \ \ldots\ \ & \gate[3, style={fill=gray!20}]{\mathsf{Grover}} & \meter{}\\
	& \vdots & & \vdots & & \vdots & & & & \vdots\\
	 & \gate{\text{H}} & &\gate{\text{H}} & &  \gate{\text{H}}& &\ \ \ldots\ \ & & \meter{}
	\end{quantikz}}
\]
\caption{Quantum search algorithm.}
\label{fig:grovers-algorithm}
\end{figure}

As an example of a \punq{} program, we will consider the algorithm for quantum search~\cite{G96}.
We are given a function $f:\{0,1\}^n\to\{0,1\}$ such that, for exactly one input $w\in\{0,1\}^n$, we have that $f(w)=1$. In the classical case, finding $w$ among $N\triangleq 2^n$ possibilities has average complexity $O(N)$, whereas using Grover's algorithm for quantum search, the value of $w$ can be found with high probability with only $O(\sqrt{N})$ operations (see Figure~\ref{fig:grovers-algorithm}).

Let $\B$ be the type of Booleans. For $n \geq 1$, define the type of tuples of bits as $\B^{n+1}\triangleq \B\times \B^{n}$ with $\B^1 \triangleq \B$.  The type $\sharp \B$ corresponds to qubits and $\sharp (\B^n)$, with $n \geq 1$, is the type of a tuple of $n$ (possibly entangled) qubits. For example, a superposition $\ket{\pm} \triangleq \frac{1}{\sqrt{2}}\cdot\ket{0}\pm\frac{1}{\sqrt{2}}\cdot\ket{1}$ has type $\sharp\B$ and, given a string $x=x_1\dots x_n\in\{0,1\}^n$, the tensor product state $\ket{x} \triangleq \ket{x_1}\otimes \dots \otimes \ket{x_n}$ can be encoded by a term of type $\sharp (\B^n)$. 

We will now describe a \punq{} term encoding the algorithm for Grover search in the case $n=2$. A very basic component is the Hadamard transformation 
\[\mathsf{H}\triangleq \lambda x.\tif{x}{\ket{+}}{\ket{-}}:\sharp\B\multimap \sharp \B,\]
 described in more detail in Example~\ref{ex:hadamard}. Intuitively, the type $\sharp (\B^n)\multimap \sharp (\B^n)$ corresponds to linear maps (unitary gates) over $n \geq 1$ qubits. The $\mathsf{if}$ construction assigns the first case to state $\ket{0}\,$ and the second to $\ket{1}\,$. Therefore, for $\alpha, \beta \in \mathbb{C}$, we have the reduction
\begin{align*}
\mathsf{H}\ (\alpha\cdot\ket{0}+\beta\cdot\ket{1}\,)\rightsquigarrow \alpha \cdot \tif{\ket{0}\,}{\ket{+}\,}{\ket{1}\,}\ +\beta \cdot \tif{\ket{1}\,}{\ket{+\,}}{\ket{1}\,}\rightsquigarrow \alpha\cdot \ket{+}\ +\beta\cdot \ket{-}
\end{align*}
We can then easily design a term applying the Hadamard gate to two qubits
\begin{align*}
\mathsf{H}_2&\triangleq  \lambda z.\tlet{x}{y}{z}{\pair{\mathsf{H}\ x}{\mathsf{H}\ y}}: \sharp(\B^2)\multimap \sharp (\B^2).
\end{align*}

We can likewise define the phase shift operator $P$, where $P\ \ket{x}\triangleq - \ket{x}$ for all $x\not = 0^n$ and $P\ \ket{0^n}\triangleq \ket{0^n}\,$. For $n=2$, we can encode $P$ by the term $\mathsf{Phase} : \sharp (\B^2)\multimap \sharp (\B^2)$:
\begin{align*}
\mathsf{Phase}\triangleq \lambda z.\tlet{x}{y}{z}{&\tif{x}{\bigl(\tif{y}{\pair{\ket{0}\,}{\ket{0}\,}}{-1\cdot \pair{\ket{0}\,}{\ket{1}\,}}\bigr)}{-1\cdot \pair{\ket{1}\,}{y}}}
\end{align*}

We will consider an oracle for $f$ given by a unitary $O$ where $O\ket{x}\triangleq (-1)^{f(x)}\ket{x}\,$, for all $x\in\{0,1\}^n$, i.e., where the oracle acts as the identity for all basis states except for $w$, on which it performs a phase shift of -1. This represents a unitary gate of dimension $n$, and therefore there exists a \punq{} term, say $\mathsf{Oracle}$, with type $\sharp(\B^n)\multimap \sharp (\B^n)$, that encodes it.
We may now define the Grover iteration step in the algorithm:
\[\mathsf{Grover}\triangleq \lambda x. \mathsf{H}_2\ (\mathsf{Phase}\ (\mathsf{H}_2\ (\mathsf{Oracle}\ x))): \sharp(\B^2)\multimap \sharp(\B^2).\]

In order to iterate the Grover step, we use the Church numeral encoding $\underline{n}\triangleq \lambda f.\lambda x. f^n(x)$ with type $\mathbb{N}\triangleq \forall X.(X\multimap X)\Rightarrow \S (X\multimap X)$. The type of Church numerals makes use of two extra constructs: a non-linear arrow $\Rightarrow$, whose input cannot be quantum data (superpositions), so that $\punq$ programs preserve the laws of quantum mechanics (e.g. no-cloning); a modality $\S$ from Girard's light linear logic~\cite{G94}, accounting for possible duplication in a term. 

This term can be applied to our term $\mathsf{Grover}$ with the substitution $X=\sharp(\B^2)$.  Notice that it does not break no-cloning as it roughly corresponds to a term of type $(\sharp(\B^2)\multimap \sharp(\B^2))\Rightarrow \S (\sharp(\B^2)\multimap \sharp(\B^2))$, whose input can be duplicated. Indeed, the input here is a quantum gate of type $(\sharp(\B^2)\multimap \sharp(\B^2))$ which does not constitute quantum data.

As such, we obtain the final term:
\[\mathsf{Search}\triangleq \lambda m.\lambda x. (m\ \mathsf{Grover})(\mathsf{H}_2\ x):\mathbb{N}\multimap \S \sharp(\B^2)\multimap \S \sharp(\B^2).\]

The term $\mathsf{Search}\ \underline{m} :\S \sharp(\B^2)\multimap \S\sharp(\B^2)$ simulates  a precise run of $m$ iterations  and produces as output a possibly entangled two-qubit state -- which is indeed the output before we perform any measurements. 

The non-trivial properties ensured by our type discipline on the above term are the following:
\begin{itemize}
\item By Theorem~\ref{thm:unitarity}, for any given integer $m$, since the term $\mathsf{Search}\ \underline{m}$ has type $\S \sharp(\B^2)\multimap \S\sharp(\B^2)$, it represents a unitary transformation with algebraic coefficients (we will later put a restriction on complex numbers to avoid the consideration of non-computable numbers), which can therefore be physically implemented in a quantum circuit, such as Figure~\ref{fig:grovers-algorithm}.
\item By Theorem~\ref{thm:soundness} (\textit{Polynomial time normalization}) the \punq{} type discipline ensures that the typed term $\mathsf{Search}\ \underline{m}\ \pair{\ket{0}\,}{\ket{0}\,}:\S \sharp(\B^2)$ reduces to a normal form $\sum_{i,j\in\{0,1\}}  \alpha_{ij} \cdot \pair{\ket{i}\,}{\ket{j}\,}$ in a number of steps that is polynomial on the original size of the term (in this case, the number of Grover iterations is linear on the parameter $m$ given -- for a more elaborate iteration example, see Example~\ref{ex:random-walk}). Notice that this property remains valid for any (possibly entangled) pair encoding a two-qubit state given as input.
\end{itemize}

\section{A programming language with quantum control}
\subsection{Syntax}
$\punq$ (short for Polytime UNitary Quantum language) is a programming language
with syntax defined by the grammar in Figure~\ref{fig:syntax}. A \emph{term}
can take the form of a variable $x$, a bit $\ket{0}$ or $\ket{1}$, a
conditional statement, an abstraction, an application, a pair, or a pair
destructor. We denote the set of terms as $\T$, and the terms are denoted by $r$,
$s$, $t_1$, $t_2$, and so on.
A \emph{superposition} can be either a term $t$, the null vector $\vec 0$, the
product $\alpha \cdot \vec t$ of a superposition $\vec t$ with an algebraic number $\alpha\in\Ct$~\cite{ADH97}, or the sum $\vec t_1 + \vec t_2$ of two superpositions. We represent the set of superpositions as $\SUP$, and superpositions themselves are denoted by $\vec r$, $\vec s$, $\vec t_1$, $\vec t_2$, and so forth.
In essence, terms correspond to objects that can reduce classically, possibly to a superposition, while superpositions represent quantum computations. We define the sets $\BV$ and $\V$ as the sets of \emph{basis values} and \emph{values}, respectively. Basis values are a subset of terms in normal form,
and values are superpositions of basis values. We use $v$, $w$, $v_1$, $v_2$,
and so on for basis values, while $\vec v$, $\vec w$, $\vec v_1$, $\vec v_2$,
and so on, denote values.
\begin{figure}[]
  \hrulefill
  \begin{align*}
    (\text{Terms}) &&\T \ni t  &:= x \mid\ket 0\mid\ket 1 \mid \tif t{\vec t}{\vec t}\mid  {\lambda x.\vec t}\mid t\ {t} \mid{\pair{t}{t}}  \mid\tlet xyt{\vec t}   \\
   (\text{Superpositions}) && \SUP \ni \vec t  &:=   t\mid \vec 0 \mid \alpha \cdot \vec t\mid\vec t+\vec t \\
   (\text{Basis values}) && \BV \ni v  &:= \ket 0\mid\ket 1\mid\lambda x.\vec t\mid\pair{v}{ v} \\ 
    (\text{Values}) && \V \ni \vec v  &:= v\mid \vec 0 \mid \alpha \cdot \vec v\mid\vec v+\vec v
  \end{align*}
  \hrulefill
  \caption{Syntax of $\punq$ programs.}
  \label{fig:syntax}
\end{figure}

A variable is free in a superposition if it is not bound by an abstraction or a
pair destructor. We denote the set of free variables in the superposition $\vec
t$ as $FV(\vec t)$. A superposition is closed if it contains no free variables.
Given a set $\mathcal{S}$ of superpositions, we define $\mathcal{S}_c$ as the
set of closed superpositions within $\mathcal{S}$. A $\punq$ program is a
closed superposition in $\SUP_c$ that can be assigned a type according to the
type discipline presented in Section~\ref{s:type}.

The size of a superposition $\vec t$, denoted $|\vec t|$, is the maximal size
of its superposed terms. Formally,
%
%
\[
  \begin{array}{r@{\ }l@{\quad\qquad}r@{\ }l}
    |x|&\triangleq 1& 
    |t_1\ t_2|=|(t_1, t_2)|&\triangleq|t_1|+|t_2|+1\\
    |\,\ket{0}\,|&\triangleq 1&
    |\tlet xyt{\vec t}| &\triangleq 1+|t|+|\vec t|\\ 
    |\,\ket{1}\,|&\triangleq 1&
    |\vec 0| &\triangleq 0\\ 
    |\tif t{\vec t_1}{\vec t_2}|&\triangleq 1+ |t|+\max(|\vec t_1|,|\vec t_2|)& 
    |\alpha \cdot \vec t| &\triangleq |\vec t|\\ 
    |{\lambda x.\vec t}|&\triangleq 1+|\vec t|& 
    |\vec t_1 + \vec t_2| &\triangleq \max(|\vec t_1|,|\vec t_2|)
  \end{array}
\]

The set $\SUP$ of  superpositions has the structure of a vector space and,
consequently, we define an equivalence relation $\equiv$ on $\SUP$ as follows:
%
\[
  \begin{array}{r@{\ \,}l@{\qquad\qquad}r@{\ \,}l}
    \vec t_1+\vec t_2 &\equiv\vec t_2+\vec t_1 &
    (\vec t_1+\vec t_2)+\vec t_3 &\equiv \vec t_1+(\vec t_2+\vec t_3)\\
    \vec 0 + \vec t &\equiv \vec t  &
    0\cdot \vec t &\equiv  \vec 0\\
    1\cdot \vec t &\equiv \vec t &
    \alpha\cdot (\beta\cdot \vec t) &\equiv\alpha\beta\cdot \vec t \\
    \alpha\cdot \vec t+\beta\cdot \vec t &\equiv (\alpha+\beta)\cdot \vec t &
    \alpha\cdot (\vec t_1+\vec t_2) &\equiv\alpha\cdot \vec t_1+\alpha\cdot \vec t_2
  \end{array}
\]

This also implies that the summation symbol $\sum$ can be used unambiguously. A superposition $\vec{t}$ is in \emph{canonical form} if it is either $\vec{0}$ or $\vec{t} = \sum_{i=1}^n \alpha_i\cdot t_i$, where $\forall i\neq j, t_i\neq t_j$, and $\forall i,\ \alpha_i \neq 0$. The canonical form of a superposition $\vec t$ is unique, modulo associativity and commutativity. We denote the set of canonical forms as $\CF$.
We also define the syntactic sugar given in Figure~\ref{fig:syntaxsugar}.
%
\begin{figure}[t]
  \hrulefill
  \begin{align*}
    \tif{\sum_{i=1}^n \alpha_i \cdot s_i}{\vec r_1}{\vec r_2} &\triangleq\sum\limits_{i=1}^n\alpha_i\cdot\tif{s_i}{\vec r_1}{\vec r_2}\\
    s\ \left(\sum_{i=1}^n \alpha_i \cdot t_i\right)&\triangleq \sum\limits_{i=1}^n\alpha_i\cdot s\ t_j \\
    \pair{\sum_{i=1}^n \alpha_i \cdot s_i}{\sum_{j=1}^m \beta_j \cdot t_j} &\triangleq \sum\limits_{i=1}^n\sum\limits_{j=1}^m\alpha_i\beta_j\cdot\pair{s_i}{t_j}\\
    \tlet{x}{y}{\sum_{i=1}^n \alpha_i \cdot s_i}{\vec r} &\triangleq\sum\limits_{i=1}^n\alpha_i\cdot\big(\tlet{x}{y}{s_i}{\vec r}\big)
  \end{align*}
  \hrulefill
  \caption{Syntactic sugar on $\punq$ syntax.}
  \label{fig:syntaxsugar}
\end{figure}

\subsection{Operational Semantics}
The semantics of $\punq$ programs is defined by the rewrite relation
$\rightsquigarrow \ \subseteq \SUP_c \times \SUP_c$, given in
Figure~\ref{fig:sem}. We denote its reflexive and transitive closure as
$\rightsquigarrow^*$.
\begin{figure}[t]
  \hrulefill
  \[
    \begin{prooftree}
      \hypo{\phantom{t}}
      \infer1[(If$_0$)]{\tif{\ket 0}{\vec s}{\vec t} \rightsquigarrow\vec s}
    \end{prooftree}
    \qquad
    \begin{prooftree}
      \hypo{\phantom{t}}
      \infer1[(If$_1$)]{\tif{\ket 1}{\vec s}{\vec t} \rightsquigarrow\vec t }
    \end{prooftree}
     \qquad
    \begin{prooftree}
      \hypo{}
      \infer1[(Abs)]{ (\lambda x.\vec t)\ v \rightsquigarrow {\vec t[v/x]}  }
    \end{prooftree}
    \]
    \\[0.2cm]
    \[
    \begin{prooftree}
      \hypo{}
      \infer1[(Let)]{\tlet xy{\pair vw}{\vec t} \rightsquigarrow {\vec t[v/x,w/y]}}
    \end{prooftree}
    \qquad
    \begin{prooftree}
      \hypo{t\rightsquigarrow\vec s}
      \infer1[(If$_{+}$)]{\tif t{\vec r_1}{\vec r_2}\rightsquigarrow\tif{\vec s}{\vec r_1}{\vec r_2}}
    \end{prooftree}
  \]
  \\[0.2cm]
  \[
    \begin{prooftree}
      \hypo{t\rightsquigarrow\vec s}
      \infer1[(App)]{r\ t\rightsquigarrow r\ \vec s}
    \end{prooftree}
    \qquad
    \begin{prooftree}
      \hypo{t\rightsquigarrow \vec s}
      \infer1[(App$_\V$)]{t\ v\rightsquigarrow\vec s\ v}
    \end{prooftree}
    \qquad
    \begin{prooftree}
      \hypo{t\rightsquigarrow\vec s}
      \infer1[(Pair)]{\pair tr\rightsquigarrow\pair{\vec s}r}
    \end{prooftree}
    \qquad
    \begin{prooftree}
      \hypo{t\rightsquigarrow \vec s}
      \infer1[(Pair$_\V$)]{\pair vt\rightsquigarrow\pair v{\vec s}}
    \end{prooftree}
  \]
  \\[0.2cm]
  \[
    \begin{prooftree}
      \hypo{t\rightsquigarrow\vec s}
      \infer1[(Let$_+$)]{\tlet xyt{\vec r}\rightsquigarrow\tlet xy{\vec s}{\vec r}}
    \end{prooftree}
   \]
  \\[0.2cm]   
   \[
    \begin{prooftree}
      \hypo{\sum_{i \in I}\alpha_i\cdot t_i + \sum_{j\in J}\beta_j\cdot v_j \in \CF}
      \hypo{\forall {i\in I},\, t_i\rightsquigarrow s_i}
      \infer2[(Sup)]{\sum_{i \in I}\alpha_i\cdot t_i + \sum_{j\in J}\beta_j\cdot v_j\rightsquigarrow \sum_{i\in I}\alpha_i\cdot s_i+\sum_{j \in J}\beta_j\cdot v_j}
    \end{prooftree}
    \qquad
    \begin{prooftree}
      \hypo{\vec t \equiv \vec t_1}
      \hypo{\vec t_1 \rightsquigarrow \vec s_1}
      \hypo{\vec s_1 \equiv \vec s}
      \infer3[(Equ)]{ \vec t \rightsquigarrow \vec s}
    \end{prooftree}
  \]
  \hrulefill
  \caption{Semantics of $\punq$ programs.}
  \label{fig:sem}
\end{figure}

It is important to emphasize that the syntactic sugar defined in Figure~\ref{fig:syntaxsugar} is used in
the reduction rules (If$_+$), (App), (App$_\V$), (Pair), (Pair$_\V$), and
(Let$_+$) to simplify notations. For instance, the reduction of rule (If$_+$)
can be written as 
\[
  \begin{prooftree}
    \hypo{t\rightsquigarrow\sum_i \alpha_i \cdot {s_i}}
    \infer1[(If$_{+}$)]{\tif t{\vec r_1}{\vec r_2}\rightsquigarrow\sum_i \alpha_i \cdot \tif{s_i}{\vec r_1}{\vec r_2}}
  \end{prooftree}
\]
Therefore, reductions through $\rightsquigarrow$ do not generate terms
or superpositions that are not syntactically valid.  Another crucial point to
note is that the relation $\rightsquigarrow$ is constrained to canonical forms
in rule (Sup) to prevent the reduction of a superposition in the form $v + 0
\cdot t$, where $t\in \T,\ v \in \V$, as it is true that $v + 0 \cdot t \equiv
v$ and, as a result, $v + 0 \cdot t$ should not reduce.

The semantics of $\punq$ is call-by-value.  Since there is an established
strategy, the confluence of this calculus is trivial.  The normal forms are
closed values and they are unique modulo the equivalence relation $\equiv$.

\section{A type system for unitarity and polytime normalization}\label{s:type}
In this section, we introduce a type system ensuring that typable closed
superpositions:
\begin{itemize}
  \item encode unitary transformations on the type of circuits over qubits (i.e., linear functions from qubits to qubits) (Section~\ref{s:unitarity});  
  \item normalize in time polynomial in their size, exhibiting only polynomial growth on the size of the superposition (Section~\ref{s:ptstrongnormalization}).
\end{itemize}
The typing discipline is created by mixing the unitary-ensuring
type system of Lambda-$\mathcal S_1$~\cite{DCM22} together with the polytime
strong normalization properties of the \dlal{}~\cite{BT09} type system.

\subsection{Types, Judgments, and Environments}
The set $\mathbb{T}$ of types in \punq{} is generated by the following grammars:
\begin{align*}
\text{(\punq{}\ types)} && \mathbb{T} \ni A,B,C,\ldots  &:= X \mid A \multimap A \mid A \Rightarrow A \mid A\times A\mid Q \mid \S A \mid\forall X.A  \\
\text{(Ground types)} && Q,R,S,\ldots &:= \B \mid \sharp Q\mid \S Q \mid Q\times Q
\end{align*}

We use $A, B, C,$ and so on for \punq{} types, and $Q,R,S,\dots$ for ground types. Types include type variables $X$, a
basic type $\B$ for bits, a linear arrow $\multimap$, an intuitionistic arrow
$\Rightarrow$, a type construct $\times$ for pairs corresponding to the tensor
product of linear logic (we do not use the tensor notation of linear logic to
avoid confusion as pairs only correspond to separable states), a modality
$\sharp$ for superpositions, a modality $\S$ as a marker for possible
duplication, and polymorphism. The set of closed types is denoted as
$\mathbb{T}_c$. Ground types $Q$ represent all types that can be inhabited by tuples of qubits. Intuitively, objects of type $\sharp Q$ are unitary
superpositions of elements of type $Q$ and hence cannot be cloned (i.e.,
duplicated). For example, $\sharp\B$ is the type of a unitary superposition of
bits, i.e., qubits. Qubits correspond to values whose canonical forms are of
the shape $\alpha \cdot \ket{0}\, + \beta \cdot \ket{1}\,$, with $\alpha, \beta \in
\Ct$ and $|\alpha|^2 + |\beta|^2 = 1$.  We write $A^n$, with $n \geq 1$, as a
shorthand for $A \times \ldots \times A$, $n$ times. The type $\sharp(\B^n)
\multimap \sharp(\B^n)$ corresponds to quantum circuits over $n$ qubits.
Finally, a type of the form $\sharp \B \Rightarrow A$, for any type $A$, will
be disallowed by the typing discipline as $\Rightarrow$ is not linear and hence
could involve duplication or deletion of the input.

A \emph{typing environment} $\Gamma$ is a mapping from variables
to closed types in $\mathbb{T}_c$ and is sometimes written as $x_1: A_1,
\ldots, x_n: A_n$. The notation $\Gamma, \Delta$ represents the disjoint union
of the typing environments $\Gamma$ and $\Delta$. 
We write as $FV(\Gamma)$ the set of free type variables appearing in $\Gamma$.

\emph{Typing judgments} are of the form $\Gamma; \Delta \vdash \vec{t} : A$,
where $\Gamma$ and $\Delta$ are disjoint typing environments, $\vec{t}$ is
a superposition, and $A$ is a type. Here, $\Gamma$ is referred to as
the \emph{exponential context}, and $\Delta$ is the \emph{linear context}.

\subsection{Orthogonality}\label{ss:orthogonality}
We define the inner product $\langle - | - \rangle : \V_c \times \V_c \to
\Ct$ over closed values as
\[ \Big\langle \sum_{i=1}^n \alpha_i\cdot v_i\ \big|\ \sum_{j=1}^m \beta_j\cdot w_j\Big\rangle\triangleq \sum
_{i=1}^n \sum_{j=1}^m \overline{\alpha_i} \beta_j \delta_{v_i,w_j},\qquad \qquad \langle \vec{v}| \vec{0}\rangle=\langle \vec{0}|\vec{v}\rangle\triangleq 0,\]
where $\delta_{x,y}$ is the Kronecker delta that is equal to $1$ if $x=y$ and
$0$ otherwise.  Notice that inner product is obviously preserved by the
equivalence relation $\equiv$ on the vector space of values.  Hence $\V_c$ is a Hilbert space as $\ell^2(\V_c)=\V_c$.

Type checking a superposition will require that we are able to test orthogonality between terms. This is straightforward to do for values in normal form, but we may also face the situation where a term correctly represents a superposition, and yet it is not fully reduced, such as the term $\frac{1}{\sqrt{2}}\cdot (\lambda x. x) \ket{0}+\frac{1}{\sqrt{2}}\cdot \ket{1}$. Therefore, we use the following definition of orthogonality between terms.
\begin{definition}[Orthogonality]
  The \emph{orthogonality relation}  $\perp\subseteq\SUP_c \times \SUP_c$ is defined as $\vec t\perp\vec s$ if and only if
  $\vec t\rightsquigarrow^*\vec v$, $\vec s\rightsquigarrow^*\vec w$ and $\langle\vec v|\vec w\rangle=0$.
\end{definition}

  For a given typing environment $\Gamma$, let $\sigma_{\Gamma}$ denote a
  type-preserving substitution of variables in $\Gamma$ by closed basis values. 
  The orthogonality relation can be extended to any two superpositions $\vec
  t,\ \vec s \in \SUP$ such that $\Gamma; \Delta \vdash \vec t : A$ and
  $\Gamma; \Delta \vdash \vec s : A$ by defining $(\vec{t} \perp_A^{\Gamma,\Delta}
  \vec{s})$ iff $\ \forall \sigma_{\Gamma \cup \Delta},\  \vec t \sigma_{\Gamma
  \cup \Delta} \perp_{A}  \vec s\sigma_{\Gamma \cup \Delta}$. 
\subsection{Type System}
We introduce a \emph{bang function} on types, which the type system utilizes to
remove the $\sharp$ modalities when they correspond to a non-linear use.

\begin{definition}[Bang function]\label{def:bang}
  The function ${!}: \mathbb{T} \to \mathbb{T}$ is an endomorphism on types defined as
  \begin{align*}
    {!}(X) &\triangleq X &
    {!}(\B) &\triangleq \B \\
    {!}(A\multimap B)& \triangleq A \multimap B & 
    {!}(A\Rightarrow B) &\triangleq A \Rightarrow B \\
    {!}(A \times B) &\triangleq {!}(A) \times {!}(B) &
    {!}(\sharp Q)& \triangleq {!}(Q) \\
    {!}(\S A)& \triangleq \S {!}(A) &
    {!}(\forall X.A) &\triangleq \forall X.{!}(A)
  \end{align*}
\end{definition}
This function is reminiscent of the bang modality in linear logic~\cite{G87}
because it transforms a non-clonable type (non-duplicable type, e.g.,
superposition) into a clonable type (duplicable type). Therefore, it
corresponds to withdrawing $\sharp$ modalities. For instance, $! (\sharp \B) =
\B$ represents the type of bits and is clonable. Similarly, $!(\sharp \B
\multimap \sharp \B) = \sharp \B \multimap \sharp \B$ is the type of unitary
maps on qubits (see Theorem~\ref{thm:unitarity}), which is clonable by default.

The \emph{subtyping} relation $\leq \subseteq \mathbb{T} \times \mathbb{T}$ is
defined in Figure~\ref{fig:subtype}. The intuition behind this relation is as
follows: if a type $Q$ is considered as the base of a vector space for its
values, then $\sharp Q$ results in the intersection of the span of $A$ with the
unitary sphere, that is, complex linear combinations of objects of type $A$ with unit norm.  Therefore, $\leq$ corresponds to set
inclusion, and it holds that $Q \leq \sharp Q$ and $\sharp Q= \sharp \sharp
Q$. Additionally, it can be shown that $\sharp Q \times \sharp R\leq \sharp
({!}(Q) \times {!}(R))$, implying the desirable property that the vector space
of separable qubits $\sharp \B \times \sharp \B$ is included in the vector
space of $2$ qubits $\sharp (\B \times \B)$.

\begin{figure}[t]
  \hrulefill
  \[    
    \begin{prooftree}
      \hypo{A'\leq A}
      \hypo{B\leq B'}
      \infer2[]{\vphantom{\sharp}A\multimap B \leq A'\multimap B'}
    \end{prooftree}
    \qquad
    \begin{prooftree}
      \hypo{A'\leq A}
      \hypo{B\leq B'}
      \infer2[]{\vphantom{\sharp}A\Rightarrow B \leq A'\Rightarrow B'}
    \end{prooftree}
    \qquad
    \begin{prooftree}
      \hypo{A\leq A'}
      \hypo{B\leq B'}
      \infer2[]{\vphantom{\sharp}A\times B \leq A'\times B'}
    \end{prooftree}
    \qquad
    \begin{prooftree}
      \hypo{A\leq B}
      \infer1[]{\S A\leq \S B}
    \end{prooftree}
    \qquad 
    \begin{prooftree}
      \hypo{A\leq B}
      \infer1[]{\forall X.A\leq \forall X. B}
    \end{prooftree}
  \]
  
    \[
    \begin{prooftree}
      \hypo{\vphantom{Q}}
      \infer1[]{\vphantom{\sharp}A\leq A}
    \end{prooftree}
    \qquad
    \begin{prooftree}
      \hypo{\vphantom{Q}}
      \infer1[]{Q\leq \sharp Q}
    \end{prooftree}
    \qquad
    \begin{prooftree}
      \hypo{\vphantom{A\leq B}}
      \infer1[]{\sharp\sharp Q\leq \sharp Q}
    \end{prooftree}
    \qquad
    \begin{prooftree}
      \hypo{\vphantom{A\leq B}}
      \infer1[]{Q\leq \sharp {!}(Q)}
    \end{prooftree}
    \qquad 
    \begin{prooftree}
      \hypo{\vphantom{A\leq B}}
      \infer1[]{\sharp \S Q \leq \S\sharp Q}
    \end{prooftree}
    \qquad
    \begin{prooftree}
      \hypo{A\leq B}
      \hypo{B\leq C}
      \infer2[]{\vphantom{\sharp}A\leq C}
    \end{prooftree}
  \]
  
  \hrulefill
  \caption{Subtyping relation.}
  \label{fig:subtype}
\end{figure}

Since the $!$ function will only be used in the double arrow in rule
introduction $(\Rightarrow_i)$ on the left of that arrow, the double arrow
carries all the information about whether the $!$ function was applied.
Therefore, $!$ is not treated as a modality and appears implicitly in type
substitution.
Given a type $C$ and a variable $X$, let $[C/X]$ represent the \emph{type
substitution} $\sigma^+$. It is inductively defined on types (up to
$\alpha$-renaming) as follows:
\begin{align*}
  \sigma^k(Y)&\triangleq Y,\ \text{with }Y \neq X,
  & \sigma^k(A\Rightarrow B)&\triangleq \sigma^-(A) \Rightarrow \sigma^+(B),\\
  \sigma^+(X)&\triangleq C, 
  & \sigma^k(\sharp Q)&\triangleq \sharp \sigma^k(Q), \\
  \sigma^-(X)&\triangleq \ !(C), 
  & \sigma^k(\S A)&\triangleq \S \sigma^k(A),\\
  \sigma^k(\B)&\triangleq  \B, 
  & \sigma^k(A\times B)&\triangleq \sigma^k(A)\times \sigma^k(B),\\
  \sigma^k(A\multimap B) &\triangleq \sigma^k(A)\multimap \sigma^k(B),
  & \sigma^k(\forall Y.A)& \triangleq \forall Y.\sigma^k(A), 
\end{align*}
where $k\in\{+,-\}$  and where $!$ is the bang function from Definition~\ref{def:bang}. This ensures that non-clonability is preserved on inputs to non-linear applications, e.g., $(X\Rightarrow X)[\sharp \B/X]\ =\ \sigma^-(X)\Rightarrow \sigma^+(X)\ =\ !(\sharp \B)\Rightarrow \sharp \B\ =\ \B \Rightarrow \sharp \B$.

The \emph{typing rules} are provided in Figure~\ref{fig:type}.
$\punq$ is the set of typable closed superpositions in $\SUP_c$.

Several key points are worth highlighting:
\begin{inparaenum}[(i)]
  \item The $!$ function is employed in rule ($\Rightarrow_i$) to ensure the
    no-cloning property.
  \item In rule ($\Rightarrow_e$), the notation $[z:C]$ indicates that the
    variable $z$ is optional: it can either appear in both the hypothesis and
    the conclusion of the rule or not at all. The variable $z$ is then passed
    to the exponential context in the conclusion of the rule. Consequently,
    qubits are treated linearly following the $\dlal$ type discipline, and they
    cannot be cloned by side effect. For instance, if $z$ is of type $\sharp
    \B$, then $\lambda z.t\ s$ is of type $\B \Rightarrow B$, as per rule
    ($\Rightarrow_i$).
  \item Rules  $(\mathsf{if}_\sharp)$ and $(\sharp_i)$ are the only two rules
    that make use of the orthogonality predicate $\perp_A \subseteq
    \mathbb{T}_c \times \mathbb{T}_c$.
  \item Rule  $(\sharp_i)$ disallows superposing arrow types.
  \item Rule $(\forall_e)$ utilizes the type substitution $[B/X]$.
\end{inparaenum}

\begin{figure}[t]
  \hrulefill
  \[
    \begin{prooftree}
      \hypo{\Gamma;\Delta \vdash \vec{t}:A}
      \infer1[($\mathsf{W}$)]{\Gamma,\Gamma';\Delta\vdash \vec{t}:A}
    \end{prooftree}
    \qquad\quad
    \begin{prooftree}
      \hypo{\Gamma,x:B,y:B;\Delta\vdash \vec{t}:A}
      \infer1[($\mathsf{C}$)]{\Gamma,x:B;\Delta \vdash \vec{t}[x/y]:A}
    \end{prooftree}
    \qquad \quad
    \begin{prooftree}
      \hypo{\Gamma;\Delta \vdash\vec t:A}
      \hypo{\vec t\equiv\vec s}
      \infer2[($\equiv$)]{\Gamma;\Delta \vdash\vec s:A}
    \end{prooftree}
  \]
  
  \[
    \begin{prooftree}
      \hypo{\Gamma;\Delta \vdash\vec t:A}
      \hypo{A\leq B}
      \infer2[($\leq$)]{\Gamma;\Delta \vdash \vec t:B}
    \end{prooftree}
    \qquad\quad
    \begin{prooftree}
      \hypo{\phantom{A\leq B}}
      \infer1[($\mathsf{Ax}$)]{;x: A\vdash x: A}
    \end{prooftree}
    \qquad\quad
    \begin{prooftree}
      \hypo{\phantom{A\leq B}}
      \infer1[($0$)]{;\vdash\ket 0:\B}
    \end{prooftree}
    \qquad\quad
    \begin{prooftree}
      \hypo{\phantom{A\leq B}}
      \infer1[($1$)]{;\vdash\ket 1:\B}
    \end{prooftree}
  \]

  \[
    \begin{prooftree}
      \hypo{\Gamma;\Delta \vdash t:\sharp\B}
      \hypo{\Gamma';\Delta';\vdash \vec s_1 :Q}
      \hypo{\Gamma';\Delta'\vdash \vec s_2 :Q}
      \hypo{\vec s_1 \perp^{\Gamma';\Delta'} \vec s_2}
      \infer4[($\mathsf{if}_\sharp$)]{\Gamma,\Gamma';\Delta,\Delta'\vdash\tif{t}{\vec s_1}{\vec s_2}:\sharp Q}
    \end{prooftree}
  \]

  \[\begin{prooftree}
      \hypo{;\Gamma,\Delta\vdash\vec t:A}
      \infer1[($\S_i$)]{\Gamma;\S\Delta\vdash\vec t:\S A}
    \end{prooftree}
    \qquad \quad
    \begin{prooftree}
      \hypo{\Gamma;\Delta \vdash t:\B}
      \hypo{\Gamma';\Delta' \vdash\vec{s_1}:A}
      \hypo{\Gamma';\Delta' \vdash\vec{s_2}:A}
      \infer3[($\mathsf{if}$)]{\Gamma,\Gamma';\Delta \vdash\tif t{\vec s_1}{\vec s_2}:A}
    \end{prooftree}
  \]

  \[
    \begin{prooftree}
      \hypo{\Gamma;\Delta \vdash s:\S B}
      \hypo{\Gamma';\Delta',x:\S B\vdash t:A}
      \infer2[($\S_e$)]{\Gamma,\Gamma';\Delta,\Delta'\vdash t[s/x]:A}
    \end{prooftree}
\qquad\quad
    \begin{prooftree}
      \hypo{\Gamma;\Delta\vdash\vec t:A}
      \hypo{X \notin FV(\Gamma,\Delta)}
      \infer2[($\forall_i$)]{\Gamma;\Delta \vdash\vec t:\forall X.A}
    \end{prooftree}
  \]

    \[
    \begin{prooftree}
      \hypo{ \forall i,\ \Gamma;\Delta\vdash \vec t_i : Q}
      \hypo{\forall j\neq k, \ \vec t_j\perp^{\Gamma;\Delta}\vec t_k }
      \hypo{\sum_{i=1}^n|\alpha_i|^2=1}
      \infer3[($\sharp_i$)]{\Gamma;\Delta\vdash\sum_{i=1}^n \alpha_i\cdot\vec t_i:\sharp Q}
    \end{prooftree}
    \qquad\quad
    \begin{prooftree}
      \hypo{\Gamma;\Delta \vdash\vec t:\forall X.A}
      \infer1[($\forall_e$)]{\Gamma;\Delta \vdash\vec t:A[B/X]}
    \end{prooftree}
  \]

  \[
    \begin{prooftree}
      \hypo{\Gamma;\Delta,x:A\vdash \vec{t}:B}
      \infer1[($\multimap_i$)]{\Gamma;\Delta \vdash \lambda x.\vec{t}:A\multimap B}
    \end{prooftree}
    \qquad\quad
    \begin{prooftree}
      \hypo{\Gamma;\Delta \vdash t: A\multimap B}
      \hypo{\Gamma' ;\Delta' \vdash s:A}
      \infer2[($\multimap_{e}$)]{\Gamma,\Gamma';\Delta,\Delta' \vdash t \ s:B}
    \end{prooftree} 
  \]
  
  \[
    \begin{prooftree}
      \hypo{\Gamma,x:A;\Delta \vdash \vec t:B}
      \infer1[($\Rightarrow_{i}$)]{\Gamma;\Delta \vdash{\lambda x.\vec t}: {!}(A)\Rightarrow B}  
    \end{prooftree} 
    \qquad\quad
    \begin{prooftree}
      \hypo{\Gamma;\Delta \vdash {t}:{A\Rightarrow B}}
      \hypo{;[z:C]\vdash  s:A}
      \infer2[($\Rightarrow_{e}$)]{\Gamma,[z:C]\ ;\Delta \vdash {t \ s}:B}
    \end{prooftree}
  \]
  
  \[
    \begin{prooftree}[separation=3mm]
      \hypo{\Gamma;\Delta \vdash t:A}
      \hypo{\Gamma';\Delta' \vdash s:B}
      \infer2[($\times_i$)]{\Gamma,\Gamma';\Delta,\Delta' \vdash\pair{t}{s}:A\times B}
    \end{prooftree}
    \qquad \quad
    \begin{prooftree}[separation=3mm]
      \hypo{\Gamma;\Delta\vdash t:A\times B}
      \hypo{\Gamma';\Delta',x:A,y:B \vdash\vec s:C}
      \infer2[($\times_{e}$)]{\Gamma,\Gamma';\Delta,\Delta'\vdash\tlet xyt{\vec s}:C}
    \end{prooftree}
  \]

  \[
    \begin{prooftree}
      \hypo{\Gamma ;\Delta \vdash  t:\sharp(Q\times R)}
      \hypo{\Gamma'; \Delta' ,x:\sharp Q,y:\sharp R\vdash\vec s:S}
      \infer2[($\times_{e\sharp}$)]{\Gamma,\Gamma';\Delta,\Delta' \vdash\tlet xy{ t}{\vec s}:\sharp S}
    \end{prooftree}
  \]

  \hrulefill
  \caption{Typing rules for $\punq$ programs.}
  \label{fig:type}
\end{figure}

 \punq{} enjoys the following standards properties.

\begin{restatable}[Subject reduction]{theorem}{thmsr}
  If $;\vdash\vec t:A$ and $\vec t\rightsquigarrow\vec r$, then $;\vdash\vec r:A$. \label{thm:sr}
\end{restatable}
\begin{proof}
  By induction on the relation $\rightsquigarrow$.  The full proof, including
  the needed substitution lemma, is given in Appendix~\ref{app:SR}.
\end{proof}

\begin{restatable}[Progress]{theorem}{thmprogress}
  If $;\vdash\vec t:A$, then either $\vec t \rightsquigarrow\vec r$ for some $\vec r$ or $\vec{t}\in \V$. \label{thm:progress}
\end{restatable}
\begin{proof}
By induction on the structure of $\vec{t}$. Full proof given in Appendix~\ref{app:SR}.
\end{proof}

\subsection{Intuitions}
Most rules in Figure~\ref{fig:type} are reminiscent of the typing discipline
of Dual Light Affine Logic~\cite{BT09}, extended with multiplicative pairs
(cf.~Figure~\ref{fig:dlal}). The weakening rule $(\mathsf{W})$ and contraction
rule $(\mathsf{C})$ are limited to exponential contexts since discarding or
duplicating a qubit would break unitarity. In our setting,
the typing rules are linear in the linear context.

Rule ($\equiv$) guarantees a type preservation property on the vector space of
superpositions. This rule is crucial for subject reduction. Indeed, terms
of a typable superposition must have a norm of $1$ (by rule $(\sharp_i)$),
preventing a superposition of the form $\frac{1}{2} \cdot t + \frac{1}{2} \cdot
t$ from typing, even though it may be obtained as a reduct of the operational
semantics. An alternative approach could have been restricting the relation
$\rightsquigarrow$ to pairs of canonical forms in $\nicefrac{\SUP_c}{\equiv} \times
\nicefrac{\SUP_c}{\equiv}$ and evaluating arbitrary superpositions only using their
canonical representative.

Rule $(\leq)$ is the subtyping rule, allowing the system to handle entangled
datatypes.

The typing rules ($\mathsf{if}$) and ($\mathsf{if}_\sharp$) are additive rules
for conditionals on type $A$ controlled by classical data $\B$ and quantum data
$\sharp \B$, respectively. In this latter case, the conditional evaluates to a
superposition of conditionals, as per rule ($\mathrm{If}_+$) in
Figure~\ref{fig:sem}. Hence, the result has type $\sharp A$.

The typing rules for linear and intuitionistic arrows resemble those
in~\cite{BP96}, with the additional requirement that the introduction and
elimination of the intuitionistic arrow, rules ($\Rightarrow_i$) and
($\Rightarrow_e$), can only be performed on ``banged'' data. The bang function
ensures that the resulting type is copyable (or clonable) by turning a qubit
into a bit, $!(\sharp \B)=\B$, propagating over pairs, and not changing
applications. For instance, $!(\sharp \B \multimap \sharp\B)=\sharp \B
\multimap \sharp\B$, as it is safe to copy a quantum gate. Therefore, the
intuitionistic arrow explicitly forbids values as inputs, and types of the form
$\sharp A \Rightarrow B$ are uninhabited.

The pair constructor and separable pair destructor are typed using the
multiplicative rules ($\times_i$) and ($\times_e$), following the encoding
of~\cite{BT09}. Rule ($\times_{e\sharp}$) is the rule for the quantum pair
destructor, allowing the handling of a superposition of pairs, hence a possibly
entangled state. In this setting, the pair can be viewed as the tensor product
used in quantum computing. Variables $x$ and $y$ are typed by $\sharp Q$ and
$\sharp R$ to preserve unitarity. Note that this typing rule can be extended to
superpositions as follows:
\[
  \begin{prooftree}
    \hypo{\Gamma_1 ;\Delta_1 \vdash \sum_i \alpha_i \cdot t_i :\sharp(Q\times R)}
    \hypo{\Gamma_2; \Delta_2 ,x:\sharp Q,y:\sharp R\vdash\vec s:S}
    \infer2[($\vec \times_{e\sharp}$)]{\Gamma_1,\Gamma_2;\Delta_1,\Delta_2 \vdash\sum_i \alpha_i \cdot \tlet xy{t_i}{\vec s}:\sharp S}
  \end{prooftree}
\]
Rule ($\vec \times_{e\sharp}$) can be simulated using rules
($\times_{e\sharp}$) and ($\sharp_i$) but has the advantage of providing a more
flexible typing discipline and simplifying the orthogonality requirements of
rule ($\sharp_i$). Indeed, orthogonality is tested on subterms $t_i$ rather
than on terms $\tlet xy{t_i}{\vec s}$.

$\dlal$ characterizes the complexity class FPTIME~\cite{BT09} by imposing
constraints on the use of the contraction rule in proofs and introducing a new
modality $\S$ in rules ($\S_i)$ and ($\S_e)$ to address some of these
limitations. Broadly speaking, the modality $\S$ serves as a marker for objects
resulting from iterations that cannot be further iterated. Using this construct
limits the number of iterations possible in proofs and, consequently, the
complexity of the system.  Specifically, $\dlal$ enforces a design principle
known as \textit{stratification} by restricting the introduction and
elimination of the modality $\S$.  This stratification is not sufficient on its
own to characterize polynomial time (it provides a characterization of
elementary time~\cite{G94}) and this is the reason why rule $(\Rightarrow_e)$
is restricted to at most one open variable.

Rule ($\sharp_i$) is the introduction rule for superpositions. A superposition
of objects can be created under the conditions that they share the same type,
are pairwise orthogonal according to the predicate $\perp_A$, that the norm
of the generated superposition $\sum_{i=1}^n|\alpha_i|^2$ equals $1$, and that the type $A$ is not an arrow type. Implicit in this rule is the fact that the construction of superpositions is limited to complex amplitudes $\alpha_i$ that are polytime approximable, a set which we denote by $\Ct$ throughout the paper. When considering quantum computations restricted to some polytime complexity class, it is standard to include such restrictions since, for instance, intractable problems could be encoded into transition amplitudes~\cite{BV97,ADH97}.

Our notion of orthogonality is that of~\cite{DCM22}, where the
superposition of functions is disallowed. It is treated more realistically
than in~\cite{DCGMV19}, where the superposition of functions is allowed, but it is generally not a
function.

\section{Examples}\label{ss:ex}

We may now consider some examples of typing judgements in \punq{} which correspond to different tools in quantum computation. Informally, a term $t$ is said to represent a unitary operator if, when applied to another term encoding some superposition, the application reduces to what would be the \punq{} representation of the output of the operator. A precise definition of \emph{representing a quantum gate or a unitary operator} is given in Section~\ref{ss:isometries}.

\subsection{Superpositions}

To demonstrate the typing of a superposition we will consider the example of the $X$-basis states $\ket{\pm}\triangleq\frac{1}{\sqrt{2}}\cdot\ket{0}\pm \frac{1}{\sqrt{2}}\cdot \ket{1}\,$, also called the \emph{plus} and \emph{minus} states, which correspond to the unit-norm eigenstates of the Pauli-$X$ operator.

The conditions in the typing rule $(\sharp_i)$ correspond to checking that the basis terms have the correct type (i.e. since we are interested in typing a qubit state, these terms should be either booleans or qubits), that they are mutually orthogonal, and that the resulting state has unit norm.
\begin{example}[X-basis states] The one-qubit states
  $\ket{\pm}\,$ have type $\sharp\B$.
  \[
    \pi_\pm \triangleq \qquad
    \scalebox{1}{
      \begin{prooftree}
	\infer0[$(0)$]{;\vdash \ket{0}:\B}
	\infer0[$(1)$]{;\vdash \ket{1}:\B}
	\hypo{\ket{0}\perp\ket{1}}
	\hypo{\big|\frac{1}{\sqrt{2}}\big|^2 + \big|\pm\frac{1}{\sqrt{2}}\big|^2= 1}
	\infer4[($\sharp_i$)]{;\vdash \frac{1}{\sqrt{2}}\cdot\ket{0}\pm\frac{1}{\sqrt{2}}\cdot\ket{1}:\sharp\B}
    \end{prooftree}}
  \]\label{example:x-basis}
\end{example}

Notice that, if we use qubit terms as the basis to form the superposition (i.e. values of type $\sharp \B$) then their superposition will have type $\sharp \sharp \B$ by rule $(\sharp_i)$ but this corresponds to the typical qubit type $\sharp \B$ and can be enforced via subtyping since $\sharp\sharp Q\leq \sharp Q$.

\subsection{Single-Qubit Gates}

We will now turn to the construction of single-qubit gates which will form the basic building blocks for expressivity in the \punq{} language.

We will consider the specific examples of the Hadamard and $Z$ gates but we will see that their construction technique can be applied for any single-qubit gate. Remember that a necessary and sufficient condition of any unitary matrix $U:\mathbb{C}^{2^n}\rightarrow\mathbb{C}^{2^n}$ is that its columns form an orthogonal basis of $\mathbb{C}^{2^n}$, meaning that they are pairwise orthogonal and have unit norm. These are precisely the requirements of the quantum control typing rule ($\mathsf{if}_\sharp$).

\begin{example}[Hadamard]\label{ex:hadamard}
  Let $\pi_\pm$ represent the derivation trees for $\ket{+}$ and $\ket{-}$ in Example~\ref{example:x-basis}. The term $\mathsf{H} \triangleq \lambda x.
  \tif{x}{\ket{+}
  }{\ket{-}
  }$ representing the Hadamard gate has type $\sharp \B \multimap \sharp \B$:
  \[  \scalebox{1}{
      \begin{prooftree}
	\infer0[$(\mathsf{Ax})$]{;x:\sharp \B\vdash x:\sharp \B}
	\hypo{\pi_+}
	\ellipsis{}{;\vdash \ket{+}: \sharp  \B}
	\hypo{\pi_-}
	\ellipsis{}{;\vdash \ket{-}: \sharp  \B}
	\hypo{\ket{+}\perp\ket{-}}
	\infer4[$(\mathsf{if}_\sharp)$]{;x:\sharp \B\vdash \tif{x}{\ket{+}
	  }{\ket{-}
	}: \sharp \sharp \B}
	\infer1[$(\leq)$]{;x:\sharp \B\vdash \tif{x}{\ket{+}
	  }{\ket{-}
	}: \sharp \B}
	\infer1[$(\multimap_{i})$]{;\vdash \mathsf{H} : \sharp \B \multimap \sharp \B}
    \end{prooftree}}
  \]
\end{example}

\begin{example}[$Z$ gate]  \label{ex:z}
  
  The term $\mathsf{Z}\triangleq \lambda x. \tif{x}{\ket{0}}{-1\cdot \ket{1}}$ representing the Pauli $Z$ gate can also be typed:
  \[  \scalebox{1}{
      \begin{prooftree}
	\infer0[$(\mathsf{Ax})$]{;x:\sharp \B\vdash x:\sharp \B}
	\infer0[(0)]{;\vdash \ket{0}:\B}
	\infer1[$(\leq)$]{;\vdash \ket{0}: \sharp  \B}
	\infer0[(1)]{;\vdash \ket{1}:\B}
	\infer1[$(\sharp_i)$]{;\vdash -1\cdot \ket{1}: \sharp  \B}
	\hypo{ \ket{0}\perp\ket{1}}
	\infer4[$(\mathsf{if}_\sharp)$]{;x:\sharp \B\vdash \tif{x}{\ket{0}
	  }{-1\cdot \ket{1}
	}: \sharp \sharp \B}
	\infer1[$(\leq)$]{;x:\sharp \B\vdash \tif{x}{\ket{0}
	  }{-1\cdot\ket{1}
	}: \sharp \B}
	\infer1[$(\multimap_{i})$]{;\vdash \mathsf{Z} : \sharp \B \multimap \sharp \B}
    \end{prooftree}}
  \]
\end{example}

Using nested if statements, it is possible to construct larger gates. In fact, it is not hard to show that this suffices to construct \emph{all} unitary gates (see proof of Theorem~\ref{thm:isocomplete}). However, it is much more common to reason about complex operations not in their matrix representation but rather as compositions of small gates. Therefore, \punq{} allows for the typing of larger gates defined from smaller ones.

\subsection{Quantum Controlled Gates}

The quantum controlled-$NOT$ gate is a unitary operation on two qubits which performs a bitswitch on the second qubit depending on the state of the first, i.e. it is the linear operator that for each basis state $\ket{xy}$ returns $\ket{x(y\oplus x)}$, where $x,y\in\{0,1\}$. The following example shows that, once we have successfully typed a term representing a unitary gate, creating its quantum controlled version in \punq{} is straightforward.

\begin{example} [Controlled-NOT gate]\label{ex:controlled-not}
  Let $\mathsf{NOT}\triangleq \lambda x.\tif{x}{\ket{1}}{\ket{0}}$ which can be derived with type $\sharp \B\multimap \sharp\B$. We can derive the following term with type $\sharp (\B \times \B) \multimap \sharp( \B \times \B)$:
  \[
    \mathsf{CNOT} \triangleq \lambda z.\tlet{x}{y}{z}{\tif{x}{\pair{\ket{0}\,}{y}}{\pair{\ket{1}\,}{\mathsf{NOT}\ y}}}
  \]  
    \[  \scalebox{.95}{
      \begin{prooftree}
      \infer0[$(\mathsf{Ax})$]{;z:\sharp(\B^2)\vdash z:\sharp(\B^2)}
      \infer0[$(\mathsf{Ax})$]{;x:\sharp\B\vdash x:\sharp \B} 
      \hypo{}
      \ellipsis{}{;y:\sharp \B\vdash \mathsf{NOT}\ y:\sharp \B}
      \ellipsis{}{  ;y:\sharp \B\vdash\pair{\ket{0}\,}{y}\perp\pair{\ket{1}\,}{\mathsf{NOT}\ y}:\sharp(\B^2)}
      \infer2[$(\mathsf{if}_\sharp)$]{;x,y:\sharp\B\vdash \tif{x}{\pair{\ket{0}\,}{y}}{\pair{\ket{1}\,}{\mathsf{NOT}\ y}}:\sharp\sharp(\B^2)}
      \infer1[$(\leq)$]{;x,y:\sharp\B\vdash \tif{x}{\pair{\ket{0}\,}{y}}{\pair{\ket{1}\,}{\mathsf{NOT}\ y}}:\sharp(\B^2)}
	\infer2[$(\times_{e\sharp})$]{;z:\sharp (\B^2)\vdash \tlet{x}{y}{z}{\tif{x}{\pair{\ket{0}\,}{y}}{\pair{\ket{1}\,}{\mathsf{NOT}\ y}}}:\sharp (\B^2)}
	\infer1[$(\multimap_i)$]{;\vdash\mathsf{CNOT}:\sharp(\B^2)\multimap \sharp(\B^2)}
    \end{prooftree}}
  \]
\end{example}

Notice that the derivation does not require any particular properties of the $NOT$ gate, only that the term $\mathsf{NOT}$ representing it has type $\sharp \B\multimap \sharp \B$. The choice of using the values $\ket{0}$ and $\ket{1}$ to encode the value of $x$ in each branch of $\mathsf{CNOT}$ comes from the fact that the \punq{} typing system prohibits direct reuse of the variable $x$ since, in general, this can break unitarity. Hence, reusing a qubit variable inside any branch controlled on it is not permitted.

\subsection{Composition}

One may also use \punq{} to define larger operations operations from smaller ones by composing them. To demonstrate this, we will consider the example of quantum teleportation~\cite{BBGCJPW93} with delayed measurements, and show that by defining smaller terms it is easy to describe a term that represents the total operation. 

\begin{example}[Quantum teleportation]\label{ex:teleportation}

  \begin{figure}[t]
    \[
      \begin{quantikz}
	\lstick{$q$} & & &\ctrl{1}\gategroup[2,steps=2,style={draw=gray!40,fill=gray!20, inner
	xsep=1pt},background]{$\mathsf{Alice}$} & \gate{\text{H}} & \gategroup[3,steps=2,style={draw=gray!40,fill=gray!20, inner
	xsep=1pt},background]{$\mathsf{Bob}$} & \ctrl{2} & \\
	\lstick{$x$} &\gate{\text{H}} \gategroup[2,steps=2,style={draw=gray!40,fill=gray!20, inner
	  xsep=1pt},background,label style={label
	position=below,anchor=north,yshift=-0.2cm}]{$\mathsf{Bell}$}& \ctrl{1} &\targ{} & &\ctrl{1} & &\\
	\lstick{$y$} & & \targ{} & &  &\targ{} &\gate{\text{Z}} &
      \end{quantikz}
    \]
    \caption{Circuit for quantum teleportation with delayed measurements.}
    \label{fig:teleportation}
  \end{figure}
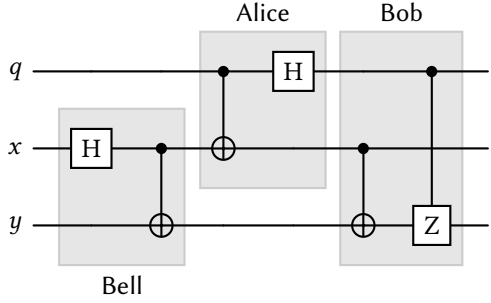

  Using terms $\mathsf{H}$ and $\mathsf{CNOT}$ as defined in Examples~\ref{ex:hadamard}~and~\ref{ex:controlled-not}, we may define a term representing the linear map $\ket{ab}\mapsto
    \frac{1}{\sqrt{2}}\left( \ket{0b}+(-1)^a\cdot\ket{1(1-b)}\right)$:
  \[
    \mathsf{Bell}\triangleq \lambda z.\tlet{x}{y}{z}{\mathsf{CNOT}\pair{\mathsf{H}\ x}{y}}:\sharp(\B^2)\multimap \sharp (\B^2)
  \]
In the specific case where $a=b=0$, it produces the Bell state
  $\frac{1}{\sqrt{2}}\big(\ket{00}\,+\ket{11}\,\big)$, which is used as a resource in
  teleportation. We can also encode Alice's and Bob's actions:
  \begin{align*}
    \mathsf{Alice} &\triangleq \lambda a.\tlet{q}{x}{\mathsf{CNOT}\ a}{\pair{\mathsf{H}\ q}{x}} \\
    \mathsf{Bob}&\triangleq \lambda b. \mathsf{let}\ \pair{w}{y}=b \ \mathsf{in}\\
    &\qquad\ \ \mathsf{let}\, \pair{q}{x}= w\ \mathsf{in}\ (\tif{q}{\pair{\ket{0}\,}{\mathsf{CNOT}\pair{x}{y}}}{\pair{\ket{1}\,}{\mathsf{CNOT}\pair{x}{\mathsf{Z}\ y}}})    
  \end{align*}
  with types $\sharp (\B^2)\multimap \sharp(\B^2)$ and $\sharp (\B^3)\multimap \sharp(\B^3)$, respectively. The final term can be obtained, representing the full operation as depicted in Figure~\ref{fig:teleportation}:
  \[
    \mathsf{telep}\triangleq \lambda z. \tlet{q}{w}{z}{(\tlet{x}{y}{\mathsf{Bell}\ w}{\mathsf{Bob}\ \pair{\mathsf{Alice}\ \pair{q}{x}}{y}})}:\sharp (\B^3)\multimap \sharp (\B^3).
  \]  
\end{example}

\subsection{Iteration}

So far we have only considered examples of constant-time operations. We will now turn to iteration in \punq{} which is inherited from the \dlal{} typing system.

Gate iteration stands at the intersection of the \dlal{} and the Lambda-$\mathcal{S}_1$ typing systems and is an example of the versatility that \punq{} can provide in safely mixing the two systems. As a toy example, consider the quantum random walk~\cite{K03,CCDFGS03} as an illustrating example for polytime iteration. This technique is known to provide a quantum speedup over its classical analog in different applications~\cite{A03} and an illustrative circuit is provided in Figure~\ref{fig:quantumwalk}.

In \punq{}, as in \dlal{}, we can implement Church numerals with the following type and syntax:
\begin{align*}
\N & \triangleq \forall X.(X\multimap X)\Rightarrow \S(X\multimap X)
&
\underline{n} & \triangleq \lambda f.\lambda x. f(f( \dots (f\ x)\dots)\quad (n\ \text{times}).
\end{align*}

For instance, we have that $\underline{2}\triangleq \lambda f.\lambda x. f(f\ x)$, which can be derived with type $\mathbb{N}$. \dlal{} is complete for polynomial time, and indeed, for any polynomial $P(n)=n ^{2^d}$, there is a DLAL term $\mathsf{poly}$  with type $\mathbb{N}\multimap \S^{2d}\, \N$ that simulates it~\cite[Proposition 11]{BT09}, which can also be derived in \punq{}. We will now use the Church encoding from \dlal{} to perform gate iteration in the example of a quantum random walk.

\begin{example}[Quantum random walk] \label{ex:random-walk}
 We define the R and L operators, for some set of $K$ nodes, correspond to ``taking a step'' in one direction or another in a given structure. For instance, if we consider a closed loop, we consider the unitary operators corresponding to a ``right'' step and a ``left'' step (see Figure~\ref{fig:quantumwalka}):
 \[\text{R}\ \ket{i}\triangleq \ket{i+1\Mod{K}}\quad\text{and}\quad \text{L}\ \ket{i}\triangleq \ket{i-1\Mod{K}}\,,\quad i\in\{0,\dots,K-1\},\]
where the basis state $\ket{i}\ $ corresponds to the particle in node $i$. For $k=\lceil \log(K-1)\rceil$, the operators R and L correspond to unitary operators in the space $\Ct^{2^k}\to \Ct^{2^k}$, and therefore by Theorem~\ref{thm:isocomplete} we show there exist corresponding \punq{} terms, say $\mathsf{R}$ and $\mathsf{L}$, of type $\sharp(\B^k)\multimap \sharp(\B^k)$ that simulate them. A single step of the quantum walk can then be done with the following term:
\[\mathsf{step}\triangleq \lambda z.\tlet{x}{y}{z}{(\tif{\mathsf{H}\ x}{\pair{\ket{0}\,}{\mathsf{L}\,y}}{\pair{\ket{1}\,}{\mathsf{R}\,y}})}:\sharp(\B^{k+1})\multimap \sharp(\B^{k+1})\]
where $x$ corresponds to the \emph{coin} qubit deciding the direction of the step, on which we continuously apply the Hadamard gate. Since these terms correspond to an endomorphism, we may use the Church encoding of numerals to perform an iteration of $\mathsf{step}$. For instance, this allows us to define an abstraction that takes in an integer input and iterates $\mathsf{step}$ precisely that number of times. Let $\mathsf{walk}\triangleq \lambda m.\lambda x. (m\ \mathsf{step})\ x$, which has type $\mathbb{N}\multimap \S \sharp(\B^{k+1})\multimap \S \sharp (\B^{k+1})$ with the following derivation, where $A=\sharp (\B^{k+1})$:
\[
\scalebox{.88}{
      \begin{prooftree}
      \hypo{}
      \ellipsis{}{;x: A, F:A\multimap A\vdash F\ x: A}
      \infer1[$(\S_i)$]{;x:\S A, F:\S (A\multimap A)\vdash F\ x:\S A}
      \infer0[$(\mathsf{Ax})$]{;m:\N\vdash m:\N}
      \infer1[$(\forall_e)$]{;m:\N \vdash m:(A\multimap A)\Rightarrow \S(A\multimap A)}
      \hypo{;\vdash \mathsf{step}:A\multimap A}
      \infer2[$(\Rightarrow_e)$]{;m:\N\vdash m\ \mathsf{step}:\S(A\multimap A)}
      \infer2[$(\mathsf{C})$]{;m:\N, x:\S A \vdash (m\ \mathsf{step})\ x: \S A}
      \infer1[$(\multimap_i)$]{;m:\N \vdash \lambda x.(m\ \mathsf{step})\ x:\S A\multimap \S A}
	\infer1[$(\multimap_i)$]{;\vdash \mathsf{walk}:\mathbb{N}\multimap \S A\multimap \S A}
    \end{prooftree}}
\]

For any $\underline{n}$ of type $\mathbb{N}$ we have that $\mathsf{walk}\ \underline{n}\rightsquigarrow^\ast \lambda x.(\mathsf{step}(\mathsf{step}\ \dots (\mathsf{step}\ x)\dots))$ where $\mathsf{step}$ is applied $n$ times. Similarly, using the term $\mathsf{poly}$ encoding a polynomial $P$, we can construct an abstraction taking in an integer $n$ that iterates $\mathsf{step}$ a number $P(n)$ of times:
\[\mathsf{pwalk}\triangleq \lambda n.\lambda x.((\mathsf{poly}\ n)\ \mathsf{step}) \ x: \N\multimap  \S^{2d+1} \sharp(\B^{k+1})\multimap \S^{2d+1} \sharp (\B^{k+1}). \]
\end{example}

\newcommand{\h}{1.5}
\begin{figure}[]
\centering
\begin{subfigure}[b]{0.47\textwidth}
\centering
\scalebox{1}{
\begin{tikzpicture}  
  [font=\scriptsize,scale=.8,auto=center,minimum size=1pt,main node/.style={circle,fill=gray!0,draw=black}] 
    
  \node[main node, style={fill=gray!65}] (a1) at ({-\h*cos(18)},{\h*sin(18)}) {0};  
  \node[main node,style={fill=gray!20}] (a2) at (0,\h)  {1}; 
  \node[main node] (a3) at ({\h*cos(18)},{\h*sin(18)})  {2};  
  \node[main node] (a4) at ({\h*cos(-54)},{\h*sin(-54)}) {3};  
  \node[main node,style={fill=gray!20}] (a5) at (-{\h*cos(-54)},{\h*sin(-54)})  {4};  
  \node[left=0 of a1] (t1) {$\ket{000}$}; 
  \node[right=0 of a2] (t2) {$\ket{001}$};
  \node[right=0 of a3] (t3) {$\ket{010}$}; 
  \node[right=0 of a4] (t4) {$\ket{011}$}; 
  \node[left=0 of a5] (t5) {$\ket{100}$}; 
  
  \path
    (a2) edge node {} (a3)
    (a3) edge node {} (a4)
    (a4) edge node {} (a5)
    (a1) edge[->,thick] node[xshift=-2mm,yshift=-1mm] {L} (a5)
    (a1) edge[->,thick] node[xshift=-2mm,yshift=1mm] {R} (a2);
  
\end{tikzpicture} }
\vspace{3mm}

\scalebox{0.9}{
$\ket{0}\otimes\ket{000}\rightarrow^{\text{step}}\frac{1}{\sqrt{2}}\ket{0}\otimes\text{L}\,\ket{000}+\frac{1}{\sqrt{2}}\ket{1}\otimes\text{R}\ket{000}$}
\caption{A step in a loop  with $K=5$ nodes.}
\label{fig:quantumwalka}
\end{subfigure}
\hfill
\begin{subfigure}[b]{0.43\textwidth}
\centering
\scalebox{1}{
\begin{quantikz}[wire types = {q,q,n,q}, row sep=3mm]
\lstick{$c$} & \gate{\text{H}}\gategroup[4,steps=3,style={draw=gray!40,fill=gray!20, inner
	xsep=1pt},background]{$\mathsf{step}$} & \octrl{1} & \ctrl{1} &\\
\lstick{$x_1$} & & \gate[3]{\text{L}} & \gate[3]{\text{R}} &  \\
\lstick{$\vdots$} & & &  &  \\
\lstick{$x_k$} & & & & 
\end{quantikz}}
\caption{Quantum circuit representation of $\mathsf{step}$.}
\label{fig:quantumwalk}
\end{subfigure}
\caption{Polytime quantum random walk.}
\end{figure}
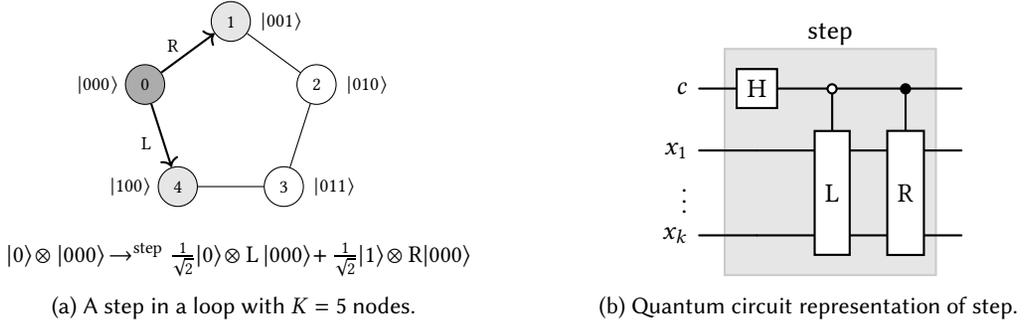

\subsection{No-Cloning and No-Deleting}

In order to allow for cloning of variable while ensuring general good behavior for qubits, \punq{} incorporates different features of the \dlal{} and \ls{} typing systems.

While \dlal{} exhibits full weakening (see Figure~\ref{fig:dlal}), there is a separation in the context between variables that may appear at most once (linear context) and more than once (exponential context), which then gets carried over into different types of abstractions. In \ls{}, on the other hand, we find a single context for typing judgements, and weakening is only allowed on so-called ``flat~types'', where a type is restricted in its use of superpositions. 

In \punq{}, the difference between the linear and exponential contexts captures the difference between allowed and disallowed treatment of qubits (or superpositions more generally), which then get resolved when creating the abstraction. Consider the following simple example of how \punq{} sensibly allows qubits to be the output of cloning, but not the input.

\begin{example}[Qubits made from cloning] With two instances of rule $(\mathsf{Ax})$ and one use of $(\times_i)$, it is straightforward to obtain the typing judgement $;x:\sharp \B,y:\sharp \B\vdash \pair{x}{y}:\sharp \B\times \sharp \B$. Now, moving $x$ and $y$ into the exponential context using rule $(\S_i)$ we obtain $x:\sharp \B,y:\sharp \B;\vdash \pair{x}{y}:\S(\sharp \B\times \sharp \B)$ where we are now allowed to contract $x$ and $y$ using rule $(\mathsf{C})$ and thus obtain $x:\sharp \B;\vdash \pair{x}{x}:\S(\sharp \B\times \sharp \B)$.

Though it seems as if the validity of this typing judgement allows for cloning of the $x$ variable representing a qubit, when we take the final step in closing the term, we must use rule $(\Rightarrow_i)$, since $x$ appears in the exponential context. Combined with the fact that $!(\sharp \B)= \B$, we obtain $;\vdash \lambda x.\pair{x}{x}:\B\Rightarrow \S(\sharp \B\times \sharp \B)$. Notice that this abstraction does not represent qubit cloning, but rather an operator that creates a pair of qubits using a boolean value that it takes as input to choose in which state the qubits should be.

\end{example}

In Section~\ref{s:unitarity} we will show that \punq{} preserves unitarity over qubit maps in the same way as \ls{} (Theorem~\ref{thm:unitarity}), where in cases in which the number of qubits increases, the \punq{} term must behave as an isometry. Likewise, no term exists that represents a map that decreases the number of qubits, which would constitute qubit deletion. A consequence of this fact is that we may not reset qubits to some pre-fixed value, as shown in the following example.

\begin{example}[Constant output zero]
To create a lambda abstraction in which the input qubit does not appear in the expression, the variable must be introduced via weakening, which is limited to the exponential context. For instance, using rules $(0)$, $(\leq)$, and $(\mathsf{W})$ we may obtain the typing judgement $x:\sharp \B;\vdash \ket{0}:\sharp \B$. Applying rule $(\Rightarrow_i)$ to create the abstraction, we obtain $;\vdash \lambda x.\ket{0}:\B\Rightarrow \sharp \B$, which corresponds to ``forgetting'' a boolean input and giving as output a qubit in state zero.
\end{example}

\section{Unitarity}
\label{s:unitarity}

In $\punq$, we have adopted the Lambda-$\mathcal{S}_1$ typing discipline to
ensure that linear maps over qubits, i.e., terms of type $\sharp \B \multimap
\sharp \B$, encode unitary operators. The proof of unitarity in~\cite{DCM22} employs
the realizability techniques developed in~\cite{DCGMV19}, which we will tailor to our setting. The major
differences include a straightforward extension to tuples of qubits --
demonstrating unitarity for types $\sharp (\B^n)\multimap \sharp(\B^n)$ -- and
the incorporation of $\dlal$ constructs: the modality $\S$, the intuitionistic
arrow $\Rightarrow$, as well as polymorphism.
All the proofs omitted in this section are fully developed in Appendix~\ref{app:unitarity}.

\subsection{Realizability}

The \textit{base} of a closed value $\vec{v}\in\V_c$ is
defined as the set of its basis values in $\BV_c$:
\begin{align*}
  \text{base}(v) & \triangleq\{v\} &
  \text{base}(\alpha\cdot \vec{v}) & \triangleq \text{base}(\vec{v}),\ \text{if }\alpha\not = 0,\ \emptyset \text{ otherwise.} \\
  \text{base}(\vec 0) &\triangleq \emptyset &
  \text{base}(\vec{v_1}+\vec{v_2}) & \triangleq \text{base}(\vec{v_1})\cup\text{base}(\vec{v_2})
\end{align*}
Given a set $S$ of closed values in a vector space $\V_c$, the \emph{span}
($\mathsf{span}$) and \emph{base} ($\flat$) of $S$ are defined as:
\begin{align*}
  \mathsf{span}(S) &\triangleq\big\{\sum_{i=1}^n \alpha_i\cdot \vec{v}_i:n\geq 0,\alpha_1,\dots,\ \alpha_n\in\Ct,\ \vec{v}_1,\dots,\vec{v_n}\in S\big\} &&\subseteq \V_c\\
  \flat S &\triangleq\bigcup_{\vec{v}\in S} \text{base}(\vec{v}) &&\subseteq \BV_c
\end{align*}
 
We first define the \textit{realizability predicate} for a superposition $\vec{t}$
and set of closed values $S\subseteq\V_c$, written $\vec{t}\Vdash S$ that is true if and
only if $\vec{t}\rightsquigarrow^* \vec{v}$ for some $\vec{v}\in S$. The
\textit{set of realizers of} $S$ 
is then defined as
$\{\vec{t}\in\SUP_c:\vec{t}\Vdash S\}$. 

Let 
$\mathcal{S}_1\subseteq \V_c$ be the set of
unit values, that is, $\mathcal S_1\triangleq\{\vec{v}\in
\V_c:|\braket{\vec{v}}{\vec{v}}|^2=1\}$.  Given a function $\tau$ from type variables to subsets of $\mathcal S_1$, the \textit{unitary semantics}
$\llbracket\cdot\rrbracket_\tau : \mathbb{T} \to \mathcal{P}(\mathcal S_1)$ is defined in Figure~\ref{fig:usem} by induction on types.
\begin{figure*}[h]
\hrulefill
\begin{align*}
  \llbracket X \rrbracket_\tau & \triangleq \tau(X),  &  \llbracket A \multimap B\rrbracket_\tau &\triangleq \left\{\lambda x.\vec{t}:\forall \vec{v} \in \llbracket A \rrbracket_\tau,(\lambda x.\vec{t})\vec v\Vdash \llbracket B\rrbracket_\tau\right\}\!,\\
  \llbracket \B \rrbracket_\tau  &\triangleq \{\ket{0}\,,\ket{1}\,\}, & \llbracket A \Rightarrow B\rrbracket_\tau & \triangleq \left\{\lambda x.\vec{t}:\forall v \in\flat \llbracket A\rrbracket_\tau, (\lambda x.\vec{t})v\Vdash\llbracket B\rrbracket_\tau\right\}\!,\\
   \llbracket \sharp A \rrbracket_\tau&\triangleq \mathsf{span}(\llbracket A \rrbracket_\tau)\cap \mathcal{S}_1, &\llbracket A\times B\rrbracket_\tau &\triangleq \left\{(\vec{v},\vec{w}):\vec v \in\llbracket A \rrbracket_\tau,\vec w \in\llbracket B \rrbracket_\tau\right\}\!,\\
  \llbracket \S A \rrbracket_\tau &\triangleq \llbracket A \rrbracket_\tau, &   \llbracket \forall X.A\rrbracket_\tau &\triangleq\bigcap_{R\subseteq\mathcal S_1} \llbracket A\rrbracket_{\tau \cup \{ X\to R\}}. 
\end{align*}
\hrulefill
\caption{Unitary semantics $\llbracket\cdot\rrbracket_\tau : \mathbb{T} \to \mathcal{P}(\mathcal S_1)$.} \label{fig:usem}
\end{figure*}

The following property of the definition of $!$ (Definition~\ref{def:bang}) clarifies the nature of $\flat$ and the role of $!$ in erasing superposition.
\begin{restatable}[]{lemma}{thmbangflat}
  Let $A\in \mathbb{T}_s$, then \(\llbracket !(A) \rrbracket_\emptyset = \flat \llbracket A \rrbracket_\emptyset\).
\label{thm:bangflat}
\end{restatable}
\begin{proof}
  By induction on $A\in\mathbb T_s$.
  The full proof is given in Appendix~\ref{app:unitarity}.
\end{proof}

The unitary semantics of a typing context $\Gamma$ with respect to $\tau$, called the \textit{unitary semantics} of $\Gamma$ and written $\llbracket \Gamma \rrbracket_\tau$, is defined as
\[
\llbracket \Gamma \rrbracket_\tau \triangleq\big\{\sigma \text{ substitution}: \text{dom}(\sigma)=\text{ dom}(\Gamma)\text{ and }\forall x\in \text{dom}(\sigma), \sigma(x)\in\llbracket \Gamma(x)\rrbracket_\tau\big\}.\]

We define $\vec{t}\langle \sum_i \alpha_i \cdot v_i/x \rangle\triangleq\sum_i \alpha_i\cdot \vec t[v_i/x]$. 
The notation $\vec{t}\langle \sigma \rangle$ is also used for any substitution $\sigma$.
We say that a typing judgment $\Gamma;\Delta \vdash \vec t: A$ is \emph{valid} when
$\text{dom}(\Delta)\subseteq FV(\vec{t})\subseteq \text{dom}(\Gamma,\Delta)$ and
$\vec{t}\langle \sigma\rangle \Vdash\llbracket  A \rrbracket_\tau$ for all $\sigma$ and $\tau$ such that $\sigma\in\llbracket \Gamma,\Delta\rrbracket_\tau$.
A typing rule is valid when the validity of its premises implies the validity of its conclusion.

\begin{restatable}[]{lemma}{lemsubstitutionRealizability}\label{lem:substitutionRealizability}
  Let $A,B\in\mathbb T$ and $\tau$ defined in $FV(A,B)$.
  Then,
  $\llbracket A[B/X]\rrbracket_\tau
  =
  \llbracket A\rrbracket_{\tau\cup\{X\mapsto\llbracket B\rrbracket_\tau\}}$.
\end{restatable}
\begin{proof}
  By induction on $A\in\mathbb T$.
  The full proof is given in Appendix~\ref{app:unitarity}.
  \end{proof}

\begin{restatable}[]{theorem}{thmvalidity}
The typing rules in Figure~\ref{fig:type} are valid.\label{thm:validity}
\end{restatable}
\begin{proof}
  We have to check each rule with respect to the definition of typing
  judgement. To exemplify, we give the case of rule $(\multimap_i)$. The
  proof for the remaining rules is given in Appendix~\ref{app:unitarity}.

  Suppose that $\Gamma;\Delta,x:A\vdash \vec{t}:B$ is a valid judgment. Then
  since $(\dom(\Delta)\cup \{x\}) \subseteq FV(\vec{t})$ we have
  $\dom(\Delta)\subseteq FV(\lambda x.\vec{t})$.  Similarly, since
  $FV(\vec{t})\subseteq (\dom(\Gamma,\Delta)\cup \{x\})$ it is also true that
  $FV(\lambda x.\vec{t})\subseteq \dom(\Gamma,\Delta).$ For any $\sigma\in
  \llbracket \Gamma,\Delta\rrbracket_\tau$, we have that $(\lambda
  x.\vec{t})\langle \sigma\rangle = \lambda x.\vec{t}\langle \sigma \rangle$,
  since $x\not\in \dom(\sigma)$. For all $\vec{v}\in\llbracket
  A\rrbracket_\tau$, we observe that $(\vec{t}\langle \sigma\rangle)\langle
  \vec{v}/x\rangle = \vec{t}\langle \sigma,\{\vec{v}/x\}\rangle\Vdash
  \llbracket B \rrbracket_\tau $, since $\sigma,\{\vec{v}/x\}\in\llbracket
  \Gamma,x:A\rrbracket_\tau$. Therefore, $(\lambda
  x.\vec{t})\langle\sigma\rangle \Vdash  \llbracket A\multimap B
  \rrbracket_\tau$.
\end{proof}

\subsection{Isometries over Qubits} \label{ss:isometries}
\begin{definition}[Isometry and unitarity]
  An operator $\mathcal{F}:\Ct^{2^n}\to \Ct^{2^k}$ is \emph{isometric} if
  it preserves the inner product, i.e., if $\langle \mathcal{F}(u) |
  \mathcal{F}(v)\rangle=\langle u | v\rangle$, $\forall u,v\in \Ct^{2^n}$.
  Furthermore, if $\mathcal{F}:\Ct^{2^n}\to \Ct^{2^n}$ is isometric then
  $\mathcal{F}$ is \emph{unitary}. \label{def:isouni}
\end{definition}

To simplify notation, for $i=0,\dots,2^n-1$, let $\ket{i}\triangleq\pair{\ket{i_1}\,}{\pair{\ket{i_2}\,}{\dots
\pair{\ket{i_{n-1}}\,}{\ket{i_n}\,}\dots}}$, where $i_1\dots i_n$ is the binary representation of $i$. We show
that any linear application between qubit terms must preserve orthogonality and the norm of basis states.

\begin{restatable}{lemma}{lemunitaryabstractions}\label{lemma:unitary-abstractions}
For any $\vec{t}\in\punq$  and any $k\geq n$ for $k,n\in\mathbb{N}-\{0\}$,
\[
\vec{t}\in\llbracket \sharp (\B^n) \multimap \sharp (\B^k)\rrbracket_\emptyset 
\iff   
\begin{cases} \forall i=0,\dots,2^n-1, \exists \vec{v}_i\in \llbracket \sharp(\B^k)\rrbracket_\emptyset \text{ such that }\vec{t}\ \ket{i} \rightsquigarrow^* \vec{v}_i\text{,}\  \text{and}\\
 \vec v_i\perp_{\sharp(\B^k)} \vec v_j,\,\forall i\not = j.
 \end{cases}
\]
\end{restatable}
\begin{proof}
  The proof is given in Appendix~\ref{app:unitarity}. 
\end{proof}

This property will allow us to show that closed superpositions of qubit applications behave as isometric operators. Let $\lmap{-}$ be the linear map from closed values of type $\sharp (\B^n)$ to $\Ct ^{2^n}$ defined by $\lmap{\sum_{i=0}^{2^n-1} \alpha_i\cdot \ket{i}} \triangleq  \begin{pmatrix}
 \alpha_0 & \cdots & \alpha_{2^n -1}
   \end{pmatrix}^T \in \Ct ^{2^n}$.
  A closed superposition $\vec{t} \in \SUP_c$ of type $\sharp (\B^n) \multimap \sharp (\B^k)$ \emph{represents} an operator $\mathcal{F}:\Ct ^{2^n}\rightarrow \Ct^{2^k}$ if, for all values $\vec{v}$ of type $\sharp (\B^n)$ and $\vec{w}$ of type $\sharp (\B^k)$,  it holds that $\vec{t}\ \vec{v}\rightsquigarrow^* \vec{w}$ if and only if $\mathcal{F}(\lmap{\vec{v}})=\lmap{\vec{w}}$.

By linearity of $\multimap$, preserving the norm of the basis states is enough to show that an abstraction represents an isometry. Moreover, when the dimensions $n$ and $k$ match the abstraction represents a unitary operator.

\begin{theorem}[Soundness]\label{thm:unitarity}
For any closed superposition  $\lambda x.\vec{t}\in\llbracket\sharp(\B^n)\multimap \sharp(\B^k)\rrbracket_\emptyset$  and any $k, n \in \mathbb{N}-\{0\}$, the following hold:
  \begin{enumerate}
    \item If $k\geq n$, then $\lambda x.\vec{t}$ represents an isometry in $\Ct^{2^n}\to\Ct^{2^k}$,
    \item If $k= n$, then $\lambda x.\vec t$ represents a unitary operator in
      $\Ct^{2^n}\to\Ct^{2^n}$.
  \end{enumerate}
  Furthermore, if $k<n$, then the type is uninhabited.
\end{theorem}
\begin{proof}
  Let $\|\cdot\|$ and $\cdot^\dagger$ represent the typical vector norm and the conjugate transpose, respectively.
  
 1) Let $\lambda x.\vec{t}\in\llbracket \sharp(\B^n)\multimap \sharp(\B^k)\rrbracket_\emptyset$. Then, from
      Lemma~\ref{lemma:unitary-abstractions}, for $i=0,\dots,2^n-1$, there exist
      $\vec{v}_i\in\llbracket \sharp(\B^k)\rrbracket_\emptyset$ such that
      $\vec{t}[\ket{i}/x]\rightsquigarrow^* \vec{v}_i$ with $\langle
      \vec{v}_i | \vec{v}_j\rangle=0$, for $i\not =j$. Let
      $\mathcal{V}:\Ct^{2^n}\to\Ct^{2^k}$ be the linear operator defined by
      $\mathcal{V}(\lmap{\ket{i}})=\lmap{\vec{v}_i} $, for each $i=0,\dots,2^n-1$.
      By the linearity of the calculus we have that $\lambda x.\vec{t}$
	  represents the operator $\mathcal{V}$. Moreover, $\mathcal{V}$ is isometric by Definition~\ref{def:isouni} since $\|\lmap{ \vec{v}_i }\|=1$     and $\lmap{ \vec{v}_i }^\dagger\cdot  
      \lmap{ \vec{v}_j} =0$, for $0\leq i\not =
      j\leq 2^n-1$.

2) The case $(k=n)$ is easy to show, since from the case $(k\geq n)$ the
      term $\lambda x.\vec{t}$ must represent some isometric operator that in
      this case is also dimension-preserving, so by Definition~\ref{def:isouni}
      it must also be unitary.
      Let us now show that the type $\sharp (\B^n)\multimap \sharp (\B^k)$ is
      uninhabited for $k<n$. We will consider an argument by contradiction.  Let us
      start with a typable and closed abstraction $\lambda
      x.\vec{t}\in\llbracket \sharp (\B^n)\multimap \sharp
      (\B^k)\rrbracket_\emptyset$. Then, let us define the values $\vec v_i\in
      \llbracket \sharp (\B^k)\rrbracket_\emptyset$ satisfying $(\lambda
      x.\vec{t})\ \ket{i}\rightsquigarrow^* \vec{v}_i $. Since there are $2^n$
      such $\vec{v}_i$, they must not be linearly independent, meaning that
      there is a choice of $\beta_i\in\Ct$ such that $\sum_{i=1}^{2^n} \beta_i\cdot
      \vec v_i = \vec{0}$ and not all values $\beta_i$ are zero. Then, we may
      define the superposition $\vec{s}\triangleq \sum_i \frac{\beta_i}{\sqrt{\sum_i
      |\beta_i|^2}} \cdot \ket{i} \in \llbracket \sharp
      (\B^n)\rrbracket_\emptyset$ for which $(\lambda x.\vec{t})\ \vec{s}
      \rightsquigarrow^* \vec{0}\not \in \llbracket \sharp
      (\B^k)\rrbracket_\emptyset$. This concludes the proof.
\end{proof}

Notice that, since we have that $\llbracket \S A \rrbracket_\tau\triangleq \llbracket A \rrbracket_\tau$ in the unitary semantics, the result of Theorem~\ref{thm:unitarity} also applies to terms $\lambda x.\vec{t}$ obtained from duplication and therefore having a type that contains the $\S$ symbol, since $\llbracket \S^a \sharp (\B^n)\multimap \S^b \sharp (\B^k)\rrbracket_ \emptyset= \llbracket \sharp (\B^n)\multimap \sharp (\B^k)\rrbracket_\emptyset$.

\begin{restatable}[Completeness]{theorem}{thmisocomplete}\label{thm:isocomplete}
For any isometry $\mathcal{I}:\Ct^{2^n}\to\Ct^{2^k}$, with $k\geq n \geq 1$, there exists a closed superposition  $\lambda x.\vec{t}\in\SUP_c$ of type $\sharp(\B^n)\multimap \sharp(\B^k)$ that represents $\mathcal{I}$.
\end{restatable}
\begin{proof}
Constructive proof given in Appendix~\ref{app:unitarity}.
\end{proof}

On top of satisfying unitarity, $\punq{}$ also satisfies another natural property of quantum computation, which is that qubits may occupy entangled states and that there is no operation that disentangles any general state.

\begin{restatable}[Non-separability]{theorem}{thmnonseparability}\label{thm:non-separability}
  For any $n,k\in\mathbb{N}-\{0\}$, the type $\sharp (\B^{n+k})\multimap
  \big(\sharp(\B^{n})\times \sharp(\B^{k})\big)$ is uninhabited. 
\end{restatable}
\begin{proof}
Proof in Appendix~\ref{app:unitarity}.
\end{proof}

\section{Polytime normalization}
\label{s:ptstrongnormalization}
In this section, we show that \punq{} preserves the polytime strong
normalization properties of $\dlal$~\cite{BT09}, when adapted to the context of quantum computation. For that purpose, we define an
encoding from superpositions in $\punq$ to sets of $\dlal$ terms and show that the polynomial time strong normalization of each of the terms in this set implies the polytime normalization of the initial $\punq$ program.
All the proofs omitted in this section are fully developed in
Appendix~\ref{app:ptstrongnormalization}.

\subsection{Dual Light Affine Logic}
$\dlal$ terms~\cite{BT09} are a strict subset of System $F$~\cite{G72,R74} that can be typed
with respect to the rules of Figure~\ref{fig:dlal}.
The set $\mathbb S$ of types in $\dlal$ is the set
produced by the following grammar:
\[
  A := X\mid A \multimap A \mid A \Rightarrow A \mid 
  \S A \mid\forall X.A,
\] 
and can be viewed as a subset of $\mathbb{T}$.  The notion of typing
environments $\Gamma,\Delta$ are defined in a standard way. For a given set of $\dlal$ terms $\mathcal{S}$, we write $\Gamma;\Delta\vdash_{\dlal} \mathcal{S}:A$ if $\forall t \in \mathcal{S}$, the derivation of $\Gamma;\Delta\vdash t:A$ can be obtained in $\dlal$.
\begin{figure}[t]
  \hrulefill
  \[
    \begin{prooftree}
      \hypo{\Gamma;\Delta \vdash {t}:A}
      \infer1[(Weak)]{\Gamma,\Gamma';\Delta,\Delta'\vdash {t}:A}
    \end{prooftree}
    \qquad\quad
    \begin{prooftree}
      \hypo{\Gamma,x:B,y:B;\Delta\vdash {t}:A}
      \infer1[(Cntr)]{\Gamma,x:B;\Delta \vdash {t}[x/y]:A}
    \end{prooftree}
    \qquad\quad
    \begin{prooftree}
      \hypo{\phantom{A\leq B}}
      \infer1[(Id)]{;x: A\vdash x: A}
    \end{prooftree}
  \]
  
  \[
    \begin{prooftree}
      \hypo{\Gamma;\Delta,x:A\vdash {t}:B}
      \infer1[($\multimap$ i)]{\Gamma;\Delta \vdash \lambda x.{t}:A\multimap B}
    \end{prooftree}
    \qquad\quad
    \begin{prooftree}
      \hypo{\Gamma;\Delta \vdash t: A\multimap B}
      \hypo{\Gamma' ;\Delta' \vdash s:A}
      \infer2[($\multimap$ e)]{\Gamma,\Gamma';\Delta,\Delta' \vdash t \ s:B}
    \end{prooftree} 
  \]
  
  \[
    \begin{prooftree}
      \hypo{\Gamma,x:A;\Delta \vdash \vec t:B}
      \infer1[($\Rightarrow$ i)]{\Gamma;\Delta \vdash{\lambda x.\vec t}: A \Rightarrow B}  
    \end{prooftree} 
    \qquad\quad
    \begin{prooftree}
      \hypo{\Gamma;\Delta \vdash {t}:{A\Rightarrow B}}
      \hypo{;[z:C]\vdash  s:A}
      \infer2[($\Rightarrow$ e)]{\Gamma,[z:C]\ ;\Delta \vdash {t \ s}:B}
    \end{prooftree}
  \]
  
  \[
    \begin{prooftree}
      \hypo{;\Gamma,\Delta\vdash t:A}
      \infer1[($\S$ i)]{\Gamma;\S\Delta\vdash t:\S A}
    \end{prooftree}
    \qquad\quad
    \begin{prooftree}
      \hypo{\Gamma;\Delta \vdash s:\S B}
      \hypo{\Gamma';\Delta',x:\S B\vdash t:A}
      \infer2[($\S$ e)]{\Gamma,\Gamma';\Delta,\Delta'\vdash t[s/x]:A}
    \end{prooftree}
  \]
  
  \[
    \begin{prooftree}
      \hypo{\Gamma;\Delta\vdash t:A}
      \hypo{X \notin FV(\Gamma,\Delta)}
      \infer2[($\forall$ i)]{\Gamma;\Delta \vdash t:\forall X.A}
    \end{prooftree}
    \qquad\quad
    \begin{prooftree}
      \hypo{\Gamma;\Delta \vdash  t:\forall X.A}
      \infer1[($\forall$ e)]{\Gamma;\Delta \vdash t:A[B/X]}
    \end{prooftree}
  \]
  \hrulefill
  \caption{Dual Light Affine Logic ($\dlal$).}
  \label{fig:dlal}
\end{figure}

The size of a \dlal{} term $t$, denoted $|t|$, is defined as the number of its
symbols, and the size of a set of terms $\mathcal{S}$, noted $|\mathcal{S}|$,
is defined as the biggest size of a term $t \in \mathcal{S}$, i.e.,
$|\mathcal{S} | \triangleq \max\{|t| \mid  t \in \mathcal{S}\}$. The reason for this choice is that superpositions in \punq{} will be encoded as sets of terms in \dlal{}, which in a sense will encode the set of base terms present in the superposition. Since in \punq{} the size of a term is the maximum size appearing in a superposition, this choice ensures that we capture the correct notion of polynomial size of the final term.

The following theorem implies both a polynomial bound on the reduction length and a polynomial bound on the size of each term obtained during the reduction in the \dlal{} system.


\begin{theorem}[Polynomial time strong normalization~\citewiththeorem{Thm.~8}{BT09}]\label{thm:BT}
\label{thm:dlal-polytime-strong}
If $\Gamma;\Delta\vdash_{\dlal} {t}:A$, then there exists a constant $d$,
referred to as the \emph{depth of $t$}, such that $t$ reduces to its normal form in at
most $O(|t|^{2^d})$ reduction steps and within a time complexity of
$O(|t|^{2^{d+2}})$ on a deterministic Turing machine. This result holds
independently of the chosen reduction strategy. 
\end{theorem}

Given two sets of \dlal{} terms $\mathcal{S}$ and $\mathcal{T}$, we write $\mathcal{S} \to \mathcal{T}$ if $\mathcal{T} = \{t \mid s \in \mathcal{S} \wedge (s \rightarrow t \vee s\not\rightarrow )\}$, meaning that the set $\mathcal{T}$ is obtained from $\mathcal{S}$ by attempting to apply a single reduction step. Therefore, if the term is already normalized it remains in the set.
Let $\to^n$ be the $n$-fold transitive
closure of $\rightarrow$, i.e., if $\mathcal{S}_0 \to \ldots \to
\mathcal{S}_{n}$ then $\mathcal{S}_0 \to^n \mathcal{S}_n$. We also write
$\mathcal{S}\to^{\leq n} \mathcal{T}$ if there exists $k \leq n$ such that
$\mathcal{S} \to^k \mathcal{T}$. Termination is defined as reaching a point where all terms in the set are in normal form, i.e., $\mathcal{S}\downarrow^n$ if $\mathcal{S}\rightarrow^n \mathcal{T}$ and $\mathcal{T}\rightarrow \mathcal{T}$. Similarly, we write $\mathcal{S}\downarrow^{\leq n}$ for termination in at most $n$ steps.

\begin{corollary}\label{coro:sound}
For any type $A\in\mathbb{T}$, there is a polynomial $P_A
\in \mathbb{N}[X]$ such that, if $\Gamma;\Delta\vdash_{\dlal} \mathcal{S}:A$, then $\mathcal{S} \downarrow^{\leq P_A(|\mathcal{S}|)}$ and, for each $\mathcal{T}$ satisfying $\mathcal{S} \to^n
\mathcal{T}$ for some $n\in\mathbb{N}$, $|\mathcal{T}| \leq  P_A(|\mathcal{S}|).$
\end{corollary}
Notice that the polynomial $P_A$ can be explicitly computed and
only depends on $A$ as the depth of a term is a function of its type
(see~\cite{BT09}). Furthermore, Corollary~\ref{coro:sound} applies for any reduction strategy, so we may fix the strategy to be call-by-value for our purposes.


\subsection{Encoding of \texorpdfstring{$\punq$}{PUNQ} in \texorpdfstring{$\dlal$}{DLAL}}

In order to prove \punq{}'s polynomial time normalization, we will give a translation from \punq{} terms into sets of DLAL terms. We will then show that, for the \punq{} terms of interest (i.e. those that are typable), their reduction length in \punq{} is upperbounded by the maximum reduction length in its translation. Given that this encoding will only increase polynomially the size of the term, this will provide a polynomial upper bound in the reduction of \punq{} terms.

We extend the \dlal{} syntax with the term $\ast$, which will represent the encoding of $\vec{0}$, and the following typing rule ensuring that $\ast$ has any type:
\[\begin{prooftree}
\infer0[$(\ast)$]{;\vdash \ast:\forall X.X}
\end{prooftree}\]

We will call this extension of the \dlal{} syntax and typing system $\dlal_\ast$. Since $\ast$ does not reduce in any way, the polytime termination result in Theorem~\ref{thm:dlal-polytime-strong} and its corollary also apply to $\dlal_\ast$.

Given two sets of terms $\mathcal{S},\mathcal{T}$, let
$\mathcal{S}\ \mathcal{T}$ be syntactic sugar for the set $\{s\ t \mid s \in
\mathcal{S} \wedge t \in \mathcal{T}\}$.  We now define formally the map
$\lpar{\cdot}: \punq{}\to \mathcal{P}(\text{$\dlal$}_\ast)$ in Figure~\ref{fig:mapenc} and the type encoding $(\cdot)^\star: \mathbb{T}\to \mathbb{S} $ in Figure~\ref{fig:typenc}.
This encoding is extended to contexts as follows: $\Gamma^\star \triangleq
\{x_i:(A_i)^\star \mid x_i:A_i\in\Gamma\}$.

Such an encoding does not require that the set of \dlal{} terms perfectly simulate the corresponding \punq{} term, but it must contain enough of the structure of the original term to provide a valid upper bound in its number of reduction steps. As such, the type encoding $(\cdot)^\star$ in Figure~\ref{fig:typenc} can be seen as a ``dequantization'' of the type (i.e. removing quantum superposition) combined with translation into \dlal{} types that correspond to the basic \punq{} types, such as the encoding of pairs and booleans.

The term translation $\lpar{\cdot}$ from $\punq$ terms to sets of $\dlal_\ast$ terms is standard, with a few exceptions. In \dlal{}, a term encoding the conditional ``$\tif t{t_1}{t_2}$'' would be typed multiplicatively, requiring disjoint contexts for $t_1$ and $t_2$. However, for the purpose of guaranteeing unitarity, the conditional is typed additively in \punq{} (i.e., the two contexts are the same).
To handle this difference, we rewrite the conditional such that the branches have no open linear variables, which allows for a correct typing derivation in $\dlal_\ast$. This transformation in general requires that we use information on the derivation of the term, but this is handled implicitly in the translation. For some natural number $k$, we define the notation $\lambda x_{[1..k]}.\vec{t}\triangleq \lambda x_1.\lambda x_2.\dots \lambda x_k.\vec{t}$. Similarly, when $\vec{t}$ contains $k$ instances of a free variable $x$, we denote by $\vec{t}_{[x:1..k]}$ the term $\vec{t}$ where each occurrence of $x$ is replaced by $x_i$, for $i=1,\dots,k$. For instance, for $\vec{t}=\lambda y.x(x\ y)$, we have $\vec{t}_{[x:1..k]}=\lambda y.x_1(x_2\  y)$. 

Also of note is the translation of the abstraction ``$\lambda x.\vec{t}$'', where we perform a \textit{linearization} of the lambda abstraction, separating all the repeated instances of $x$ with separate variables, labeled $x_i$. Likewise, in the application case ``$t_1\ t_2$'', we repeat the encoding of $t_2$ the appropriate number of times. This is encoded in a function $\eta:\punq{}\to\{0,1,\dots\}$, encoding the number of times the first input is repeated. This ensures that generated set of \dlal{} terms will account for all possible branches contained in the original \punq{} term (see Example~\ref{ex:2had}). For instance, $\eta(\lambda x.x)=1$ since $x$ is not duplicated, and $\eta((\lambda x.x)(\lambda f.\lambda y.f\ f\ y))= 2$ since once the term reduces to the value $\lambda f.\lambda y.f\ f\ y$, the input $f$ appears twice in the term.

Let us now look at a few properties of the type encoding $(\cdot)^\star$. A first property we may notice is that the type encoding does not change if we first transform a type into its clonable version, i.e. when we apply the ! function.

\begin{restatable}[]{lemma}{lemtranslationbang}\label{lem:translation-bang}
For any $A\in\mathbb{T}$, we have that $(!(A))^\star=A^\star.$
\end{restatable}
\begin{proof}
By induction on the structure of $A$.
\end{proof}

Similarly, a useful fact is that the encoding is preserved by \punq{} subtyping.

\begin{restatable}[]{lemma}{lemtranslationsubtype}\label{lem:translation-subtype}
For any $A,B\in\mathbb{T}$, if $A\leq B$ then $A^\star = B^\star$.
\end{restatable}
\begin{proof}
By inspection of the rules in Figure~\ref{fig:subtype}.
\end{proof}

\begin{figure*}[t]
\hrulefill
\centering
$$ {\lpar{\vec{0}}\triangleq \{\ast\}}  \qquad\qquad {\lpar{x}} \triangleq \{x\}\qquad\qquad 
  {\lpar{\ket{0}}} \triangleq\{ \lambda x.\lambda y.x \} \qquad\qquad
  {\lpar{\ket{1}}} \triangleq \{\lambda x.\lambda y.y\}$$
  \vspace*{-2mm}
 $$ \lpar{\tif{t}{t_1}{t_2}}  \triangleq  \Big(\lpar{t}\ \lambda x_{[1..k]}.\lpar{t_1}\ \lambda x_{[1..k]}.\lpar{t_2}\Big)\ x_1 \dots x_k,$$
 with $\Gamma;x_1:A_1,\dots,x_k:A_k\vdash t_1,t_2:B$
   \vspace*{1mm}
 $${\lpar{\tlet{x}{y}{t_1}{t_2}}} \triangleq {\lpar{t_1}\ (\lambda x.\lambda y.\lpar{t_2})}\qquad \qquad \quad {\lpar{\alpha\cdot \vec t}} \triangleq {\lpar{\vec t} }\textnormal{, if $\alpha\not = 0$, $\{\ast\}$ otherwise}$$
  \vspace*{-1mm} 
$$ {\lpar{\vec{t}_1+\vec{t}_2}} \triangleq {\lpar{\vec{t}_1}\cup \lpar{\vec{t}_2}} \qquad\quad  {\lpar{t_1\ t_2}} \triangleq \big(\dots \big({\lpar{t_1}}\ \lpar{t_2}\big)\ \lpar{t_2}\ \big)\dots \lpar{t_2}\big)\quad \text{($\eta(t_1)$ times)}$$
  \vspace*{-1mm}
$${\lpar{\pair{t_1}{t_2}}} \triangleq {\lambda x. (x\ \lpar{t_1}\ \lpar{t_2}) } \qquad \qquad  \lpar{\lambda x.\vec{t}}\triangleq \lambda x_{[1..k]}.\lpar{\vec{t}_{[x:1.. k]}},\quad k=\eta(\lambda x.\vec{t})$$

\hrulefill
\caption{$\lpar{\cdot}: \punq{}\to \mathcal{P}(\text{$\dlal$}_\ast).$}\label{fig:mapenc}
\end{figure*}
%
%
%
%
\begin{figure*}[t]
\hrulefill
\begin{align*}
  (X)^\star &\triangleq X &   (\B)^\star &\triangleq\forall X .(X\multimap X\multimap X) \\  
  (A_1\multimap A_2)^\star &\triangleq (A_1)^\star \multimap (A_2)^\star &   (A_1\Rightarrow A_2)^\star &\triangleq (A_1)^\star \Rightarrow (A_2)^\star \\ 
  (A_1\times A_2)^\star &\triangleq \forall X.(((A_1)^\star \multimap (A_2)^\star \multimap X) \multimap X) &
  (\sharp Q)^\star &\triangleq (Q)^\star \\
  (\S A)^\star &\triangleq \S (A)^\star &
  (\forall X.A)^\star &\triangleq\forall X.(A)^\star  
\end{align*}
\hrulefill
\caption{$(\cdot)^\star: \mathbb{T}\to \mathbb{S}.$}\label{fig:typenc}
\end{figure*}

Regarding the term translation $\lpar{\cdot}$, we can check that it is preserved by the equivalence relation $\equiv$ in \punq{}, modulo the presence of the term $\ast$.

\begin{definition} Let $\mathcal{S},\mathcal{T}$ be two sets. Then $\mathcal{S}\subsetast \mathcal{T}$ if $\mathcal{S}\subseteq \mathcal{T}\cup\{\ast\}$. We also define $\mathcal{T}\eqast \mathcal{S}$ if $\mathcal{T}\subsetast \mathcal{S}$ and $\mathcal{S}\subsetast \mathcal{T}$.
\end{definition}

\begin{restatable}[]{lemma}{lemtransequiv}\label{lem:transequiv}
If $\vec{t}\equiv \vec{r}$, then $\lpar{\vec{t}}\eqast\lpar{\vec{r}}$.
\end{restatable}
\begin{proof}
By inspection of the equivalence relation. The full proof is given in
  Appendix~\ref{app:ptstrongnormalization}.
\end{proof}

We now turn to the normalization behavior of each term in the set corresponding to a given translation. First, we may notice that every term in a translation has the same type, which corresponds to the type given by the $(\cdot)^\star$. Notice also that, though not explicitly stated, the notion of depth of a typing derivation~\cite{BT09} is asymptotically preserved by the translation.


\begin{restatable}[]{lemma}{lemdlaltranslation}\label{lem:dlaltranslation}
$\Gamma;\Delta\vdash \vec t: A$ implies $\Gamma^\star;\Delta^\star \vdash_{\dlal} \lpar{\vec t}:A^\star$. 
\end{restatable}
\begin{proof}
  By induction on the derivation of $\Gamma;\Delta\vdash \vec t: A$. The full proof is given in
  Appendix~\ref{app:ptstrongnormalization}.
\end{proof}

Another important detail is that the encoding increase the size of the original \punq{} term only polynomially and is therefore an appropriate way to connect the complexity of the two systems.

\begin{restatable}[]{lemma}{lemsize}\label{lem:size}
For any $A\in\mathbb{T}$, there exists a polynomial $P_A$ such that, for any $\vec t$ where $;\vdash \vec{t}:A$, the following holds $|\lpar{\vec t}|= O(P_A(|\vec t|))$.
\end{restatable}
\begin{proof}
  By induction on the structure of $\vec t$. The full proof is given in
  Appendix~\ref{app:ptstrongnormalization}.
\end{proof}

To demonstrate how this translation guarantees polynomial-time termination in \punq{}, we consider the relation between the reduction of a \punq{} term $\vec{t}$ and the reduction of its translation $\lpar{\vec{t}}$. Given our choice of translation, a term in \punq{} will correspond to a set in \dlal{} containing not only all basis states present in any superposition, but also all possible branchings in the program.

We now show that a reduction in $\punq{}$ corresponds to at least one reduction over its encoding in $\dlal$. Notice first that, once we reach a value in \punq{}, its DLAL translation will not reduce.

\begin{restatable}[]{lemma}{lemvalues}\label{lem:values}
For any $\vec{v}\in\V_c$, we have that $\lpar{\vec v}$ does not reduce.
\end{restatable}
\begin{proof}
  By induction on the structure of $\vec v$. The full proof is given in Appendix~\ref{app:ptstrongnormalization}.
\end{proof}

Now we consider the central lemma for the proof of Theorem~\ref{thm:soundness}.

\begin{restatable}[]{lemma}{lemtransitiondlal}\label{lem:transition-dlal}
For any $\vec t \in \punq{}$, $\vec t\rightsquigarrow \vec r$ implies $\lpar{\vec r}\subsetast \bigcup_{n>0}\{\mathcal{S}_n \mid \lpar{\vec t} \to^n \mathcal{S}_n\}$.
\end{restatable}
\begin{proof}
  By induction on relation $\rightsquigarrow$. The full proof is given in
  Appendix~\ref{app:ptstrongnormalization}.
\end{proof}

\begin{example}[Two Hadamards]\label{ex:2had} Let $\vec{t}=(\underline{2}\ \mathsf{H})\ \ket{0}\,$, having \punq{} type $\S \sharp \B$, correspond to applying twice the Hadamard gate (Example~\ref{ex:hadamard}) to state $\ket{0}\,$. We denote by $\mathcal{S}_0$ its translation into a set of \dlal{} terms, and, for all $n>0$, let $\mathcal{S}_{n}$ be the set such that $\mathcal{S}_0\rightarrow^n \mathcal{S}_n$. Let $\mathtt{tt}\triangleq \lambda x.\lambda y. x$ and $\mathtt{ff}\triangleq \lambda x.\lambda y. y$ be a shorthand notation for the boolean encoding in \dlal{}. First, we translate the Hadamard term. Notice that, since $\ket{+}$ and $\ket{-}$ are closed terms, we do not need to add any new variables in the translation of the $\tif{(\cdot)}{(\cdot)}{(\cdot)}$. The translation is then:
\[\lpar{\mathsf{H}}=\lambda x.\lpar{\tif{x}{\ket{+}}{\ket{-}}}=\lambda x.x\ \lpar{\ket{+}}\ \lpar{\ket{-}} = \lambda x.x\ \{\mathtt{ff},\mathtt{tt}\}\ \{\mathtt{ff},\mathtt{tt}\},\]
which corresponds to the set of 4 terms $\{\lambda x.x\ b_1\ b_2 \mid b_1,b_2\in\{\mathtt{ff},\mathtt{tt}\}\}$.
The translation of $\vec{t}$ corresponds to 16 \dlal{} terms, all of which have type $(\S \sharp \B)^\star =\S(\forall X. X\multimap X\multimap X)$ in $\dlal{}_*$:
\[\lpar{\vec{t}}=\mathcal{S}_0=\big\{ (\lambda f_1.\lambda f_2.\lambda q.f_1(f_2\ q)\ (\lambda x.x\ b_1\ b_2)\ (\lambda x.x\ b_3\ b_4))\ \mathtt{tt}\mid b_i\in \{\mathtt{ff},\mathtt{tt}\}\big\}.\]
This set reduces in 7 steps, according to the following progression:
\begin{align*}
\mathcal{S}_1&=\big\{ (\lambda f_2.\lambda q.(\lambda x.x\ b_1\ b_2)(f_2\ q)\ (\lambda x.x\ b_3\ b_4))\ \mathtt{tt}\mid b_i\in \{\mathtt{ff},\mathtt{tt}\}\big\}\\
\mathcal{S}_2&=\big\{(\lambda q. (\lambda x.x\ b_1\ b_2)\ (\lambda x.x\ b_3\ b_4)\ q)\ \mathtt{tt}\mid b_i\in\{\mathtt{ff},\mathtt{tt}\}\big\}\\
 \mathcal{S}_3 &= \big\{(\lambda x.x\ b_1\ b_2)\ (\lambda x.x\ b_3\ b_4)\ \mathtt{tt}\mid b_i\in\{\mathtt{ff},\mathtt{tt}\}\big\}\\
 \mathcal{S}_4 &=\big\{(\lambda x.x\ b_1\ b_2)\ (\mathtt{tt}\ b_3\ b_4)\mid b_i\in\{\mathtt{ff},\mathtt{tt}\}\big\}\\
 \mathcal{S}_5 &=\big\{(\lambda x.x\ b_1\ b_2)\ b_3\mid b_i\in\{\mathtt{ff},\mathtt{tt}\}\big\}\\
 \mathcal{S}_6 &=\big\{b_3 \ b_1\ b_2\mid b_i\in\{\mathtt{ff},\mathtt{tt}\}\big\}\\
 \mathcal{S}_7 &=\{\mathtt{ff},\mathtt{tt}\}
\end{align*}

In \punq{}, $\vec{t}$ reduces to a value in 6 steps. Defining $\vec{r}_i$ as $\vec{r}_0\triangleq \vec{t}$ and $\vec{r}_i\rightsquigarrow \vec{r}_{i+1}$, we have  
\begin{align*}
\vec{r}_1 &= (\lambda x. \mathsf{H}\ \mathsf{H}\ x)\ \ket{0} && \lpar{\vec r_1}= \mathcal{S}_2\\
\vec{r}_2 &= \mathsf{H}\ \mathsf{H}\ \ket{0} &&\lpar{\vec r_2}=\mathcal{S}_3\\
\vec{r}_3 &= \mathsf{H}\ (\tif{\ket{0}}{\ket{+}}{\ket{-}\,})&&\lpar{\vec r_3}= \mathcal{S}_4\\
\vec{r}_4 &= \sfrac{1}{\sqrt{2}}\cdot\mathsf{H}\ \ket{0} + \sfrac{1}{\sqrt{2}}\cdot\mathsf{H}\ \ket{1}&&\lpar{\vec r_4}= \mathcal{S}_5\\
\vec{r}_5 &= \sfrac{1}{\sqrt{2}}\cdot\tif{\ket{0}}{\ket{+}}{\ket{-}\,} + \sfrac{1}{\sqrt{2}}\cdot\tif{\ket{1}}{\ket{+}}{\ket{-}\,}&&\lpar{\vec r_5}= \mathcal{S}_6\\
\vec{r}_6 &= \ket{0}&&\lpar{\vec r_6}=\{\mathtt{tt}\}\subseteq \mathcal{S}_7
\end{align*}
Notice that the PUNQ term was not perfectly simulated since $\mathtt{ff}$ also appears in $\mathcal{S}_7$. However, in every step of the reduction $\vec{r}_0\rightsquigarrow ^\ast\vec{r}_6$, we find the translation of $\vec{r}_i$ in a further reduction of $\mathcal{S}_0$, which suffices to bound the number of steps until achieving normalization.
\end{example}

The previous results allow us to conclude that, like $\dlal$, \punq{} ensures polytime normalization.

\begin{theorem}[Polynomial time normalization]\label{thm:soundness}
  For any $A\in \mathbb{T}_c$, there is a polynomial $P_A \in \mathbb{N}[X]$
  such that, for any $\punq{}$ superposition $\vec t$ with type $A$,  $\vec t$ reduces to a value $\vec v$ in at most $O(P_A(|\vec t|)$
    reduction steps and $|\vec v|=O(P_A(|\vec t|))$.
\end{theorem}
\begin{proof}
  Let $\vec t\in\punq{}$ such that $\Gamma;\Delta\vdash \vec t: A$. Then, by
  Lemma~\ref{lem:dlaltranslation},  $\Gamma^\star;\Delta^\star
  \vdash_{\dlal_\ast}  \lpar{\vec t}: A^\star$ holds and each term $t_i\in\lpar{\vec t}$ has size at most $O(|\vec t|)$.  By Corollary~\ref{coro:sound}, there is a polynomial $P_{A^\star} \in
  \mathbb{N}[X]$ such that every $\dlal_\ast$ term $t_i\in\lpar{\vec t}$ can be reduced in $N\triangleq O(P_ {A^\star}(|t_i|))$ steps. Since the size $|t_i|$ of $\dlal_\ast$ terms is bounded by $P_A(|\vec t|)$, from Lemma~\ref{lem:size}, we have $N=O(P_{A^\star}(P_A(|\vec t|)))$. Therefore all the terms in
  $\lpar{\vec t}$ reduce in polynomial time on the size of $\vec t$. As shown in Lemma~\ref{lem:transition-dlal}, each reduction step in {\punq} corresponds to at least one reduction step applied to all the (non-normalized) terms in its translation into $\dlal_\ast$, we conclude that $\vec t$ must reduce to normal form in at most $N$ steps.  
\end{proof}

\section{Type Checking}\label{s:type checking}

In this section we discuss type checking in \punq{}, which corresponds to the following problem.
\begin{definition}[Type checking] We define the problem of type checking in $\punq$, denoted $\CHK{}$, as follows. Given $\vec{t}\in \SUP_c$ and $A\in\mathbb{T}_c$, $\CHK(\vec{t},A)$ is true if and only if $;\vdash\vec{t}:A$ can be derived.
\end{definition}

\paragraph*{Preliminaries.} Let $\mathsf{PTIME}$ be the set of decision problems solvable by a deterministic Turing machine running in polynomial time and, given an oracle $\mathcal{C}$, let us denote by $\mathsf{PTIME}^\mathcal{C}$ the set of decision problems that can be determined by a polytime deterministic oracle Turing machine with oracle $\mathcal{C}$, which can be easily extended to the case where $\mathcal{C}$ is a complexity class. Finally, the class $\Pi^0_1$ of the arithmetical hierarchy is defined as
\(
\Pi_1^0 \triangleq \Big\{\psi\mid \exists \phi \in \mathsf{REC}, \forall \overline{x}.(\psi(\overline{x})\Leftrightarrow \forall \overline{y}.\phi(\overline{x},\overline{y}))\Big\},
\)
where $\overline{x},\overline{y}$ are variables with values belonging to enumerable sets, $\psi$ and $\phi$ are formulas in the language of first-order arithmetic, and $\mathsf{REC}$ is the class of recursive sets (decidable problems)~\cite{O92}.

\paragraph*{Orthogonality criteria}

Outside of checking orthogonality in rules $(\mathsf{if}_\sharp)$ and $(\sharp_i)$, \punq{} type checking is no more complicated than that of \dlal{}, which is decidable in polynomial time~\cite{BT09}. This is because the complexity added by the \ls{} typing discipline (modulo orthogonality checking) corresponds to that of a simply typed lambda calculus, where type checking can be done efficiently~\cite{M04}. Therefore, if one has access to an oracle for checking orthogonality between terms, then type checking in \punq{} can be done efficiently.

\begin{definition} Let $\varphi$ be an \emph{oracle for orthogonality} between \punq{} terms, where, for all $\vec{t},\vec{s}\in \SUP$, and contexts $\Gamma,\Delta$, we have $\varphi(\vec{t},\vec{s},\Gamma,\Delta)\triangleq 0\iff \vec{t}\perp^{\Gamma;\Delta} \vec{s}$.
\end{definition}

\begin{restatable}[]{lemma}{lemmaOrthoOracle}\label{lemma:ortho-oracle}
$\CHK{}\in \mathsf{PTIME}^{\varphi}$.
\end{restatable}

We will address the complexity of type checking in three fragments of the \punq{} language, defined by their orthogonality criteria. The first two fragments only require access to the syntax of the term, whereas the last one takes in as input a ``candidate type'' used to bound the time needed to check orthogonality. Notice that, since we only check orthogonality in cases where we also check the type -- rules $(\mathsf{if}_\sharp)$ and $(\sharp_i)$, this does not lead to a loss of generality.

\begin{definition}
We define the following three criteria for orthogonality over superpositions.
\begin{itemize}
\item $\ortho_\emptyset: \SUP\times \SUP\to\{0,1\}$ defined as
\[\ortho_\emptyset(\vec{t},\vec{s})\triangleq \begin{cases} 0, & \textnormal{ if }\vec{t}\perp \vec{s},\\ 1, & \textnormal{otherwise.}\end{cases}\]
\item $\ortho_\times: \SUP\times \SUP\to\{0,1\}$ defined as 
\[\ortho_\times(\vec{t},\vec{s})\triangleq \begin{cases} \ortho_\times(t_1,s_1)\cdot \ortho_\times(t_2,s_2),&\text{if }\vec{t}=\pair{t_1}{t_2}\text{ and }\vec{s}=\pair{s_1}{s_2},\\ 
\ortho_\emptyset(\vec{t},\vec{s}), & \textnormal{otherwise.}\end{cases}\]
\item $\ortho_\mathsf{untyped}: \SUP\times \SUP\times \mathbb{T}\to\{0,1\}$~defined as follows. Let $FV(\vec{t})\cup FV(\vec{s})=x_1,\dots,x_n$.
\[\ortho_{\mathsf{untyped}}(\vec{t},\vec{s},Q)\triangleq 1-\prod\nolimits_{\vec{v}_1,\dots \vec{v}_n\in \V_c} \mathsf{PERP}\Big(\vec{t}\langle \vec{v}_1/x_1,\dots, \vec{v}_n/x_n\rangle,\  \vec{s}\langle \vec{v}_1/x_1,\dots, \vec{v}_n/x_n\rangle,Q\Big),\]
with $\mathsf{PERP}:\SUP_c\times\SUP_c\times \mathbb{T} \to \{0,1\}$ defined as
$$\mathsf{PERP}(\vec{t}_1,\vec{t}_2,Q)\triangleq 1 \iff \vec{t}_1,\vec{t}_2\rightsquigarrow^* \vec{w}_1,\vec{w}_2\in\llbracket Q\rrbracket_\emptyset\textnormal{ and }\ortho_\emptyset(\vec{w}_1,\vec{w}_2)=0.$$
\end{itemize}
For $\mathcal{X}\in\{\emptyset,\times,\mathsf{untyped}\}$, we denote by $\CHK_\mathcal{X}$ set of terms that can be typed in $\punq{}$ using $\ortho_\mathcal{X}$ as definition of orthogonality.
\end{definition}

In $\ortho_\emptyset$, orthogonality is defined only over closed terms. Notice, for instance, that the Hadamard and $Z$ gates (Examples~\ref{ex:hadamard}~and~\ref{ex:z}) only employ this orthogonality, and therefore $\mathsf{H},\mathsf{Z}\in\CHK{}_\emptyset$.
In $\ortho_\times$, we only detect orthogonality via closed terms in a pair structure, for instance by checking that $\pair{\ket{0}\,}{x}$ is orthogonal to $\pair{\ket{1}\,}{y}$ over any possible variable substitutions on $x$ and $y$. This is used in the case of the $CNOT$ gate (Example~\ref{ex:controlled-not}) and the circuit for teleportation which employs it (Example~\ref{ex:teleportation}), and so $\mathsf{CNOT},\mathsf{telep}\in\CHK_\times$. In $\ortho_\mathsf{untyped}$, we check over all possible value substitutions -- even those where the value used does not have the correct type -- check if the term reduces into the correct type and, if so, test its orthogonality.

\begin{restatable}[]{lemma}{lemmachecksincluded}\label{lemma:checks-included}
$\CHK_\emptyset\subseteq \CHK_\times \subseteq \CHK_\mathsf{untyped}\subseteq \CHK$.
\end{restatable}

%
%

$\ortho_\emptyset$ and $\ortho_\times$ can be verified efficiently, and therefore ensure that type checking can be done in polytime. However, we conjecture that general orthogonality is undecidable.

\begin{restatable}[]{lemma}{lemmaorthoscomplexity}\label{lemma:orthos-complexity}
$\ortho_\emptyset,\ortho_\times\in\mathsf{PTIME}$, and $\ortho_\mathsf{untyped}\in \Pi^0_1$.
\end{restatable}
\begin{proof}
Proof given in Appendix~\ref{app:type-checking}.
\end{proof}

The proof of Lemma~\ref{lemma:orthos-complexity} rests also on the choice of $\Ct$ for the amplitudes in \punq{}, where comparisons can be done efficiently~\cite[Proposition 2.2]{HHHJ05}. Furthermore, since type checking is done classically, we must bound the complexity of reducing terms to check their orthogonality in the classical Turing machine model (see Appendix~\ref{app:type-checking}).

\begin{restatable}[]{theorem}{thmchkcomplexity}\label{thm:chk-complexity}
$\CHK{}_\emptyset, \CHK{}_\times\in \mathsf{PTIME}$ and $\CHK{}_\mathsf{untyped}\in\mathsf{PTIME}^{\Pi_1^0}$.
\end{restatable}
\begin{proof}
From Lemmas~\ref{lemma:ortho-oracle} and~\ref{lemma:orthos-complexity}.
\end{proof}

\section{Conclusions and Future Work}
We have introduced $\punq$, a new quantum programming language with quantum
control capabilities.  
$\punq$ stands out as the first feasible
language with quantum control. Supporting higher-order programs, the language boasts a robust type system ensuring both unitarity
and polynomial-time normalization.
To conclude, we claim $\punq$ is a pioneering language, seamlessly combining
theoretical depth with practical applicability in the realm of quantum
computing. Plenty of interesting questions remain open:
\begin{itemize}
\item \textit{Extending quantum control}. So far, in \punq{}, we have limited superpositions to ground types, i.e., non functional types such as qubits. Extending superpositions to functional types, e.g., $\sharp(A\multimap B)$ or $\sharp (A\Rightarrow B)$, would allow for natural characterizations of experiments in indefinite causal order such as the quantum switch~\cite{QuantumSwitch15,QuantumSwitch17}. This question is also interesting for polymorphism, since the lists of type $A$ in \dlal{}, $\forall X.(A\multimap X\multimap X)\Rightarrow \S(X\multimap X)$, cannot be superposed in \punq{}.
\item \textit{Orthogonality criteria}. Extending the tractable criteria for orthogonality given in Section~\ref{s:type checking} to allow for more expressivity in the oracle-free case.
\item \textit{Circuit compilation}. Limiting quantum programs to polynomial time in a quantum Turing machine corresponds to a limitation to uniform quantum circuits~\cite{Y93}. However, direct compilation strategies that preserve polynomial time are not trivial~\cite{HPS23}. In this vein, it would be interesting to show a compilation strategy that would assign a polysize circuit representation of any \punq{} term that corresponds to a unitary transformation.
\end{itemize}

\bibliography{biblio}

\newpage

\appendix

\section{Proofs}

\subsection{Proofs of Section~\ref{s:type}}\label{app:SR}
\begin{lemma}[Substitution lemma]
  \label{lem:subs}
  For any two contexts $\Gamma,\Delta$, the following hold:
  \begin{enumerate}
    \item If $\Gamma,x:A;\Delta\vdash \vec t:B$ and $;[z:C]\vdash \vec u:A$ then $[z:C],\Gamma;\Delta \vdash \vec t[\vec u/x]:B$.
    \item If $\Gamma;\Delta_1,x: A \vdash \vec t:B$ and $;\Delta_2\vdash \vec u: A$ then $\Gamma;\Delta_1,\Delta_2\vdash \vec t[\vec u/x]:B$.
  \end{enumerate}
\end{lemma}
\begin{proof}
  By structural induction on $t$.

  For 1, consider:
  \begin{itemize}    
    \item If $\vec t=x$, then by a generation lemma we have that $\Delta=\emptyset$ and for some type $\vec t$ we have that $A=\S^{a_\S} T$ and $\sharp^b\S\{\sharp,\S\}^{a_\sharp,a_\S} T\leq B$, for $a_\sharp,a_\S,b\geqslant 0$. By another generation lemma we have that if $;[z:C]\vdash u:\S^{a_\S}T$ then we also have $;[z:C]\vdash u:\{\sharp,\S\}^{a_\sharp,a_\S}T$ and by $\S_i$ and $\sharp_i$ and finally by $\leq$ we obtain $[z;C];\vdash u:B$, after which context $\Gamma$ can be added via $\mathsf{W}$.

    \item Let $\vec t = \lambda y. \vec s$. Then we consider two possible cases:
      \begin{itemize}
	\item  $\Gamma,x:A;y:C,\Delta \vdash \vec s:D$ and $C\multimap D \leq B$. By induction hypothesis we have that $[z:C],\Gamma;y:C,\Delta \vdash \vec s[u/x]: D$ and by rule $\multimap_i$ we obtain $[z:C],\Gamma;\Delta\vdash \lambda y.\vec s[u/x]:C\multimap D$ as intended. Notice in this case that $\vec t[y/x]= (\lambda y.\vec s)[u/x] = \lambda y.(\vec s[u/x])$.
	\item $\Gamma,x:A,y:C;\Delta \vdash \vec s:D$ and $C\Rightarrow D\leq B$. Similarly to the first case, by applying the induction hypothesis and rule $\Rightarrow_i$ we obtain $[z:C],\Gamma;\Delta \vdash \lambda y.\vec s[u/x]:C\Rightarrow D$.
      \end{itemize} 

    \item $\vec t=s\ r$. Let us consider two cases: 
      \begin{itemize}
	\item  $\Gamma_1;\Delta_1 \vdash s:C\multimap D$ and $\Gamma_2;\Delta_2 \vdash r:C$, with $D \leq B$, where the contexts are such that $(\Gamma_1,\Gamma_2)=(\Gamma,x:A)$ and $(\Delta_1,\Delta_2)=\Delta$. Let us take the case where $\Gamma_1=(x:A,\Gamma_1')$. Then by induction hypothesis we have that $\Gamma_1';\Delta_1\vdash s[u/x]:C\multimap D$. By rule $\multimap_i$ we can obtain $\Gamma_1',\Gamma_2;\Delta_1,\Delta_2\vdash s[u/x]\vec{r}:D$, as desired. Notice that $(s\vec{r})[u/x]=s[u/x]\vec{r}$ in this case. The case where $x:A\in \Gamma_2$ is analogous.
	\item  $\Gamma_1;\Delta_1 \vdash s:C\Rightarrow D$ and $\Gamma_2;\Delta_2 \vdash r:C$, with $D \leq B$, where the contexts are such that $(\Gamma_1,\Gamma_2)=(\Gamma,x:A)$ and $(\Delta_1,\Delta_2)=\Delta$. This can be treated similarly to the previous case.
      \end{itemize} 

    \item  Let $\vec t = \tif{s}{\vec p_1}{\vec p_2}$. Consider two cases:

      \begin{itemize}
	\item $\Gamma_1;\Delta_1 \vdash s:\B$ and $\Gamma_2;\Delta_2\vdash (\vec p_1 \bot \vec p_2): C$ such that $C\leq B$, with $(\Gamma_1,\Gamma_2)=(\Gamma,x:A)$. If $x:A\in \Gamma_1$ then applying the induction hypothesis on $s$ gives the desired result. If $\Gamma_2 = (x:A,\Gamma_2')$, then by induction hypothesis, we have that $\Gamma_2';\Delta_2 \vdash (\vec p[u/x] \bot \vec p[u/x]):C$. Notice that $\Gamma_2';\Delta_2 \vdash \vec p[u/x] \bot \vec p[u/x]$ by definition of orthogonality for open superpositions: that they remain orthogonal for \textit{any} substitution of open variables by closed superpositions. In particular this means that they remain orthogonal in partial substitutions by open superpositions, since if not we could find a total substitution that would be non-orthogonal. We can apply rule $\mathsf{if}$ and obtain $\Gamma_1,\Gamma_2';\Delta_1,\Delta_2 \vdash \tif{s}{\vec p_1[u/x]}{\vec p_2[u/x]}$ as intended.

	\item  Case where $\Gamma_1;\Delta_1 \vdash s:\sharp \B$ is analogous.
      \end{itemize} 

    \item Let $\vec t = (t_1,t_2)$, in which case $\Gamma_1;\Delta_1 \vdash t_1 : B_1$ and $\Gamma_2;\Delta_2 \vdash t_2:B_2$ with $(\Gamma_1,\Gamma_2) = (\Gamma,x:A)$ and $(\Delta_1,\Delta_2) = \Delta$ and $B_1\times B_2 \leq B$. Assume that $\Gamma_1 = (\Gamma'_1,x:A)$, then, by induction hypothesis, $\Gamma'_1;\Delta_1 \vdash t_1[u/x]:B_1$ and so by rule $\times_i$ we have that $\Gamma_1',\Gamma_2;\Delta \vdash (t_1[u/x],t_2):B$, which is our desired result. The case where $\Gamma_2 = (\Gamma_2',x:A)$ is analogous.

    \item If $\vec t =(\tlet{x_1}{x_2}{r}{\vec s})$, then consider two options:
      \begin{itemize}
	\item  $\Gamma_1;\Delta_1 \vdash r: C\times D$ and $\Gamma_2,x_1:C,x_2:D;\Delta_2\vdash \vec s:E$ such that $E\leq B$, with $(\Gamma_1,\Gamma_2)=(\Gamma,x:A)$ and $(\Delta_1,\Delta_2)=\Delta$. Assume that $\Gamma_1 =( \Gamma_1', x:A)$. Then, by induction hypothesis, $\Gamma_1';\Delta_1 \vdash r[u/x]: C\times D$ and so by rule $\times_{e_1}$ we have that $\Gamma_1',\Gamma_2;\Delta\vdash \tlet{x_1}{x_2}{r[u/x]}{\vec s}:E$. Notice that in this case $(\Gamma_1',\Gamma_2)=\Gamma$ and that $\tlet{x_1}{x_2}{r[u/x]}{\vec s}$ is $(\tlet{x_1}{x_2}{r}{\vec s})[u/x]$. The case where $\Gamma_2 =(\Gamma_2',x:A)$ is analogous, and so is the case for rule $\times_{e_2}$. 
	\item  $\Gamma_1;\Delta_1 \vdash r: \sharp(C\times D)$ and $\Gamma_2,x_1:\sharp C,x_2:\sharp D;\Delta_2\vdash \vec s:E$ such that $\sharp E\leq B$. This case can also be treated via the induction hypothesis.
      \end{itemize}

  \end{itemize}

  For 2, consider:
  \begin{itemize}

    \item If $\vec t=x$, by a generation lemma we have that $\Delta_1=\emptyset$ and that for some type $\vec t$ we have $A=\S^{a_\S} T$ and $B=\{\sharp,\S\}^{a_\sharp,a_\S}T$. From $;\Delta_2\vdash u:\S^{a_\S}T$ we can also derive $;\Delta_2\vdash u:B$ by the fact that $\sharp_i$ does not alter the context. By $\mathsf{W}$ we obtain $\Gamma;\Delta_2\vdash u:B$ as desired.

    \item Let $\vec t = \lambda y. \vec s$. Then we consider two possible cases:
      \begin{itemize}
	\item  $\Gamma;x:A,y:C,\Delta_1 \vdash \vec s:D$ and $C\multimap D \leq B$. By induction hypothesis we have that $\Gamma;\Delta_1,y:C,\Delta_2 \vdash \vec s[u/x]: D$ and by rule $\multimap_i$ we obtain $\Gamma;\Delta_1,\Delta_2\vdash \lambda y.\vec s[u/x]:C\multimap D$ as intended. Notice in this case that $\vec t[y/x]= (\lambda y.\vec s)[u/x] = \lambda y.(\vec s[u/x])$.
	\item $\Gamma,y:C;x:A,\Delta_1 \vdash \vec s:D$ and $C\Rightarrow D\leq B$. Similarly to the first case, by applying the induction hypothesis and rule $\Rightarrow_i$ we obtain $\Gamma;\Delta_1,\Delta_2\vdash \lambda y.\vec s[u/x]:C\Rightarrow D$.
      \end{itemize}

    \item $\vec t=s\ r$. Let us consider two cases: 
      \begin{itemize}
	\item  $\Gamma_1;\Delta_{1,1} \vdash s:C\multimap D$ and $\Gamma_2;\Delta_{1,2} \vdash r:C$, with $D \leq B$, where the contexts are such that $(\Gamma_1,\Gamma_2)=\Gamma$ and $(\Delta_{1,1},\Delta_{1,2})=(\Delta_1,x:A)$. Let us take the case where $\Delta_{1,1}=(x:A,\Delta'_{1,1})$. Then by induction hypothesis we have that $\Gamma_1;\Delta'_{1,1}\vdash s[u/x]:C\multimap D$. By rule $\multimap_i$ we can obtain $\Gamma_1,\Gamma_2;\Delta_{1,1}',\Delta_{1,2}\vdash s[u/x]\vec{r}:D$, as desired. Notice that $(s\vec{r})[u/x]=s[u/x]\vec{r}$ in this case. The case where $x:A\in \Delta_{1,2}$ is analogous.
	\item  $\Gamma_1;\Delta_{1,1} \vdash s:C\Rightarrow D$ and $\Gamma_2;\Delta_2 \vdash r:C$, with $D \leq B$, can be treated similarly to the previous case.
      \end{itemize} 

    \item  Let $\vec t = \tif{s}{\vec p_1}{\vec p_2}$. Consider two cases:

      \begin{itemize}
	\item $\Gamma_1;\Delta_{1,1} \vdash s:\B$ and $\Gamma_2;\Delta_{1,2}\vdash (\vec p_1 \bot \vec p_2): C$ such that $C\leq B$, with $(\Gamma_1,\Gamma_2)=\Gamma$. If $x:A\in \Delta_{1,1}$ then applying the induction hypothesis on $s$ gives the desired result. If $\Delta_{1,2} = (x:A,\Delta_{1,2}')$, then by induction hypothesis, we have that $\Gamma_2;\Delta_{1,2}' \vdash (\vec p[u/x] \bot \vec p[u/x]):C$. We can apply rule $\mathsf{if}$ and obtain $\Gamma_1,\Gamma_2';\Delta_1,\Delta_2 \vdash \tif{s}{\vec p_1[u/x]}{\vec p_2[u/x]}$ as intended.

	\item  Case where $\Gamma_1;\Delta_{1,1} \vdash s:\sharp \B$ is analogous.
      \end{itemize} 

    \item Let $\vec t = (t_1,t_2)$, in which case $\Gamma_1;\Delta_{1,1} \vdash t_1 : B_1$ and $\Gamma_2;\Delta_2 \vdash t_2:B_2$ with $(\Gamma_1,\Gamma_2) = \Gamma $ and $(\Delta_{1,1},\Delta_{1,2}) = (\Delta_1,x:A)$ and $B_1\times B_2 \leq B$. Assume that $\Delta_{1,1} = (\Delta'_{1,1},x:A)$, then, by induction hypothesis, $\Gamma_1;\Delta'_{1,1} \vdash t_1[u/x]:B_1$ and so by rule $\times_i$ we have that $\Gamma;\Delta'_{1,1},\Delta_{1,2} \vdash (t_1[u/x],t_2):B$, which is our desired result. The case where $\Delta_{1,2} = (\Delta_{1,2}',x:A)$ is analogous.

    \item If $\vec t =(\tlet{x_1}{x_2}{r}{\vec s})$, then consider two options:
      \begin{itemize}
	\item  $\Gamma_1;\Delta_{1,1} \vdash r: C\times D$ and $\Gamma_2,x_1:C,x_2:D;\Delta_{1,2}\vdash \vec s:E$ such that $E\leq B$, with $(\Gamma_1,\Gamma_2)=\Gamma$ and $(\Delta_{1,1},\Delta_{1,2})=\Delta_1$. Assume that $\Delta_{1,1} =(\Delta_{1,1}', x:A)$. Then, by induction hypothesis, $\Gamma_1;\Delta_{1,1}' \vdash r[u/x]: C\times D$ and so by rule $\times_{e_1}$ we have that $\Gamma;\Delta'_{1,1},\Delta_{1,2}\vdash \tlet{x_1}{x_2}{r[u/x]}{\vec s}:E$. Notice that in this case $\tlet{x_1}{x_2}{r[u/x]}{\vec s}$ is $(\tlet{x_1}{x_2}{r}{\vec s})[u/x]$. The case where $\Delta_{1,2} =(\Delta_{1,2}',x:A)$ is analogous, and so is the case for rule $\times_{e_2}$. 
	\item  $\Gamma_1;\Delta_1 \vdash r: \sharp(C\times D)$ and $\Gamma_2,x_1:\sharp C,x_2:\sharp D;\Delta_2\vdash \vec s:E$ such that $\sharp E\leq B$. This case can also be treated via the induction hypothesis.  
      \end{itemize}
  \end{itemize}
\end{proof}

\thmsr*

\begin{proof}
  By induction on the relation $\rightsquigarrow$.
  \begin{itemize}
    \item $\vec t$ is a value. Then, since $\vec t$ does not reduce, there is no $\vec r$ satisfying the condition and the theorem is trivially satisfied. In the next items assume $v$ is a value, if not then apply induction hypothesis to $v$.
    \item $\vec t = (\lambda x.\vec{p})\ v$ and therefore $\vec r = \vec{p}[v/x]$. Then, from $; \vdash \vec t: A $ we consider two cases:
      \begin{itemize}
	\item $;\vdash \lambda x.\vec p : B\Rightarrow A$, with $x:B;\vdash \vec p: A$ and $;\vdash v: B$. Then by the substitution lemma we have that $;\vdash \vec{p}[v/x]:A$ which is the desired conclusion.
	\item $; \vdash \lambda x.\vec p : B\multimap A$, with $;x:B \vdash \vec p:A$ and $;\vdash v: B$. Substitution lemma provides the desired result.
      \end{itemize}

    \item $\vec t = \tif{v}{\vec p}{\vec s} :A$.

      If $;\vdash \vec{t}$ then rule $\mathsf{if}$ implies that $\vec{p} \perp_A \vec{s}$. There are two cases:
      \begin{enumerate}
	\item $v=\ket{0}$, and therefore $t \rightsquigarrow \vec p$ and the conclusion follows,
	\item $v=\ket{1}$, and $t \rightsquigarrow \vec s$ and the conclusion also follows.
      \end{enumerate}
      Analogous case where we used rule $\mathsf{if}_\sharp$.                                                            

    \item $\vec t=\tlet{x}{y}{\pair{v}{w}}{\vec s}$, then consider 3 cases:
      \begin{enumerate}
	\item We used $\times_{e_1}$ and so from $; \vdash \vec t : A$ we have that $;\vdash \pair{v}{w} :B \times C$ and $;x:B,y:C\vdash \vec{s}:A$. This means that $;\vdash v:B$ and $;\vdash w:C$ and therefore by applying the substitution lemma twice we obtain $;\vdash \vec{s}[v/x,w/y]: A$ as desired.
	\item Similar case to $\times_{e_2}$.
	\item Similar case for $\times_{e_\sharp}$.
	  
      \end{enumerate}
  \end{itemize}
\end{proof}

\thmprogress*

\begin{proof}
By induction on the structure of $\vec{t}$.
\begin{itemize}
\item $\vec{t}=x$. Cannot be typed with a closed context.
\item $\vec{t}=\ket{0}\,$. This is a value.
\item $\vec{t}=\ket{1}\,$. This is a value.
\item $\vec{t}=\tif{r}{\vec{s}_0}{\vec{s}_1}$. Since $r$ must be a boolean, either $r\rightsquigarrow \vec r'$ and therefore $\vec{t}\rightsquigarrow \tif{\vec r'}{\vec{s}_0}{\vec{s}_1}$ by rule (If$_+$), or by induction hypothesis $r$ must be a value and therefore either $r=\ket{0}$ or $r=\ket{1}\,$. In any of these two cases, $\vec{t}$ reduces, either by (If$_0$) or (If$_1$).
\item $\vec{t}=\lambda x.\vec{p}$ is a value.
\item $\vec{t}=s\ r$. We consider the three cases.
\begin{enumerate}
\item $r$ is not a value. By the induction hypothesis, $r\rightsquigarrow \vec r'$ and by rule (App), $\vec{t}\rightsquigarrow s\ \vec r'$.
\item $r$ is a value and $s$ is not. By the induction hypothesis, $s\rightsquigarrow \vec{s}'$ and by rule (App$_\V$), $\vec{t}\rightsquigarrow \vec{s}'\ r$.
\item Both $r$ and $s$ are values. Therefore $s=\lambda x.\vec{p}$ for some $\vec{p}$ and by rule (Abs) we have that $\vec{t}=(\lambda x.\vec{p})\ r\rightsquigarrow \vec{p}[s/x]$.
\end{enumerate}
\item $\vec{t}=\pair{s}{r}$. We consider the two possible cases where $\vec{t}$ is not a value:
\begin{enumerate}
\item $s$ is a value and $r$ is not, in which case by the induction hypothesis $r\rightsquigarrow \vec r'$ and therefore by rule (Pair$_\V$) we have that $\vec{t}\rightsquigarrow \pair{s}{\vec r'} $. 
\item $s$ is not a value, in which case by the induction hypothesis $s\rightsquigarrow \vec s'$ and therefore by rule (Pair) we have that $\vec{t}\rightsquigarrow \pair{\vec{s}'}{r} $. 
\end{enumerate}
\item $\vec{t}=\tlet{x}{y}{s}{\vec{r}}$. Consider two cases:
\begin{enumerate}
\item $s$ is a value of the form $\pair{v}{w}$, in which case by rule (Let), we have $\vec{t}\rightsquigarrow \vec{r}[v/x,w/y]$.
\item $s$ is not a value and therefore by the induction hypothesis $s\rightsquigarrow \vec s'$ and we have that $\vec{t}\rightsquigarrow \tlet{x}{y}{\vec{s}'}{\vec{r}}$.
\end{enumerate}
\item $\vec{t}=\vec{0}$. This is a value.
\item $\vec{t}=\alpha\cdot \vec{s}$, in which case either $\vec{s}$ is a value and therefore $\vec{t}$ is also a value, or $\vec{s}\rightsquigarrow \vec{s}'$ and by rule (Sup), $\vec{t}\rightsquigarrow \alpha\cdot \vec{s}'$.
\item $\vec{t}=\vec{s}+\vec{r}$, and by rule (Sup) if $\vec{s}$ or $\vec{r}$ are not values then $\vec{t}$ also reduces.
\end{itemize}
\end{proof}

\subsection{Proofs of Section~\ref{s:unitarity}}\label{app:unitarity}
\thmbangflat*
\begin{proof}
  First we state the following straightforward properties.
  \begin{enumerate}
    \item\label{prop:0} For any set $S\subseteq \BV_c$, we have $S=\flat S$.
    \item\label{prop:1} For any set $S\subseteq \V_c$, we have $\flat S = \flat\,\mathsf{span}(S)$.
    \item\label{prop:2} For any set $S\subseteq \BV_c$, we have $S = S\cap \mathcal S_1$.
  \end{enumerate}

  Then we proceed by induction on the structure of simple types.
  \begin{itemize}
    \item If $A=\B$, then $!(A)=!(\B)=\B$. Hence
      \[
	\llbracket {!}(A)\rrbracket_\emptyset
	=\llbracket \B\rrbracket_\emptyset
	=^{\textrm{\eqref{prop:0}}}\flat\llbracket \B\rrbracket_\emptyset
	=\flat \llbracket A\rrbracket_\emptyset
      \]
    \item If $A=B\multimap C$, then $!(A)={!}(B\multimap C)=B\multimap C$. Hence,
      \[
	\llbracket {!}(A)\rrbracket_\emptyset
	=\llbracket {!}(B\multimap C)\rrbracket_\emptyset
	=\llbracket B\multimap C\rrbracket_\emptyset
	=^{\textrm{\eqref{prop:0}}} \flat\llbracket B\multimap C\rrbracket_\emptyset
	=\flat\llbracket A\rrbracket_\emptyset
      \]
    \item If $A=B\Rightarrow C$, then $!(A)={!}(B\Rightarrow C)=B\Rightarrow C$. Hence,
      \[
	\llbracket {!}(A)\rrbracket_\emptyset
	=\llbracket {!}(B\Rightarrow C)\rrbracket_\emptyset
	=\llbracket B\Rightarrow C\rrbracket_\emptyset
	=^\textrm{\eqref{prop:0}} \flat\llbracket B\Rightarrow C\rrbracket_\emptyset
	=\flat\llbracket A\rrbracket_\emptyset
      \]
    \item If $A=B\times C$, then $!(A)={!}(B)\times {!}(C)$. Hence,
      \begin{align*}
	\llbracket {!}(A)\rrbracket_\emptyset
	&=\llbracket {!}(B)\times {!}(C)\rrbracket_\emptyset
	=\{(\vec v,\vec w):\vec v\in\llbracket{!}(B)\rrbracket_\emptyset, \vec w\in\llbracket{!}(C)\rrbracket_\emptyset\}\\
	&=^{\textrm{(IH)}} \{(\vec v,\vec w):\vec v\in \flat\llbracket B\rrbracket_\emptyset, \vec w\in \flat\llbracket C\rrbracket_\emptyset\}=\flat\{(\vec v,\vec w):\vec v\in \llbracket B\rrbracket_\emptyset, \vec w\in \llbracket C\rrbracket_\emptyset\}\\
	&=\flat(\llbracket B\times C\rrbracket_\emptyset) =\flat\llbracket A\rrbracket_\emptyset
      \end{align*}
    \item If $A=\sharp B$, then $!(A)=!(\sharp B)  = !(B)$. 
      Hence,
      \begin{align*}
	\llbracket !(A)\rrbracket_\emptyset 
	& =\llbracket !(B)\rrbracket_\emptyset
	=^{\textrm{(IH)}} \flat\llbracket B\rrbracket_\emptyset
	=^{\textrm{\eqref{prop:1}}} \flat\mathsf{span}(\llbracket B\rrbracket_\emptyset)
	=^{\textrm{\eqref{prop:2}}}\flat\mathsf{span}(\llbracket B\rrbracket_\emptyset)\cap\mathcal S_1\\
	&=^{\textrm{\eqref{prop:2}}}\flat(\mathsf{span}(\llbracket B\rrbracket_\emptyset)\cap\mathcal S_1)
	=\flat\llbracket \sharp B\rrbracket_\emptyset
	=\flat\llbracket A\rrbracket_\emptyset
      \end{align*}

    \item If $A=\S B$, then $!(A) = !(\S B) = \S !(B)$. Hence,
      \[
	\llbracket !(A)\rrbracket_\emptyset 
	=\llbracket \S !(B)\rrbracket_\emptyset
	=\llbracket !(B)\rrbracket_\emptyset
	=^{\textrm{(IH)}} \flat\llbracket B\rrbracket_\emptyset
	=\flat\llbracket \S B\rrbracket_\emptyset
	=\flat\llbracket A\rrbracket_\emptyset
      \]
      
  \end{itemize}
\end{proof}

\lemsubstitutionRealizability*
\begin{proof}
  By induction on $A$.
  \begin{itemize}
    \item If $A=X$, then $A[B/X] = B$. Hence,
      \[
	\llbracket A[B/X]\rrbracket_\tau
	= \llbracket B\rrbracket_\tau
	= \llbracket X\rrbracket_{\tau\cup\{X\mapsto \llbracket B\rrbracket_\tau\}}
	= \llbracket A\rrbracket_{\tau\cup\{X\mapsto \llbracket B\rrbracket_\tau\}}
      \]
    \item If $A=Y\neq X$, then $A[B/X] = Y$. Hence,
      \[
	\llbracket A[B/X]\rrbracket_\tau
	= \llbracket Y\rrbracket_\tau
	= \tau(Y)
	=(\tau\cup\{X\mapsto \llbracket B\rrbracket_\tau\})(Y)
	= \llbracket Y\rrbracket_{\tau\cup\{X\mapsto \llbracket B\rrbracket_\tau\}}
	= \llbracket A\rrbracket_{\tau\cup\{X\mapsto \llbracket B\rrbracket_\tau\}}
      \]

    \item If $A=\B$, then $A[B/X] = \B$. Hence,
      \[
	\llbracket A[B/X]\rrbracket_\tau
	= \llbracket \B\rrbracket_\tau
	= \{\ket 0,\ket 1\}
	= \llbracket \B\rrbracket_{\tau\cup\{X\mapsto \llbracket B\rrbracket_\tau\}}
	= \llbracket A\rrbracket_{\tau\cup\{X\mapsto \llbracket B\rrbracket_\tau\}}
      \]
    \item If $A=C\multimap D$, then $A[B/X] = C[B/X]\multimap D[B/X]$. Hence,
      \begin{align*}
	\llbracket A[B/X]\rrbracket_\tau
	&=
	\llbracket C[B/X]\multimap D[B/X]\rrbracket_\tau
	=
	\{
	  \lambda x.\vec t:\vec v\in \llbracket C[B/X]\rrbracket_\tau,
	  \vec t\langle \vec v/x\rangle\Vdash\llbracket D[B/X]\rrbracket_\tau
	\}
	\\
	&=^{\textrm{(IH)}}
	\{
	  \lambda x.\vec t:\vec v\in \llbracket C\rrbracket_{\tau\cup\{X\mapsto \llbracket B\rrbracket_\tau\}},
	  \vec t\langle \vec v/x\rangle\Vdash\llbracket D\rrbracket_{\tau\cup\{X\mapsto \llbracket B\rrbracket_\tau\}}
	\}
	\\
	&=\llbracket C\multimap D\rrbracket_{\tau\cup\{X\mapsto \llbracket B\rrbracket_\tau\}}
	= \llbracket A\rrbracket_{\tau\cup\{X\mapsto \llbracket B\rrbracket_\tau\}}
      \end{align*}
    \item If $A=C\Rightarrow D$, then $A[B/X] = C[B/X]\Rightarrow D[B/X]$. Hence,
      \begin{align*}
	\llbracket A[B/X]\rrbracket_\tau
	&=
	\llbracket C[B/X]\Rightarrow D[B/X]\rrbracket_\tau
	=
	\{
	  \lambda x.\vec t:\vec v\in \flat\llbracket C[B/X]\rrbracket_\tau,
	  \vec t[\vec v/x]\Vdash\llbracket D[B/X]\rrbracket_\tau
	\}
	\\
	&=^{\textrm{(IH)}}
	\{
	  \lambda x.\vec t:\vec v\in \flat\llbracket C\rrbracket_{\tau\cup\{X\mapsto \llbracket B\rrbracket_\tau\}},
	  \vec t[\vec v/x]\Vdash\llbracket D\rrbracket_{\tau\cup\{X\mapsto \llbracket B\rrbracket_\tau\}}
	\}
	\\
	&=\llbracket C\Rightarrow D\rrbracket_{\tau\cup\{X\mapsto \llbracket B\rrbracket_\tau\}}
	= \llbracket A\rrbracket_{\tau\cup\{X\mapsto \llbracket B\rrbracket_\tau\}}
      \end{align*}
    \item If $A=C\times D$, then $A[B/X]=C[B/X]\times D[B/X]$. Hence,
      \begin{align*}
	\llbracket A[B/X]\rrbracket_\tau
	&=
	\llbracket C[B/X]\times D[B/X]\rrbracket_\tau
	=
	\{
	  (\vec v,\vec w):
	  \vec v\in \llbracket C[B/X]\rrbracket_\tau,
	  \vec w\in \llbracket D[B/X]\rrbracket_\tau
	\}
	\\
	&=^{\textrm{(IH)}}
	\{
	  (\vec v,\vec w):
	  \vec v\in \llbracket C\rrbracket_{\tau\cup\{X\mapsto \llbracket B\rrbracket_\tau\}},
	  \vec w\in \llbracket D\rrbracket_{\tau\cup\{X\mapsto \llbracket B\rrbracket_\tau\}},
	\}
	\\
	&=\llbracket C\times D\rrbracket_{\tau\cup\{X\mapsto \llbracket B\rrbracket_\tau\}}
	= \llbracket A\rrbracket_{\tau\cup\{X\mapsto \llbracket B\rrbracket_\tau\}}
      \end{align*}
    \item If $A=\sharp C$, then $A[B/X]=\sharp C[B/X]$. Hence,
      \begin{align*}
	\llbracket A[B/X]\rrbracket_\tau
	&=
	\llbracket\sharp C[B/X]\rrbracket_\tau
	=
	\mathsf{span}(\llbracket C[B/X]\rrbracket_\tau)\cap\mathcal S_1
	\\
	&=^{\textrm{(IH)}}
	\mathsf{span}(\llbracket C\rrbracket_{\tau\cup\{X\mapsto \llbracket B\rrbracket_\tau\}})\cap\mathcal S_1
	=\llbracket\sharp C\rrbracket_{\tau\cup\{X\mapsto \llbracket B\rrbracket_\tau\}}
	= \llbracket A\rrbracket_{\tau\cup\{X\mapsto \llbracket B\rrbracket_\tau\}}
      \end{align*}
    \item If $A=\S C$, then $A[B/X]=\S C[B/X]$. Hence,
      \begin{align*}
	\llbracket A[B/X]\rrbracket_\tau
	&=
	\llbracket\S C[B/X]\rrbracket_\tau
	=
	\llbracket C[B/X]\rrbracket_\tau
	\\
	&=^{\textrm{(IH)}}
	\llbracket C\rrbracket_{\tau\cup\{X\mapsto \llbracket B\rrbracket_\tau\}}
	=\llbracket\S C\rrbracket_{\tau\cup\{X\mapsto \llbracket B\rrbracket_\tau\}}
	= \llbracket A\rrbracket_{\tau\cup\{X\mapsto \llbracket B\rrbracket_\tau\}}
      \end{align*}
    \item If $A=\forall Y.C$, then $A[B/X] = \forall Y.C[B/X]$. Hence,
      \begin{align*}
	\llbracket A[B/X]\rrbracket_\tau
	&=
	\llbracket\forall Y.C[B/X]\rrbracket_\tau
	=
	\bigcap_{R\subseteq\mathcal S_1}\llbracket C[B/X]\rrbracket_{\tau\cup\{Y\mapsto R\}}
	\\
	&=^{\textrm{(IH)}}
	\bigcap_{R\subseteq\mathcal S_1}\llbracket C\rrbracket_{\tau\cup\{Y\mapsto R, X\mapsto \llbracket B\rrbracket_\tau\}}
	=\llbracket\forall Y.C\rrbracket_{\tau\cup\{X\mapsto \llbracket B\rrbracket_\tau\}}
	= \llbracket A\rrbracket_{\tau\cup\{X\mapsto \llbracket B\rrbracket_\tau\}}
      \end{align*}
      
  \end{itemize}
\end{proof}

\thmvalidity*
\begin{proof}
  ~
  \begin{itemize}
    \item[$(\mathsf{W})$] Assume that $\Gamma;\Delta\vdash \vec{t}: A$ is a
      valid typing judgment. Then $\dom(\Delta)\subseteq FV(\vec t)\subseteq
      \dom(\Gamma,\Delta)$, and $\vec t\langle \sigma\rangle \Vdash A$, for any
      $\sigma\in\llbracket \Gamma,\Delta\rrbracket_\tau$. Given a variable $x$
      with type $B$ such that $x\not\in \dom(\Gamma,\Delta)$, we have that
      $\dom(\Delta)\subseteq FV(\vec{t})\subseteq
      \dom(\Gamma,\Delta)\cup\{x\}$. Now, for any substitution $\sigma\in
      \llbracket \Gamma,\Delta,\{x:B\} \rrbracket_\tau$, we have that $\sigma =
      \sigma_0 \langle \vec{v}/x\rangle$ for some $\sigma_0\in\llbracket
      \Gamma,\Delta\rrbracket_\tau$ and $\vec{v}\Vdash B$. Then,
      \[
	\vec{t}\langle \sigma\rangle=\vec{t}[\vec{v}/x]\langle \sigma_0\rangle=\vec{t}\langle\sigma_0\rangle \Vdash A,
      \]
      since $\sigma_0\in \llbracket \Gamma,\Delta\rrbracket_\tau$ and $x\not
      \in FV(\vec t)$.

    \item[$(\mathsf{C})$] The premise being valid, this means that
      $\dom(\Delta)\subseteq FV(\vec t)\subseteq
      \dom(\Gamma,\Delta)\cup\{x,y\}$ and for all $\sigma\in\llbracket
      \Gamma,\Delta,\{x:B,y:B\}\rrbracket_\tau$ we have that $\vec{t}\langle
      \sigma \rangle\Vdash A$. It is clear that $FV(\vec{t}[x/y])\subseteq
      \dom(\Gamma,\Delta)\cup \{x\}$. For any $\sigma\in\llbracket
      \Gamma,\Delta,\{x:B\}\rrbracket_\tau$, $\sigma = \sigma_0 \langle
      \vec{v}/x\rangle$ for some $\sigma_0
      \in\llbracket\Gamma,\Delta\rrbracket_\tau$ and some value
      $\vec{v}\in\llbracket B\rrbracket_\tau$. Therefore,
      \begin{align*}
	(\vec t[x/y])\langle \sigma\rangle &= (\vec t  [x/y])\langle x/\vec{v},\sigma_0\rangle = (\vec t [x/y][\vec{v}/x])\langle \sigma_0\rangle\\
	&=(\vec t [\vec{v}/x][x/y])\langle \sigma_0\rangle = (\vec t [\vec{v}/x][\vec{v}/y])\langle \sigma_0\rangle\\
	& = \vec t \langle \vec{v}/x,\vec{v}/y,\sigma_0\rangle\Vdash A,
      \end{align*}
      since $\langle \vec{v}/x,\vec{v}/y,\sigma_0\rangle\in \llbracket \Gamma,\Delta,\{x:B,y:B\}\rrbracket_\tau.$

    \item[$(\equiv)$] By associativity of the substitution.

    \item[$(\leq)$] By definition of subtyping.

    \item[$(\mathsf{ax})$] We have that $\dom(\Delta)=\{x\}=FV(x)$ and
      $\dom(\Gamma)=\tau$. For any substitution $\sigma\in\llbracket
      x:A\rrbracket_\tau$, we have $\sigma =\{\vec{v}/x\}$ for some
      $\vec{v}\in\llbracket A\rrbracket_\tau$. Therefore $x\langle
      \sigma\rangle = x\langle \vec{v}/x\rangle = \vec{v}\Vdash \llbracket
      A\rrbracket_\tau$.

    \item[$(0)$] The rule is valid since $FV(\ket{0})=\emptyset$ and $\ket{0}\in \llbracket \B \rrbracket_\emptyset$ by definition.

    \item[$(1)$] Similar to $(0)$.

    \item[$(\mathsf{if})$] Supposing that the premises are valid, we obtain:
      \begin{enumerate}
	\item\label{eq:ift} $\dom(\Delta_1)\subseteq FV(t)\subseteq
	  \dom(\Gamma_1,\Delta_1)\quad$ and $\quad t\langle \sigma \rangle
	  \Vdash \llbracket \B \rrbracket_\emptyset $ for all $\sigma\in\llbracket
	  \Gamma_1,\Delta_1\rrbracket_\tau$.
	\item\label{eq:ifr} $\dom(\Delta_2)\subseteq FV(\vec{r})\subseteq
	  \dom(\Gamma_2,\Delta_2)\quad$ and $\quad \vec{r}\langle \sigma
	  \rangle \Vdash \llbracket A\rrbracket_\tau $ for all $\sigma\in\llbracket
	  \Gamma_2,\Delta_2\rrbracket_\tau$.
	\item\label{eq:ifs} $\dom(\Delta_2)\subseteq FV(\vec{s})\subseteq
	  \dom(\Gamma_2,\Delta_2)\quad$ and $\quad \vec{s}\langle \sigma
	  \rangle \Vdash \llbracket A \rrbracket_\tau$ for all $\sigma\in\llbracket
	  \Gamma_2,\Delta_2\rrbracket_\tau$.
      \end{enumerate}
      From \eqref{eq:ift}--\eqref{eq:ifs} we have that
      $\dom(\Delta_1,\Delta_2)\subseteq FV(\tif{t}{\vec r}{\vec s})\subseteq
      \dom(\Gamma_1,\Delta_1,\Gamma_2,\Delta_2)$. Given a $\sigma\in\llbracket
      \Gamma_1,\Delta_1,\Gamma_2,\Delta_2\rrbracket_\tau$, then $\sigma =
      \sigma_1\sigma_2$ for $\sigma_i\in\llbracket
      \Gamma_i,\Delta_i\rrbracket_\tau$, and we have that
      \begin{align*}
	&\ (\tif{t}{\vec{r}}{\vec{s}})\langle\sigma\rangle\\
	&=\ (\tif{t}{\vec{r}}{\vec{s}})\langle\sigma_1\rangle\langle\sigma_2\rangle\\
	&=\ \tif{t\langle\sigma_1\rangle}{\vec{r}\langle \sigma_2\rangle}{\vec{s}\langle \sigma_2\rangle}.
      \end{align*}
      Since $\sigma_1\in\llbracket \Gamma_1,\Delta_1\rrbracket_\tau$ we have
      that $t\langle \sigma_1\rangle \Vdash \llbracket \B \rrbracket_\emptyset$ from \eqref{eq:ift}. Therefore we can
      consider two cases:
      \begin{itemize}
	\item $t\langle \sigma_1\rangle \rightsquigarrow^* \ket{0}$ such that
	  \begin{align*}
	    &(\tif{t}{\vec{r}}{\vec{s}})\langle\sigma\rangle\\
	    =\ &\tif{t\langle\sigma_1\rangle}{\vec{r}\langle \sigma_2\rangle}{\vec{s}\langle \sigma_2\rangle} \\
	    \rightsquigarrow^*\ & \tif{\ket{0}}{\vec{r}\langle \sigma_2\rangle}{\vec{s}\langle \sigma_2\rangle} \\
	    \rightsquigarrow^*\ & \vec{r}\langle \sigma_2\rangle \Vdash \llbracket A \rrbracket_\tau
	  \end{align*}
	  using \eqref{eq:ifr} and the fact that $\sigma_2\in\llbracket
	  \Gamma_2,\Delta_2\rrbracket_\tau$. The second case is similar:
	\item $t\langle \sigma_1\rangle \rightsquigarrow^* \ket{1}$ such that
	  \begin{align*}
	    &(\tif{t}{\vec{r}}{\vec{s}})\langle\sigma\rangle\\
	    =\ &\tif{t\langle\sigma_1\rangle}{\vec{r}\langle \sigma_2\rangle}{\vec{s}\langle \sigma_2\rangle} \\
	    \rightsquigarrow^*\ & \tif{\ket{1}}{\vec{r}\langle \sigma_2\rangle}{\vec{s}\langle \sigma_2\rangle} \\
	    \rightsquigarrow^*\ & \vec{s}\langle \sigma_2\rangle \Vdash \llbracket A\rrbracket_\tau.
	  \end{align*}
      \end{itemize}

    \item[$(\mathsf{if}_\sharp)$] From the premises we deduce:
      \begin{enumerate}
	\item \label{eq:ifst}$\dom(\Delta_1)\subseteq FV(\vec t)\subseteq \dom(\Gamma_1,\Delta_1)\quad$ and $\quad \vec t\langle \sigma \rangle \Vdash \llbracket \sharp \B \rrbracket_\emptyset$ for all $\sigma\in\llbracket \Gamma_1,\Delta_1\rrbracket_\tau$.
	\item \label{eq:ifsr}$\dom(\Delta_2)\subseteq FV(\vec{r})\subseteq \dom(\Gamma_2,\Delta_2)\quad$ and $\quad \vec{r}\langle \sigma \rangle \Vdash \llbracket Q\rrbracket_\tau$ for all $\sigma\in\llbracket \Gamma_2,\Delta_2\rrbracket_\tau$.
	\item \label{eq:ifss}$\dom(\Delta_2)\subseteq FV(\vec{s})\subseteq \dom(\Gamma_2,\Delta_2)\quad$ and $\quad \vec{s}\langle \sigma \rangle \Vdash \llbracket Q\rrbracket_\tau$ for all $\sigma\in\llbracket \Gamma_2,\Delta_2\rrbracket_\tau$.
	\item \label{eq:ifsortho}For all $\sigma\in\llbracket \Gamma_2,\Delta_2\rrbracket_\tau$, we have that $\vec r\langle\sigma\rangle \rightsquigarrow^* \vec v$ and $\vec s\langle\sigma\rangle \rightsquigarrow^* \vec w$ such that $\langle \vec v ,\vec w \rangle =0$.
      \end{enumerate}

      Similarly to $(\mathsf{if})$, we have that $\dom(\Delta_1,\Delta_2)\subseteq FV(\tif{t}{\vec r}{\vec s})\subseteq \dom(\Gamma_1,\Delta_1,\Gamma_2,\Delta_2)$ and, given a $\sigma\in\llbracket \Gamma_1,\Delta_1,\Gamma_2,\Delta_2\rrbracket_\tau$, then $\sigma = \sigma_1\sigma_2$ for $\sigma_i\in\llbracket \Gamma_i,\Delta_i\rrbracket_\tau$, and we have that
      \begin{align*}
	&\ (\tif{t}{\vec{r}}{\vec{s}})\langle\sigma\rangle\\
	&=\ (\tif{t}{\vec{r}}{\vec{s}})\langle\sigma_1\rangle\langle\sigma_2\rangle\\
	&=\ \tif{t\langle\sigma_1\rangle}{\vec{r}\langle \sigma_2\rangle}{\vec{s}\langle \sigma_2\rangle}.
      \end{align*}
      From (\ref{eq:ifst}) and the fact that $\sigma_1\in\llbracket \Gamma_1,\Delta_1\rrbracket_\tau$ we deduce $t\langle \sigma_1\rangle \Vdash \sharp \B$ and therefore $t\langle \sigma_1\rangle \rightsquigarrow^* \alpha \cdot \ket{0}+\beta \cdot \ket{1}$, for $\alpha,\beta\in\Ct$ such that $|\alpha|^2+|\beta|^2=1$. Therefore,
      \begin{align*}
	&(\tif{t}{\vec{r}}{\vec{s}})\langle\sigma\rangle\\
	=\ &\tif{t\langle\sigma_1\rangle}{\vec{r}\langle \sigma_2\rangle}{\vec{s}\langle \sigma_2\rangle}\\
	\rightsquigarrow^* \ &\tif{(\alpha \cdot \ket{0}+\beta\cdot \ket{1})}{\vec{r}\langle \sigma_2\rangle}{\vec{s}\langle \sigma_2\rangle}\\
	=\ &\alpha \cdot(\tif{\ket{0}}{\vec{r}\langle \sigma_2\rangle}{\vec{s}\langle \sigma_2\rangle})+\beta \cdot(\tif{\ket{1}}{\vec{r}\langle \sigma_2\rangle}{\vec{s}\langle \sigma_2\rangle})\\
	\rightsquigarrow^*\ &\alpha \cdot \vec{r}\langle \sigma_2\rangle + \beta \cdot \vec{s}\langle \sigma_2\rangle\\
	\rightsquigarrow^*\ &\alpha \cdot \vec{v}+\beta\cdot \vec{w} \Vdash \llbracket \sharp Q \rrbracket_\tau
      \end{align*}
      derived from (\ref{eq:ifsr})--(\ref{eq:ifsortho}) and the fact that $\sigma_2\in\llbracket \Gamma_2,\Delta_2\rrbracket_\tau$.

\item[$(\multimap_i)$] Suppose that $\Gamma;\Delta,x:A\vdash \vec{t}:B$ is
  a valid judgment. Then since $(\dom(\Delta)\cup \{x\}) \subseteq
  FV(\vec{t})$ we have $\dom(\Delta)\subseteq FV(\lambda x.\vec{t})$.
  Similarly, since $FV(\vec{t})\subseteq (\dom(\Gamma,\Delta)\cup \{x\})$
  it is also true that $FV(\lambda x.\vec{t})\subseteq
  \dom(\Gamma,\Delta).$ For any $\sigma\in \llbracket
  \Gamma,\Delta\rrbracket_\tau$, we have that $(\lambda x.\vec{t})\langle
  \sigma\rangle = \lambda x.\vec{t}\langle \sigma \rangle$, since $x\not\in
  \dom(\sigma)$. For all $\vec{v}\in\llbracket A\rrbracket_\tau$, we
  observe that $(\vec{t}\langle \sigma\rangle)\langle \vec{v}/x\rangle =
  \vec{t}\langle \sigma,\{\vec{v}/x\}\rangle\Vdash \llbracket B
  \rrbracket_\tau $, since $\sigma,\{\vec{v}/x\}\in\llbracket
  \Gamma,x:A\rrbracket_\tau$. Therefore, $(\lambda
  x.\vec{t})\langle\sigma\rangle \Vdash  \llbracket A\multimap B
  \rrbracket_\tau$.

\item[$(\Rightarrow_i)$] Suppose that $\Gamma,x:A;\Delta\vdash t:B$ is a
  valid judgment. Then $\dom(\Delta)\subseteq FV(t)\subseteq
  \dom(\Gamma,\{x:A\},\Delta)$. We can therefore conclude that
  $\dom(\Delta)\subseteq FV(\lambda x.t)\subseteq \dom(\Gamma,\Delta)$.
  Given $\sigma\in\llbracket \Gamma,\Delta\rrbracket_\tau$ we want to show
  that $(\lambda x.t)\langle\sigma\rangle\Vdash \llbracket !(A)\Rightarrow
  B\rrbracket_\tau $. Since $x\not\in \dom(\sigma)$, we have $(\lambda
  x.t)\langle \sigma \rangle=\lambda x.(t\langle \sigma \rangle)$, and for
  all values $\vec{v}\in \llbracket A\rrbracket_\tau$, we have that 
  \[
    t\langle \sigma \rangle \langle \vec v/x\rangle= t\langle \sigma,\{\vec v/x\}\rangle \Vdash \llbracket B\rrbracket_\tau
  \]
  from the premise. Since $\flat \llbracket A\rrbracket_\tau\subseteq
  \llbracket A \rrbracket_\tau $, we have that $(\lambda x.t)\langle \sigma
  \rangle \Vdash \llbracket !(A)\Rightarrow B\rrbracket_\tau $. 

\item[$(\multimap_e)$] Suppose that both judgments $\Gamma_1;\Delta_1\vdash
  t:A\multimap B$ and $\Gamma_2;\Delta_2\vdash \vec{r}:A$ are valid,
  meaning that:
  \begin{itemize}
    \item $\dom(\Delta_1)\subseteq FV(t)\subseteq
      \dom(\Gamma_1,\Delta_1)\quad$ and $\quad t\langle \sigma \rangle
      \Vdash \llbracket A\multimap B\rrbracket_\tau $ for all
      $\sigma\in\llbracket \Gamma_1,\Delta_1\rrbracket_\tau$.
    \item $\dom(\Delta_2)\subseteq FV(\vec{r})\subseteq
      \dom(\Gamma_2,\Delta_2)\quad$ and $\quad\vec{r}\langle \sigma \rangle
      \Vdash \llbracket A\rrbracket_\tau$ for all $\sigma\in\llbracket
      \Gamma_2,\Delta_2\rrbracket_\tau$.
  \end{itemize}
  This implies that $\dom(\Delta_1,\Delta_2)\subseteq FV(t\
  \vec{r})\subseteq \dom(\Gamma_1,\Delta_1,\Gamma_2,\Delta_2)$. Given
  $\sigma\in\llbracket
  \Gamma_1,\Delta_1,\Gamma_2,\Delta_2\rrbracket_\tau$,then $\sigma=\sigma_1
  \sigma_2$ for $\sigma_i\in\llbracket \Gamma_i,\Delta_i\rrbracket_\tau$,
  and from $FV(t)\bigcap \dom(\sigma_2)=FV(\vec r)\bigcap
  \dom(\sigma_1)=\emptyset$, then
  \[
    (t\ \vec{r})\langle \sigma \rangle = (t\ \vec{r})\langle\sigma_1\rangle\langle\sigma_2\rangle=(t\langle \sigma_1\rangle\ \vec{r})\langle \sigma_2\rangle=t\langle \sigma_1\rangle\ \vec
    r\langle \sigma_2\rangle,
  \]
  and so we conclude that $(t\ \vec{r})\langle\sigma\rangle=t\langle
  \sigma_1\rangle\ \vec{r}\langle\sigma_2\rangle \Vdash \llbracket
  B\rrbracket_\tau,$ by the application of realizers~\cite[Lemma
  A.4]{DCGMV19}.

\item[($\Rightarrow_e$)] Assuming the premises are valid, we have that
  \begin{itemize}
    \item $\dom(\Delta)\subseteq FV(t)\subseteq \dom(\Gamma,\Delta)\quad$
      and $\quad t\langle \sigma\rangle \Vdash \llbracket !(A)\Rightarrow
      B\rrbracket_\tau $ for all $\sigma\in\llbracket
      \Gamma,\Delta\rrbracket_\tau$.
    \item $FV(r) = \{z\}\quad$ and $\quad r\langle \sigma\rangle \Vdash !A$
      for all $\sigma \in\llbracket z:C\rrbracket_\tau$.
  \end{itemize}
  We can directly conclude that $\dom(\Delta)\subseteq FV(t)\subseteq
  \dom(\Gamma,\Delta)\cup\{z\}.$ From the fact that $z\not \in
  \dom(\Gamma,\Delta)$, we have that for any $\sigma \in\llbracket
  \Gamma,\Delta,\{z:C\}\rrbracket_\tau$, $\sigma=\sigma_1\sigma_2$ for
  $\sigma_1\in\llbracket \Gamma,\Delta\rrbracket_\tau$ and
  $\sigma_2\in\llbracket \{z:C\}\rrbracket_\tau$. Therefore
  \[
    (t\ r)\langle \sigma \rangle = (t\ r)\langle \sigma_1\rangle\langle \sigma_2\rangle = t\langle \sigma_1\rangle\ r\langle \sigma_2\rangle \Vdash \llbracket  B \rrbracket_\tau,
  \]
  by the application of realizers.

\item[$(\times_i)$] Assuming that $\Gamma_1;\Delta_1\vdash \vec{t}: A$ and
  $\Gamma_2;\Delta_2\vdash \vec{r}: B$ are valid typing judgments, we have
  that
  \begin{enumerate}
    \item $\dom(\Delta_1)\subseteq FV(\vec t)\subseteq
      \dom(\Gamma_1,\Delta_1)\quad$ and $\quad \vec t\langle \sigma \rangle
      \Vdash \llbracket  A \rrbracket_\tau $ for all $\sigma\in\llbracket
      \Gamma_1,\Delta_1\rrbracket_\tau$.
    \item $\dom(\Delta_2)\subseteq FV(\vec{r})\subseteq
      \dom(\Gamma_2,\Delta_2)\quad$ and $\quad \vec{r}\langle \sigma
      \rangle \Vdash \llbracket B \rrbracket_\tau$ for all $\sigma\in\llbracket
      \Gamma_2,\Delta_2\rrbracket_\tau$.
  \end{enumerate}
  We have that $\dom(\Delta_1,\Delta_2)\subseteq FV(\pair{
  \vec{t}}{\vec{r}}) \subseteq \dom(\Gamma_1,\Delta_1,\Gamma_2,\Delta_2)$.
  For any $\sigma\in \llbracket
  \Gamma_1,\Delta_1,\Gamma_2,\Delta_2\rrbracket_\tau$ we have that $\sigma
  =\sigma_1 \sigma_2$ for $\sigma_i\in\llbracket
  \Gamma_i,\Delta_i\rrbracket_\tau$ and, from the disjointness of contexts,
  we deduce:
  \begin{align*}
    \pair{ \vec{t}}{\vec{r}} \langle \sigma\rangle & =\pair{ \vec{t}}{\vec{r}} \langle \sigma_1\rangle\langle \sigma_2\rangle\\
    &=\pair{\vec{t}\langle \sigma_1\rangle}{\vec{r}\langle \sigma_2\rangle}\Vdash \llbracket A \times B \rrbracket_\tau,
  \end{align*}
  derived from the definition of $\llbracket A \times B \rrbracket_\tau$.

\item[$(\times_e)$] Supposing that the judgments $\Gamma_1;\Delta_1\vdash
  t:A\times B$ and $\Gamma_2;\Delta_2,x:A,y:B\vdash \vec{s}:C$ are valid,
  we obtain:
  \begin{enumerate}
    \item $\dom(\Delta_1)\subseteq FV(t)\subseteq
      \dom(\Gamma_1,\Delta_1)\quad$ and $\quad t\langle\sigma\rangle \Vdash
      \llbracket A\times B \rrbracket_\tau$ for all $\sigma\in\llbracket
      \Gamma_1,\Delta_1\rrbracket_\tau$.
    \item $\dom(\Delta_2)\cup \{x,y\}\subseteq FV(\vec{s})\subseteq
      \dom(\Gamma_2,\Delta_2)\cup\{x,y\}$ and $\vec{s}\langle \sigma\rangle
      \Vdash \llbracket C \rrbracket_\tau$ for all
      $\sigma\in\llbracket\Gamma_2,\Delta_2,\{x:A,y:B\}\rrbracket_\tau$.
  \end{enumerate}

  We have that $\dom(\Delta_1,\Delta_2)\subseteq
  FV(\tlet{x}{y}{t}{\vec{s}})\subseteq
  \dom(\Gamma_1,\Delta_1,\Gamma_2,\Delta_2).$ For any $\sigma\in \llbracket
  \Gamma_1,\Delta_1,\Gamma_2,\Delta_2\rrbracket_\tau$ we have that $\sigma
  =\sigma_1 \sigma_2$ for $\sigma_i\in\llbracket
  \Gamma_i,\Delta_i\rrbracket_\tau$ and therefore
  \begin{align*}
    (\tlet{x}{y}{t}{\vec{s}})\langle \sigma \rangle &= (\tlet{x}{y}{t}{\vec{s}})\langle \sigma_1 \rangle\langle \sigma_2\rangle\\
    &=  \tlet{x}{y}{t\langle\sigma_1\rangle}{\vec{s}\langle \sigma_2\rangle}\\
    &\rightsquigarrow^* \tlet{x}{y}{\langle \vec{v},\vec{w}\rangle}{\vec{s}\langle\sigma_2\rangle}\\
    &\rightsquigarrow^* (\vec{s}[\vec{v}/x,\vec{w}/y])\langle \sigma_2\rangle\\
    &=\vec{s}\langle \vec{v}/x,\vec{w}/y,\sigma_2\rangle\Vdash \llbracket C \rrbracket_\tau
  \end{align*}
  from the fact that $\langle
  \vec{v}/x,\vec{w}/y,\sigma_2\rangle\in\llbracket\Gamma_2,\Delta_2,\{x:A,y:B\}\rrbracket_\tau$
  and~\cite[Lemmas A.10 and A.3]{DCGMV19}.

\item[$(\times_{e\sharp})$] From the validity of $\Gamma_1;\Delta_1\vdash
  \vec{t}:\sharp (Q\times R)$ and $\Gamma_2;\Delta_2,x:\sharp Q, y:\sharp R
  \vdash \vec{s}:S$ we have that
  \begin{enumerate}
    \item $\dom(\Delta_1)\subseteq FV(\vec t)\subseteq
      \dom(\Gamma_1,\Delta_1)\quad$ and $\quad \vec t\langle\sigma\rangle
      \Vdash \llbracket \sharp (Q\times R)\rrbracket_\tau$ for all $\sigma\in\llbracket
      \Gamma_1,\Delta_1\rrbracket_\tau$.
    \item $\dom(\Delta_2)\cup \{x,y\}\subseteq FV(\vec{s})\subseteq
      \dom(\Gamma_2,\Delta_2)\cup\{x,y\} $ and $ \vec{s}\langle
      \sigma\rangle \Vdash \llbracket S\rrbracket_\tau$ for all
      $\sigma\in\llbracket\Gamma_2,\Delta_2,\{x:\sharp Q,y:\sharp
      R\}\rrbracket_\tau$.
  \end{enumerate}
  We obtain directly that $\dom(\Delta_1,\Delta_2)\subseteq
  FV(\tlet{x}{y}{t}{\vec{s}})\subseteq
  \dom(\Gamma_1,\Delta_1,\Gamma_2,\Delta_2)$ and, for any $\sigma\in
  \llbracket \Gamma_1,\Delta_1,\Gamma_2,\Delta_2\rrbracket_\tau$ we have
  that $\sigma =\sigma_1 \sigma_2$ for $\sigma_i\in\llbracket
  \Gamma_i,\Delta_i\rrbracket_\tau$ and thus we have:
  \begin{align*}
    (\tlet{x}{y}{t}{\vec{s}})\langle \sigma \rangle &= (\tlet{x}{y}{t}{\vec{s}})\langle \sigma_1 \rangle\langle \sigma_2\rangle\\
    &=  \tlet{x}{y}{t\langle\sigma_1\rangle}{\vec{s}\langle \sigma_2\rangle}\\
    &\rightsquigarrow^* \tlet{x}{y}{\sum_{i=1}^n \alpha_i\cdot \pair{ \vec{v}_i}{\vec{w}_i}}{\vec{s}\langle\sigma_2\rangle}\\
    &= \sum_{i=1}^n \alpha_i \cdot \tlet{x}{y}{\pair{ \vec v_i}{\vec w_i}}{\vec{s}\langle \sigma_2\rangle}\\
    &\rightsquigarrow^* \sum_{i=1}^n \alpha_i\cdot \vec{s}[\vec{v}_i/x,\vec{w}_i/y]\langle \sigma_2\rangle\\
    &=\sum_{i=1}^n \alpha_i\cdot \vec{s}\langle \vec{v}_i/x,\vec{w}_i/y, \sigma_2\rangle\rightsquigarrow^* \sum_{i=1}^n \alpha_i\cdot \vec{z}_i \in \mathsf{span}(\llbracket S \rrbracket_\tau)
  \end{align*}
  To see that this has unit norm, consider:
  \begin{align*}
    \big\|\sum_{i=1}^n \alpha_i \cdot \vec{z}_i\big\|^2 &=\big\langle \sum_{i=1}^n \alpha_i \cdot \vec{z}_i\ \big|\ \sum_{i=1}^n \alpha_i \cdot \vec{z}_i\big\rangle 
    = \sum_{i=1}^n \sum_{j=1}^n \alpha_i \bar{\alpha}_j \langle \vec{z}_i|\vec{z}_j\rangle\\
    &= \sum_{i=1}^n \sum_{j=1}^n \alpha_i \bar{\alpha}_j \langle \vec{v}_i|\vec{v}_j\rangle \langle\vec{w}_i|\vec{w}_j\rangle\\
    &= \sum_{i=1}^n \sum_{j=1}^n \alpha_i \bar{\alpha}_j \langle (\vec{v}_i,\vec{w}_i)\mid (\vec{v}_j,\vec{w}_j)\rangle\\
    &=\langle \sum_{i=1}^n \alpha_i  (\vec{v}_i,\vec{w}_i)\mid \sum_{i=1}^n \alpha_i  (\vec{v}_i,\vec{w}_i) \rangle = \big\|\sum_{i=1}^n \alpha_i \cdot (\vec{v}_i,\vec{w}_i)\big\|^2=1,
  \end{align*}
  where we have used~\cite[Lemma A.9 and Proposition A.2]{DCGMV19}.
\item[$(\S_i)$] Assuming that the premise is a valid typing judgment, we
  have that $FV(\vec{t})=\dom(\Gamma,\Delta).$ Using the fact that
  $\llbracket \S B\rrbracket_\tau \triangleq \llbracket B\rrbracket_\tau$,
  we directly deduce that $\dom(\S \Delta)=\dom(\Delta)\subseteq
  FV(\vec{t})\subseteq \dom(\Gamma,\Delta)=\dom(\Gamma,\S \Delta)$. Since
  $\llbracket \Gamma,\Delta\rrbracket_\tau = \llbracket \Gamma,\S
  \Delta\rrbracket_\tau$ we conclude that the rule is valid.

\item[$(\sharp_i)$] If the premises are valid then we have that $\dom(\Delta)\subseteq FV(\vec{t}_i)\subseteq \dom(\Gamma,\Delta)$, for all $i$. Furthermore, for all $\sigma \in \llbracket \Gamma,\Delta\rrbracket$, $\vec{t}_i\langle \sigma \rangle \Vdash \llbracket A\rrbracket_\tau$. Since $\sum_{i=1}^n |\alpha_i|^2=1$, we have that 
\[\bigg(\sum_{i=1}^n \alpha_i\cdot\vec{t}_i\bigg)\langle \sigma\rangle =\sum_{i=1}^n \alpha_i\cdot(\vec{t}_i\langle \sigma \rangle)\Vdash \llbracket \sharp A\rrbracket_\tau.\]
\item[$(\S_e)$] If the premises are valid then:
  \begin{enumerate}
    \item $\dom(\Delta_1)\subseteq FV(\vec{r})\subseteq
      \dom(\Gamma_1,\Delta_1)$ and for all $\sigma\in\llbracket
      \Gamma_1,\Delta_1\rrbracket_\tau $, $\vec{r}\langle\sigma\rangle
      \Vdash \llbracket  B\rrbracket_\tau .$
    \item $\dom(\Delta_2)\cup \{x\}\subseteq FV(t) \subseteq \dom(\Gamma_2,
      \Delta_2)\cup \{x\}$ and for all $\sigma\in\llbracket
      \Gamma_2,\Delta_2,\{x: B\}\rrbracket_\tau $, $t\langle \sigma\rangle
      \Vdash \llbracket A \rrbracket_\tau $. 
  \end{enumerate}
  We immediately obtain that $\dom(\Delta_1,\Delta_2)\subseteq
  FV(t[\vec{r}/x])\subseteq \dom(\Gamma_1,\Delta_1,\Gamma_2,\Delta_2)$. 
  \[
    t[\vec{r}/x]\langle \sigma\rangle = t[\vec{r}/x]\langle \sigma_1,\sigma_2\rangle = t[\vec{r}\langle \sigma_1\rangle/x]\langle \sigma_2\rangle \Vdash \llbracket A\rrbracket_\tau,
  \]
  since $\vec{r}\langle\sigma_1\rangle\Vdash B$ (from 1 and the fact that
  $\sigma_1\in\llbracket \Gamma_1,\Delta_1\rrbracket_\tau$), and so we have
  that $\langle\{\vec{r}\langle\sigma_1\rangle/x\},\sigma_2\rangle\in
  \llbracket \Gamma_2,\Delta_2,\{x:B\}\rrbracket_\tau.$

\item[$(\forall_i)$] Suppose that $\Gamma;\Delta\vdash\vec t:A$ is valid and that $X\notin FV(\Gamma,\Delta)$. Then
  \[
    \dom(\Delta)\subseteq FV(\vec t)\subseteq\dom(\Gamma,\Delta)
    \textrm{ and }
    \vec t\langle\sigma\rangle\Vdash\llbracket A\rrbracket_\tau
    \textrm{ for all }
    \sigma\in\llbracket{\Gamma,\Delta}\rrbracket_\tau
  \]
  Let $\tau' = \tau\setminus\{X\}$. By definition, we have
  \[
  \llbracket \forall X.A\rrbracket_{\tau'} = \bigcap_{B\in\mathcal S_1}\llbracket A\rrbracket_{\tau'\cup\{X\mapsto\llbracket B\rrbracket_\emptyset\}}
  \]
  Since $X\notin FV(\Gamma,\Delta)$, 
  for all $\sigma$ such that
  $\sigma\in\llbracket{\Gamma,\Delta}\rrbracket_\tau$, we also have
  $\sigma\in\llbracket{\Gamma,\Delta}\rrbracket_{\tau'\cup\{X\mapsto R\}}$ for any $R\subseteq\mathcal S_1$.

  Therefore, 
  $\vec t\langle\sigma\rangle\Vdash\llbracket A\rrbracket_\tau$
  implies
  $\vec t\langle\sigma\rangle\Vdash\llbracket A\rrbracket_{\tau'\cup\{X\mapsto R\}}$
  for all $R\in\mathcal S_1$, thus,  
  $\vec t\langle\sigma\rangle\Vdash\llbracket \forall X.A\rrbracket_{\tau'}$.

\item[$(\forall_e)$]
  Suppose that $\Gamma;\Delta\vdash\vec t:\forall X.A$. Then
  \[
    \dom(\Delta)\subseteq FV(\vec t)\subseteq\dom(\Gamma,\Delta)
    \textrm{ and }
    \vec t\langle\sigma\rangle\Vdash\llbracket \forall X.A\rrbracket_\tau
    \textrm{ for all }
    \sigma\in\llbracket{\Gamma,\Delta}\rrbracket_\tau
  \]
  Thus, by definition, since $\llbracket B\rrbracket_\tau\subseteq\mathcal S_1$, we have  $\vec t\langle\sigma\rangle\Vdash\llbracket A\rrbracket_{\tau\cup\{X\mapsto\llbracket B\rrbracket_\tau\}}$.

  We conclude by Lemma~\ref{lem:substitutionRealizability}.
  
  \end{itemize}
\end{proof}

\lemunitaryabstractions*
\begin{proof}
(The condition is necessary.) For some $\lambda x.\vec{t}\in\llbracket \sharp (\B^n) \multimap \sharp (\B^k)\rrbracket_\emptyset$, let $\vec{v}_i\in\llbracket \sharp (\B^k)\rrbracket_\emptyset$ be the values that satisfy $\vec{t}[\ket{i}/x]\rightsquigarrow^* \vec{v}_i,\,\forall i=0,\dots,2^n-1$. Then, for any $\alpha_i\in\Ct$ such that $\sum_i|\alpha_i|^2=1$ we have, by linearity,
\[\vec{t}\bigg[\sum_i \alpha_i\cdot \ket{i}/x\bigg]= \sum_i \alpha_i\cdot \vec{t}[\ket{i}/x]\rightsquigarrow^*\sum_i \alpha_i \cdot \vec{v}_i.\]
Since $\sum_i \alpha_i \cdot\ket{i}\in\llbracket\sharp(\B^n)\rrbracket_\emptyset$, we have that $\sum_i \alpha_i\cdot \vec{v}_i\in\llbracket \sharp(\B^k)\rrbracket_\emptyset$, and therefore $\|\sum_i \alpha_i\cdot \vec{v}_i\|=1$. From this, we can derive
\begin{align*}
1=\bigg|\bigg|\sum_i \alpha_i\cdot \vec{v}_i\bigg|\bigg|^2 &= \big\langle \sum_i \alpha_i\cdot \vec{v}_i\big| \sum_i \alpha_i\cdot \vec{v}_i \big\rangle\\
&= \sum_i|\alpha_i|^2\langle \vec{v}_i\mid\vec{v}_i\rangle + \sum_{i<j}\big(\alpha_i\overline{\alpha_j}\langle \vec{v}_i\mid\vec{v}_j\rangle + \alpha_j\overline{\alpha_i}\langle \vec{v}_j\mid\vec{v}_i\rangle\big)\\
&=\sum_i |\alpha_i|^2 +\sum_{i<j}\big(\alpha_i\overline{\alpha_j}\langle \vec{v}_i\mid\vec{v}_j\rangle + \overline{\alpha_i\overline{\alpha_j}\langle \vec{v}_i\mid\vec{v}_j\rangle}\big)\\
&=1+ \sum_{i<j} 2\text{Re}\big(\alpha_i\overline{\alpha_j}\langle \vec{v}_i|\vec{v}_j\rangle\big).
\end{align*}
Picking $\alpha_0=\alpha_1=\frac{1}{\sqrt{2}}$, and $\alpha_i=0$, for $i\not = 0,1$ we deduce $\text{Re}(\langle \vec{v}_0\mid\vec{v}_1\rangle)=0$ and from $\alpha_0=\frac{1}{\sqrt{2}},\alpha_1=\frac{i}{\sqrt{2}}$ and $\alpha_i=0,\forall i\not = 0,1$ we obtain $\text{Im}(\langle \vec{v}_0\mid\vec{v}_1\rangle)=0$ and therefore $\langle \vec{v}_0\mid\vec{v}_1\rangle=0$. The same reasoning can be applied to all other pairs $\alpha_i,\alpha_j$ such that $i\not = j$ and therefore $\langle \vec{v}_i\mid \vec{v}_j\rangle=0$.

(The condition is sufficient.) Suppose there exist $\vec{v}_i\in \llbracket\sharp(\B^k)\rrbracket_\emptyset$, $i=0,\dots,2^n-1$,  for which $\langle \vec v_i\mid\vec v_j\rangle = 0,\,\forall i\not = j$ and such that $\vec{t}[\ket{i}/x]\rightsquigarrow^* \vec v_i$. In particular, we have that, for all $i$, $\vec{v}_i\in\mathsf{span}(\{\ket{i}\mid i=0,\dots,2^{n}-1\})$ and $\|\vec v_i\|=1.$ For any given $\vec{v}\in\llbracket\sharp (\B^n)\rrbracket_\emptyset$, $\vec{v}=\sum_i \alpha_i \cdot\ket{i}$ such that $\sum_i|\alpha_i|^2 =1$. Then,
\[\vec{t}[\vec{v}/x]=\sum_i\alpha_i\cdot\vec{t}[\ket{i}/x]\rightsquigarrow^*\sum_i \alpha_i\cdot \vec{v}_i\in \llbracket \sharp (\B^k)\rrbracket_\emptyset,\]
since $\sum_i\alpha_i\cdot\vec{v}_i\in\mathsf{span}(\{\ket{i}\mid i=0,\dots,2^{n}-1\})$ and
\begin{align*}
\bigg|\bigg|\sum_i \alpha_i\cdot \vec{v}_i\bigg|\bigg|^2 &= \big\langle \sum_i \alpha_i\cdot \vec{v}_i\big| \sum_i \alpha_i\cdot \vec{v}_i\big\rangle\\
&= \sum_i|\alpha_i|^2\langle \vec{v}_i\mid\vec{v}_i\rangle + \sum_{i<j}\big(\alpha_i\overline{\alpha_j}\langle \vec{v}_i\mid\vec{v}_j\rangle + \alpha_j\overline{\alpha_i}\langle \vec{v}_j\mid\vec{v}_i\rangle\big)\\
&= \sum_i|\alpha_i|^2 \|\vec{v}_i\|^2 + \sum_{i<j}\big(0+0\big)\\
& = \sum_i|\alpha_i|^2 =1.
\end{align*}
Therefore, for all $\vec{v}\in\llbracket \sharp\B^n\rrbracket_\emptyset$, $\vec{t}[\vec{v}/x]\Vdash \sharp (\B^n)$, and finally $\lambda x.\vec{t}\in\llbracket \sharp (\B^n)\multimap\sharp(\B^n)\rrbracket_\emptyset$.
\end{proof}

\thmisocomplete*
\begin{proof}
We provide a constructive proof. If $\mathcal{I}:\Ct^{2^n}\to\Ct^{2^k}$ represents an isometry, then the columns in its matrix representation have unit norm and are mutually orthogonal. This means that $\forall i,j=0,\dots,2^n-1$, $\mathcal{I}(\lmap{\ket{i}})^\dagger\cdot \mathcal{I}(\lmap{\ket{j}})=0$ if $i\not= j$. Furthermore, for any $i=0,\dots,2^n-1$, $\|\mathcal{I}(\lmap{\ket{i}})\|=1$ and so $\mathcal{I}(\lmap{\ket{i}})\triangleq (\alpha^i_0 \dots \alpha^i_{2^k-1})$ can be encoded by a closed \punq{} term $\vec{v}_i= \sum_{j=0}^{2^k-1}\alpha^i_j\cdot\ket{j}\in\llbracket \sharp (\B^k)\rrbracket$.
We now perform the encoding. Let $x,x_m$ and $x_{[1..m]}$ all represent different variables, for $1\leq m\leq n$. We can encode the operator $\mathcal{I}$ using nested $\mathsf{if}$ statements, as follows:
\begin{align*}
&\vec{t}\triangleq\tlet{x_{[1..n-1]}}{x_n}{x\\
&\qquad}{\tlet{x_{[1..n-2]}}{x_{n-1}}{x_{[1..n-1]}\\
&\qquad\quad}{\dots}}\\
&\qquad\qquad\ \mathsf{in}\ \tlet{x_1}{x_2}{x_{[1..2]}}{\mathsf{if}}\ x_1 \!\!\!\!\!\!\!\!\!\!&\mathsf{then} \dots \mathsf{if}\ x_{n-1}\ &\mathsf{then}\ (\tif{x_n}{\vec{v}_0} {\vec{v}_1})\\
&&&\mathsf{else}\ \ \,(\tif{x_n}{\vec{v}_2}{\vec{v}_3})\\
&&&\qquad\qquad\vdots\\
&&\mathsf{else} \dots \mathsf{if}\ x_{n-1}\ &\mathsf{then}\ (\tif{x_n}{\vec{v}_{2^n-4}} {\vec{v}_{2^n-3}})\\
&&&\mathsf{else}\ \ \,(\tif{x_n}{\vec{v}_{2^n-2}}{\vec{v}_{2^n-1}}).
\end{align*}
Then, the term $\lambda x.\vec{t}\in \SUP_c$ has type $\sharp(\B^n)\multimap \sharp(\B^k)$, and $\forall i=0,\dots,2^n-1$, $(\lambda x.\vec{t})\ \ket{i}\rightsquigarrow^* \vec{v}_i$.
\end{proof}

\thmnonseparability*
\begin{proof}
  Consider an argument by contradiction, starting with a candidate closed term
  $\lambda x.\vec{t}\in \llbracket \sharp (\B^{n+k})\multimap
  \big(\sharp(\B^{n})\times \sharp(\B^{k})\big)\rrbracket_\emptyset $. Define
  the values $(\vec u_i,\vec v_i)\in \llbracket \sharp (\B^{n})\times
  \sharp(\B^k)\rrbracket_\emptyset$ satisfying $(\lambda x.\vec{t})\
  \ket{i}\rightsquigarrow^* (\vec u_i,\vec v_i) $. Notice that $\lambda
  x.\vec{t}$ can be subtyped as $\sharp (\B^{n+k})\multimap \sharp (\B^{n+k})$ and
  therefore by Theorem~\ref{thm:unitarity} must satisfy unitarity. We conclude,
  therefore, that since we have $2^{n+k}$ values of $(\vec u_i,\vec v_i)$ such that
  they are all orthogonal and of unit norm, then they must form an orthonormal
  basis of $\llbracket \sharp (\B^{n+k})\rrbracket_\emptyset$. In particular, they
  suffice to represent any value $\vec{s}\triangleq \sum_i \alpha_i \cdot (\vec
  u_i, \vec v_i)$ with $\sum_i |\alpha_i|^2 =1,$ and in particular we may choose
  $\vec{s}$ to be an entangled state and therefore $\vec{s}\not \in \llbracket
  \sharp (\B^{n})\times \sharp (\B^k)\rrbracket_\emptyset$. However, we have
  that $\sum_i \alpha_i \cdot \ket{i}\in \llbracket \sharp
  (\B^{n+k})\rrbracket_\emptyset$ and that $\lambda x.\vec{t}\ (\sum_i \alpha_i \cdot
  \ket{i}\,)\rightsquigarrow^* \vec{s}$. 
\end{proof}

\subsection{Proofs of Section~\ref{s:ptstrongnormalization}}\label{app:ptstrongnormalization}

Shortand notations. For two sets $\mathcal{S},\mathcal{T}$, we have that
\begin{itemize}
\item $\mathcal{S}\ \mathcal{T}\triangleq \{s\ t\mid s\in\mathcal{S}, \ t\in \mathcal{T}\}$
\item $\lambda z.\mathcal{S}\triangleq \{\lambda z. s\mid s\in \mathcal{S}\}$
\item $\mathcal{S}[\mathcal{T}/z]\triangleq \{s[t/z]\mid s\in \mathcal{S},\ t\in \mathcal{T}\}$ 
\end{itemize}

\begin{lemma}\label{lemma:substitution}
Let $\vec{t}\in \SUP{}$ and $v\in \BV{}$, then $\lpar{\vec{t}[v/x]}= \lpar{\vec{t}}[\lpar{v}/x]$. 
\end{lemma}
\begin{proof}
By induction on the structure of $\vec{t}$.
\begin{itemize}
\item ($\vec{t}=x$) $\lpar{\vec{t}[v/x]}=\lpar{v}=x[\lpar{v}/x]=\lpar{\vec{t}}[\lpar{v}/x]$.
\item ($\vec{t}=y$) $\lpar{\vec{t}[v/x]}=\lpar{y}=y=y[\lpar{v}/x]=\lpar{y}[\lpar{v}/x]=\lpar{\vec{t}}[\lpar{v}/x]$.
\item ($\vec{t}=\ket{0}$) $\lpar{\ket{0}[v/x]}=\lpar{\ket{0}}=\lambda z_1.\lambda z_2.z_1=(\lambda z_1.\lambda z_2.z_1)[\lpar{v}/x]=\lpar{\ket{0}}[\lpar{v}/x]=\lpar{\vec{t}}[\lpar{v}/x]$.
\item ($\vec{t}=\ket{1}$) Similar to case $(\ket{0})$.
\item ($\vec{t}=\tif{r}{\vec{s_1}}{\vec{s_2}}$)
\begin{align*}
\lpar{\vec{t}[v/x]}&=\lpar{\tif{r[v/x]}{\vec{s_1}[v/x]}{\vec{s_2}[v/x]}} \\
&=(\lpar{r[v/x]}\ \lambda x_{1\dots k}.\lpar{t_1[v/x]}\ \lambda x_{1\dots k}.\lpar{t_2[v/x]})\ x_1\dots x_k\\
&=^\text{IH}(\lpar{r}[\lpar{v}/x]\ \lambda x_{1\dots k}.\lpar{t_1}[\lpar{v}/x]\ \lambda x_{1\dots k}.\lpar{t_2}[\lpar{v}/x])\ x_1\dots x_k\\
&=(\lpar{r}\ \lambda x_{1\dots k}.\lpar{t_1}\ \lambda x_{1\dots k}.\lpar{t_2})[\lpar{v}/x]\ x_1\dots x_k\\
&=((\lpar{r}\ \lambda x_{1\dots k}.\lpar{t_1}\ \lambda x_{1\dots k}.\lpar{t_2})\ x_1\dots x_k)[\lpar{v}/x]= \lpar{\vec{t}}[\lpar{v}/x]
\end{align*}
\item ($\vec{t}= \lambda z.\vec{s}$)
\begin{align*}
\lpar{\vec{t}[v/x]}&=\lpar{\lambda z.\vec{s}[v/x]}\\
&= \lambda z.\lpar{\vec{s}[v/x]}\\
&=^\text{IH} \lambda z. \lpar{\vec{s}}[\lpar{v}/x]\\
&=\lpar{\lambda z.\vec{s}}[\lpar{v}/x]=\lpar{\vec{t}}[\lpar{v}/x].
\end{align*}
\item ($\vec{t} = s\ r$)
\begin{align*}
\lpar{\vec{t}[v/x]}&=\lpar{(s\ r)[v/x]}\\
&= \lpar{s[v/x]\ r[v/x]}\\
&=\lpar{s[v/x]}\ \lpar{r[v/x]}\\
&=^\text{IH} \lpar{s}[\lpar{v}/x]\ \lpar{r}[\lpar{v}/x]\\
&=(\lpar{s}\ \lpar{r})[\lpar{v}/x]=\lpar{s\ r}[\lpar{v}/x]=\lpar{\vec{t}}[\lpar{v}/x].
\end{align*}
\item ($\vec{t} = \pair{s}{r}$)
\begin{align*}
\lpar{\vec{t}[v/x]}&=\lpar{\pair{s}{r}[v/x]}\\
&= \lpar{\pair{s[v/x]}{r[v/x]}}\\
&=\lambda z.(z\ \lpar{s[v/x]}\ \lpar{r[v/x]})\\
&=^\text{IH} \lambda z.(z\ \lpar{s}[\lpar{v}/x]\ \lpar{r}[\lpar{v}/x])\\
&=(\lambda z.(z\ \lpar{s}\ \lpar{r}))[\lpar{v}/x]=\lpar{\pair{s}{r}}[\lpar{v}/x]=\lpar{\vec{t}}[\lpar{v}/x].
\end{align*}
\item ($\vec{t}=\tlet{z_1}{z_2}{s}{\vec{r}}$)
\begin{align*}
\lpar{\vec{t}[v/x]}&=\lpar{\tlet{z_1}{z_2}{s[v/x]}{\vec{r}[v/x]}}\\
&= \lpar{s[v/x]} (\lambda z_1.\lambda z_2.\lpar{\vec{r}[v/x]})\\
&=^\text{IH} \lpar{s}[\lpar{v}/x] (\lambda z_1.\lambda z_2.\lpar{\vec{r}}[\lpar{v}/x])\\
&= \big(\lpar{s}(\lambda z_1.\lambda z_2.\lpar{\vec{r}})\big)[\lpar{v}/x]\\
&= \lpar{\tlet{z_1}{z_2}{s}{\vec{r}}}[\lpar{v}/x]=\lpar{\vec{t}}[\lpar{v}/x].
\end{align*}
\item ($\vec{t}=\vec{0}$) Similar to case of $\ket{0}$.
\item ($\vec{t}=\alpha\cdot \vec{s}$)  $\lpar{\vec{t}[v/x]}=\lpar{\alpha\cdot \vec{s}[v/x]}=$ if $\alpha=0$
\item ($\vec{t}=\vec{s}+\vec{r}$) \qedhere
\end{itemize}
\end{proof}

\lemtransequiv*
\begin{proof}
By inspection of the equivalence relation.
\begin{itemize}
\item $\vec{t}_1+\vec{t}_2\equiv \vec{t}_2+\vec{t}_1$, where we have $\lpar{\vec{t}_1+\vec{t}_2}=\lpar{\vec{t}_1}\cup\lpar{\vec{t}_2}=\lpar{\vec{t}_2}\cup\lpar{\vec{t}_1}=\lpar{\vec{t}_2+\vec{t}_1}$.

\item $\vec 0+\vec{t}_1\equiv \vec{t}_1$, in which case $\lpar{\vec{0}+\vec{t}_1}=\{\ast\}\cup\lpar{\vec{t}_1}$, therefore $\lpar{\vec{0}+\vec{t}_1} \eqast \lpar{\vec{t}_1}$.
\end{itemize}
For all other cases, we easily find that $\lpar{\vec{t}}=\lpar{\vec{r}}$ and the conclusion follows.
\end{proof}

\lemdlaltranslation*
\begin{proof}
 We prove this statement by induction on the derivation tree of $\Gamma;\Delta\vdash\vec t:A$.
\begin{itemize}
\item[$(\mathsf{W})$] By the induction hypothesis, we have that $\Gamma^\star;\Delta^\star \vdash_{\text{DLAL}} \lpar{\vec t} : A^\star$. Using rule (Weak) in DLAL we obtain $\Gamma^\star,\Gamma'^\star;\Delta^\star \vdash_{\text{DLAL}} \lpar{\vec t} : A^\star$.
\item[$(\mathsf{C})$] By the IH we have $\Gamma^\star,x:B^\star,y:B^\star;\Delta^\star \vdash_{\text{DLAL}} \lpar{\vec{t}}: A^\star$.  Applying rule (Cntr) in DLAL we obtain $\Gamma^\star,x:B^\star;\Delta^\star \vdash_{\text{DLAL}} \lpar{\vec t}[x/y]: A^\star$ which by Lemma~\ref{lemma:substitution}
\item[$(\equiv)$] We can check that this is the case for all terms $\vec{t_1}, \vec{t_2}$ such that $\vec{t}_1\equiv \vec{t}_2$. It can be shown by inspection that, for almost all cases in the equivalence relation $\equiv$, the translation is preserved, i.e. that if $\vec{t}_1\equiv \vec{t}_2$, then $\lpar{\vec{t}_1}=\lpar{\vec{t}_2}$. However, this is not the case in two instances:
\begin{itemize}
\item[$\bullet$] $\vec{t_1}=\vec{0}+\vec{t}\equiv \vec{t}=\vec{t_2}$. In this case, $\lpar{\vec{t_1}}=\lpar{\vec{t_2}}\cup \{\ast\}$ and therefore $\lpar{\vec{t_2}}\subseteq \lpar{\vec{t_1}}$ from which the conclusion follows.
\item[$\bullet$] $\vec{t_1}= 0\cdot \vec{t}\equiv \vec{0}=\vec{t_2}$. Here, since $\lpar{\vec{t_2}}=\{\ast\}$ which admits any type, the conclusion follows.
\end{itemize}
\item[$(\leq)$] By Lemma~\ref{lem:translation-subtype}, for all subtying relations $A\leq B$ we have that $A^\star = B^\star$ and therefore the result follows.
\item[$(\mathsf{Ax})$] We have that $\lpar{x}=\{x\}$, and by rule (Id) in \dlal{}, we may obtain $;x:A^\star\vdash x:A^\star$ for any $A\in\mathbb{T}$.
\item[$(0)$] We have that $\lpar{\ket{0}}=\{\lambda x.\lambda y.x\}$ and $\B^\star=\forall X. (X\multimap X\multimap X)$, and by DLAL rules (Id), (Weak), and two uses of ($\multimap$ i), and finally ($\forall$ i) we obtain the desired typing derivation. Indeed, this is the standard encoding in \dlal{} for booleans.
\item[$(1)$] Analogous to $(0)$.
\item[$(\mathsf{if})$] We have that $\lpar{\tif{t}{t_1}{t_2}}=(t\ \lambda x_{1\dots k}. \lpar{t_1}\ \lambda x_{1\dots k}.\lpar{t_2})\ x_1'\dots x_k'$. By the IH we have that $\Gamma'^\star;\Delta'^\star \vdash \lpar{t_1}\cup \lpar{t_2}:A^\star $, and $\Gamma^\star;\Delta^\star \vdash t:\forall X.(X\multimap X\multimap X)$. Let $\Delta'^\star= \{x_i:A_i^\star\}_{i=1\dots k}$,
\[
 \begin{prooftree}
      \hypo{\Gamma'^\star;\Delta'^\star \vdash \lpar{t_1}\cup \lpar{t_2}:A^\star}
      \infer1[($\multimap$ e) $k$ times]{\Gamma'^\star;\vdash \lpar{\lambda x_{1\dots k}.t_1}\cup \lpar{\lambda x_{1\dots k}.t_2}:A_1^\star \multimap \dots \multimap A_k^\star \multimap A^\star}
    \end{prooftree}
\]

Without loss of generality, we may consider that the variables in $\Gamma'^\star$ have different names in the typing of $t_1$ and $t_2$. Let us use $\widetilde{\Gamma'^\star}$ to denote the nonlinear context of $\widetilde{t_2}$, where we have renamed the nonlinear variables. 

Therefore, we may derive the following judgement in DLAL. Let $B\triangleq A_1^\star \multimap \dots \multimap A_k^\star \multimap A^\star$.

\[
\scalebox{0.75}{
 \begin{prooftree}
      \hypo{\Gamma^\star;\Delta^\star \vdash t: \forall X.(X\multimap X\multimap X)}
      \infer1[($\forall$ e)]{ \Gamma^\star;\Delta^\star \vdash t: B\multimap B\multimap B }
      \hypo{\Gamma'^\star; \vdash \lpar{\lambda x_{1\dots k}.t_1}: B}
      \infer2[($\multimap$ e)]{\Gamma^\star,\Gamma'^\star;\Delta^\star \vdash t\  \lpar{\lambda x_{1\dots k}.t_1}: B\multimap B}
      \hypo{\widetilde{\Gamma'^\star}; \vdash \lpar{\lambda x_{1\dots k}.\widetilde{t_2}}: B}
      \infer2[($\multimap$ e)]{\Gamma^\star,\Gamma'^\star,\widetilde{\Gamma'^\star};\Delta^\star \vdash t\  \lpar{\lambda x_{1\dots k}.t_1}\ \lpar{\lambda x_{1\dots k}.\widetilde{t_2}}: B}
      \infer1[(Cntr)]{\Gamma^\star,\Gamma'^\star;\Delta^\star\vdash t\  \lpar{\lambda x_{1\dots k}.t_1}\ \lpar{\lambda x_{1\dots k}.t_2}: A_1^\star \multimap \dots \multimap A_k^\star \multimap A^\star}
       \infer0[(Id)]{x_i:A_i^\star \vdash x_i:A_i^\star}
      \infer2[($\multimap$ e) ]{\Gamma^\star,\Gamma'^\star;\Delta^\star, \Delta'^\star\vdash (t\  \lpar{\lambda x_{1\dots k}.t_1}\ \lpar{\lambda x_{1\dots k}.t_2})\ x_1\dots x_k: A^\star}
    \end{prooftree}}
\]

\item[($\mathsf{if}_\sharp$)] Since $(\sharp A)^\star = A^\star$, the proof is precisely the same as done in the case of ($\mathsf{if}$).

\item[($\multimap_i$)] By the IH we have that $\Gamma^\star;\Delta^\star, x: A^\star \vdash \lpar{\vec{t}}:B^\star$. By rule $(\multimap$ i) in DLAL we obtain $\Gamma^\star;\Delta^\star\vdash \lambda x.\lpar{\vec{t}}:A^\star \multimap B^\star$. Since $\lpar{\lambda x.\vec{t}}=\lambda x.\lpar{\vec{t}}$ and $(A\multimap B)^\star = A^\star \multimap B^\star$ this is precisely our desired conclusion.

\item[($\multimap_e$)] By the IH we have $\Gamma^\star;\Delta^\star\vdash \lpar{t}: A^\star \multimap B^\star$ and $\Gamma'^\star;\Delta'^\star \vdash \lpar{s}: A^\star$. Applying the DLAL rule $(\multimap$ e) we obtain $\Gamma^\star, \Gamma'^\star;\Delta^\star, \Delta'^\star\vdash \lpar{t}\ \lpar{s}: B^\star$. Since $\lpar{t\ s}=\lpar{t}\ \lpar{s}$ the proof is concluded.

\item[($\Rightarrow_i$)] By the IH we have $\Gamma^\star, x:A^\star;\Delta^\star\vdash \lpar{\vec{t}}:B^\star$. By rule $(\Rightarrow$ i) we obtain $ \Gamma^\star;\Delta^\star\vdash \lambda x.\lpar{\vec{t}}: A^\star \Rightarrow B^\star$. By Lemma~\ref{lem:translation-bang} we have that $(!(A))^\star=A^\star$.
\item[($\Rightarrow_e$)] Similar to $(\multimap_e)$. We may simply use the rule ($\Rightarrow$ e) in DLAL to obtain our desired conclusion. 
\item[($\times_i$)]  By the IH we have that $\Gamma^\star;\Delta^\star \vdash \lpar{t}:A^\star$ and $\Gamma'^\star;\Delta'^\star \vdash \lpar{s}: B^\star$. We may obtain the following derivation in DLAL (i.e. the pair encoding). 
\[
\scalebox{0.9}{
\begin{prooftree}
      \infer0[(Id)]{;x:A^\star\multimap B^\star\multimap X\vdash x:A^\star\multimap B^\star\multimap X}
      \hypo{\Gamma^\star;\Delta^\star \vdash \lpar{t}:A^\star}
      \infer2[($\multimap$ e)]{ \Gamma^\star;\Delta^\star,x:A^\star\multimap B^\star\multimap X \vdash x\ \lpar{t}: B^\star \multimap X}
       \hypo{\Gamma'^\star;\Delta'^\star \vdash \lpar{s}:B^\star}
      \infer2[($\multimap$ e)]{ \Gamma^\star,\Gamma'^\star;\Delta^\star, \Delta'^\star,x:A^\star\multimap B^\star\multimap X \vdash x\ \lpar{t}\ \lpar{s}: X}
      \infer1[($\multimap$ i)]{\Gamma^\star,\Gamma'^\star;\Delta^\star, \Delta'^\star\vdash \lambda x.(x\ \lpar{t}\ \lpar{s}): (A^\star\multimap B^\star \multimap X)\multimap X}
      \infer1[($\forall$ i)]{\Gamma^\star,\Gamma'^\star;\Delta^\star, \Delta'^\star\vdash \lambda x.(x\ \lpar{t}\ \lpar{s}): \forall X.( (A^\star\multimap B^\star \multimap X)\multimap X)}
    \end{prooftree}}
\]
This is precisely our desired result.
\item[$(\times_e)$] We have $\Gamma^\star;\Delta^\star \vdash \lpar{t}:\forall X. ((A^\star\multimap B^\star \multimap X )\multimap X)$ and $\Gamma'^\star;\Delta'^\star, x:A^\star,y:B^\star\vdash \lpar{\vec{s}}: C^\star$ by the IH. In DLAL we are able to derive
\[
\scalebox{0.9}{
\begin{prooftree}
	\hypo{\Gamma^\star;\Delta^\star \vdash \lpar{t}:\forall X. ((A^\star\multimap B^\star \multimap X )\multimap X)}
	\infer1[($\forall$ e)]{\Gamma^\star;\Delta^\star \vdash \lpar{t}:(A^\star\multimap B^\star \multimap C^\star )\multimap C^\star}
      \hypo{\Gamma'^\star;\Delta'^\star, x:A^\star,y:B^\star\vdash \lpar{\vec{s}}: C^\star}
      \infer1[($\multimap$ i)]{\Gamma'^\star;\Delta'^\star, x:A^\star\vdash \lambda y.\lpar{\vec{s}}: B^\star\multimap C^\star}
      \infer1[($\multimap$ i)]{\Gamma'^\star;\Delta'^\star\vdash \lambda x.\lambda y.\lpar{\vec{s}}: A^\star \multimap B^\star\multimap C^\star}
      \infer2[($\multimap$ e)]{\Gamma^\star,\Gamma'^\star;\Delta^\star,\Delta'^\star\vdash \lpar{t}\ \lambda x.\lambda y.\lpar{\vec{s}}:C^\star}
    \end{prooftree}}
\]
which concludes this case.
\item[$(\times_{e\sharp})$] Same proof as in $(\times_e)$.
\item[$(\sharp_i)$] Trivial.
\item[$(\S_i)$]$(\S_e)\ (\forall_i)\ (\forall_e)$ These rules have a direct counterpart in DLAL so their proof is trivial.
\end{itemize}
\end{proof}

\lemsize*

\begin{proof}
By induction on the structure of $\vec t$.
\begin{itemize}
\item If $\vec t=\ket{0}$, then $\lpar{\vec{t}}=\{\lambda x.\lambda y. x\}$ and therefore $|\vec t|=1$ and $|\lpar{\vec t}|=5$. Similar case for $t=\ket{1}\,$.
\item If $\vec t= x$, then $|t|=1$ and $\lpar{\vec t}=\{x\}$ and therefore $|\lpar{\vec t}|=1$.
\item If $\vec t=\lambda x.\vec{s}$, then $|\vec t|=1+|\vec{s}|$ and $|\lpar{\vec t}|=|\lambda x_{[1\dots k]}.\lpar{\vec{s}_{[x:1\dots k]}}|= O(2|\vec s|)=O(|\vec t|)$.
\item If $\vec t = t_1\ t_2$, we have $|\vec t|=|t_1|+|t_2|$ and 
\begin{align*}
|\lpar{t_1\ t_2}| \triangleq & |(\dots(\lpar{t_1}\ \lpar{t_2})\ \dots \lpar{t_2})|\\
=& |\lpar{t_1}| + O(|t_1|)\cdot |\lpar{t_2}|\\
=&^\text{IH}\ O(|t_1|^{d_1})+O(|t_1|)\cdot O(|t_2|^{d_2})\\
=&\ O(|\vec{t}|^{\max(d_1,d_2+1)}).
\end{align*}
Notice that, if $t_1$ does not duplicate $t_2$, then the bounding degree is instead $\max(d_1,d_2)$. 
\item If $\vec t= \tif{s}{p_1}{p_2}$, we have $|\vec t|=1+|s|+|p_1|+|p_2|$, and
\begin{align*}
|\lpar{\vec t}| \triangleq & \Big|(\lpar{s}\ \lambda x_{[1\dots k]}.\lpar{p_1}\ \lambda x_{[1\dots k]}.\lpar{p_2})\ x_1'\dots x_k'\Big|\\
=& |\lpar{s}|+O(|p_1|) +O(|\lpar{p_1}|)+O(|p_2|)+O(|\lpar{p_2}|)\\
=&^\text{IH}\ O(|s|^{d_1})+O(|p_1|^{d_2})+O(|p_2|^{d_3})+O(|p_1|)+O(|p_2|)\\
=& O(|\vec{t}|^{\max(d_1,d_2,d_3)}).
\end{align*}
\item If $\vec t =\pair{ t_1}{t_2 }$, then $|\vec t|=1+|t_1|+|t_2|$ and
\begin{align*}
|\lpar{\vec t}| \triangleq  |\lambda x.(x\ \lpar{t_1}\ \lpar{t_2})|=3+|\lpar{t_1}|+|\lpar{t_2}|&=^\text{IH} 2+ O(|t_1|^{d_1})+ O(|t_2|^{d_2})\\
&=O(|\vec t|^{\max(d_1,d_2)}).
\end{align*}
\item If $\vec t=\tlet{x}{y}{t_1}{t_2}$ then $|t|=1+\max(|t_1|,|t_2|)$ and
\begin{align*}
|\lpar{\vec t}|  \triangleq  |\lpar{t_1}\ (\lambda x.\lambda y.\lpar{t_2})|&=2+|\lpar{t_1}|+ |\lpar{t_2}|\\
&=^\text{IH} 2+O(|t_1|^{d_1})+ O(|t_2|^{d_2})\\
&=O(|\vec t|^{\max(d_1,d_2)})
\end{align*}
\item If $\vec{t}= \vec{t}_1+\vec{t}_2$, then $|\vec t|=\max(|\vec t_1|,|\vec t_2|)$, and
\begin{align*}
|\lpar{\vec{t}}|=|\lpar{\vec{t}_1}\cup\lpar{\vec{t}_2}|=\max(\lpar{\vec{t}_1},\lpar{\vec{t}_2})&=^\text{IH} O(\max(|\vec t_1|^{d_1},|\vec t_2|^{d_2}))\\
&=O(|\vec{t}|^{\max(d_1,d_2)}).
\end{align*}
\item If $\vec{t}= \vec{t}_1+\vec{t}_2$, then $|\vec t|=\max(|\vec t_1|,|\vec t_2|)$, and
\begin{align*}
|\lpar{\vec{t}}|=|\lpar{\vec{t}_1}\cup\lpar{\vec{t}_2}|=\max(\lpar{\vec{t}_1},\lpar{\vec{t}_2})&=^\text{IH} O(\max(|\vec t_1|^{d_1},O(|\vec t_2|^{d_2}))\\
&=O(|\vec{t}|^{\max(d_1,d_2)}).
\end{align*}
\item If $\vec{t}= \alpha\cdot\vec{s}$, then $|\vec t|=|\vec s|$, and $|\lpar{\vec{t}}|=|\lpar{\vec{s}}|=^\text{IH} O(|\vec{s}|^{d})=O(|\vec{t}|^d)$.
\end{itemize}
The degree of the polynomial bounding the size of $\lpar{\vec{t}}$ is only increased in translating duplications (see case $\vec{t}=t_1\ t_2$), so it will only depend on the number of such duplications present in the term, which is bounded by the type of the term.
\end{proof}

\lemvalues*

\begin{proof}
  By induction on the structure of $\vec v$. 
\begin{itemize}
\item $\vec{v}=\ket{0}\,$. Then $\lpar{\vec{v}}=\{\lambda x.\lambda y. x\}$, which does not reduce. Same for $\vec{v}=\ket{1}\,$.
\item $\vec{v}=\lambda x.\vec{t}$. Then $\lpar{\vec{v}}=\lambda x.\lpar{\vec{t}}$ which will not reduce in DLAL in call by value.
\item $\vec{v} =\pair{v_1}{v_2}$, in which case $\lpar{\vec{v}}=\lambda x. (x\ \lpar{v_1}\ \lpar{v_2})$. By IH, we consider that no element in $\lpar{v_1}\cup \lpar{v_2}$ reduces in DLAL, therefore $\lpar{\vec{v}}$ will not reduce.
\item $\vec{v}=\vec{0}$, where we have $\lpar{\vec{v}}=\{\ast\}$, which by definition does not reduce.
\item $\vec{v}=\alpha \cdot \vec{w}$, where we apply the IH to $\lpar{\vec{w}}=\lpar{\vec{v}}$.
\item $\vec{v}=\vec{v}_1+\vec{v}_2$, where we apply the IH to $\lpar{\vec{v}_1}$ and $\lpar{\vec{v}_2}$, since $\lpar{\vec{v}}=\lpar{\vec{v}_1}\cup\lpar{\vec{v}_2}$.\qedhere
\end{itemize}  
\end{proof}

\lemtransitiondlal*

\begin{proof}
By induction on the reduction of $\vec t$.
\begin{itemize}
\item[(If$_0$)] $\vec{t}=\tif{\ket{0}}{\vec{s_0}}{\vec{s_1}}\rightsquigarrow \vec{s_0}$, therefore $\lpar{\vec{r}}= \lpar{\vec{s_0}}$ and we have
\begin{align*}
\lpar{\vec{t}}=&((\lambda x.\lambda y.x)\ \lambda x_{1\dots k}.\lpar{\vec{s_0}}\ \lambda x_{1\dots k}.\lpar{\vec{s_1}})\ x_1\dots x_k\\
\rightarrow & ((\lambda y. \lambda x_{1\dots k}.\lpar{\vec{s_0}})\ \lambda x_{1\dots k}.\lpar{\vec{s_1}})\ x_1\dots x_k \\
\rightarrow & (\lambda x_{1\dots k}.\lpar{\vec{s_0}})\ x_1\dots x_k \\
\rightarrow &^k  \lpar{\vec{s_0}} 
\end{align*}

\item[(If$_1$)] $\vec{t}=\tif{\ket{1}}{\vec{s_0}}{\vec{s_1}}\rightsquigarrow \vec{s_1}$, therefore $\lpar{\vec{r}}= \lpar{\vec{s_1}}$ and we have
\begin{align*}
\lpar{\vec{t}}=&((\lambda x.\lambda y.x)\ \lambda x_{1\dots k}.\lpar{\vec{s_0}}\ \lambda x_{1\dots k}.\lpar{\vec{s_1}})\ x_1\dots x_k\\
\rightarrow & ((\lambda y.y)\ \lambda x_{1\dots k}.\lpar{\vec{s_1}})\ x_1\dots x_k \\
\rightarrow & (\lambda x_{1\dots k}.\lpar{\vec{s_1}})\ x_1\dots x_k \\
\rightarrow &^k  \lpar{\vec{s_1}} 
\end{align*}

\item[(Abs)] $\vec  t = (\lambda x.\vec{p})\ v\rightsquigarrow \vec{p}[v/x]$, and therefore
\begin{align*}
\lpar{\vec{t}}=\lpar{(\lambda x.\vec{p})\ v}&=\lambda x_1\dots \lambda x_k.\lpar{\vec{p}_{[x:1\dots k]}}\ \overbrace{\lpar{v}\ \dots\ \lpar{v}}^{k\ \text{times}}\\
& =\{(\lambda x_1\dots\lambda x_k .p_i)\ v_1\ \dots v_k\mid p_i\in \lpar{\vec{p}_{[x:1\dots k]}},\ v_j\in\lpar{v}, j=1\dots k\}\\
&\rightarrow \{p_i[v_1/x_1,\dots, v_k/x_k] \mid p_i \in \lpar{\vec{p}},\ v_j\in\lpar{v}, j=1\dots k\}\\
&= \lpar{\vec{p}}[\lpar{v}/x]= \lpar{\vec{p}[v/x]}=\lpar{\vec{r}}
\end{align*}

\item[(Let)] $\vec t=(\tlet{x}{y}{\pair{v}{w}}{\vec s})                                                                                                                                                                                                                                                                                                                                                                                                                                                                                                                                                                                           $, then $\vec t\rightsquigarrow \vec{s}[v/x,w/y]=\vec{s}[v/x][w/y]$, since $y\not\in FV(v)$.
\begin{align*}
\lpar{\vec{t}}&=\lambda z.(z\ \lpar{v}\ \lpar{w})\ \lambda x.\lambda y.\lpar{\vec{s}}\\ &\rightarrow (\lambda x.\lambda y.\lpar{\vec{t}})\ \lpar{v}\ \lpar{w}\\
& \rightarrow (\lambda y.\lpar{\vec s}[\lpar{v}/x]) \lpar{w}\\
& \rightarrow \lpar{\vec{s}} [\lpar{v}/x][\lpar{w}/y]\\
& = \lpar{\vec{s}[v/x]}[\lpar{w}/y] &\text{(by Lemma~\ref{lemma:substitution})}\\
& = \lpar{\vec{s}[v/x][w/y]} & \text{(by Lemma~\ref{lemma:substitution})}\\
& = \lpar{\vec{r}}, &\text{ since }y\not \in FV(v).
\end{align*}
\item[(If$_+$)] If $s_1\rightsquigarrow \vec{s}_2$, then $\vec{t}=(\tif{s_1}{\vec{p}_1}{\vec{p}_2})\rightsquigarrow (\tif{\vec{s}_2}{\vec{p}_1}{\vec{p}_2})=\vec{r}$. By the induction hypothesis, $\lpar{\vec s_2}\subseteq \bigcup_{n>0} \{\mathcal{R}_n\mid \lpar{s_1}\rightarrow^n \mathcal{R}_n\}$.
\begin{align*}
\lpar{\vec{r}}=\lpar{\tif{\vec{s_2}}{\vec{p}_1}{\vec{p}_2}}& =\lpar{\vec{s_2}}\ \lpar{\vec{p}_1}\ \lpar{\vec{p}_2} \\
&\subsetast^\text{IH} \bigcup_{n>0} \{\mathcal{R}_n\ \lpar{\vec{p_1}}\ \lpar{\vec{p}_2}\mid \lpar{s_1}\rightarrow^n \mathcal{R}_n\}\\
&\subseteq \bigcup_{n>0} \{\mathcal{S}_n\mid \lpar{s_1}\ \lpar{\vec{p_1}}\ \lpar{\vec{p_2}}\rightarrow^n \mathcal{S}_n\}\\
& =\bigcup_{n>0} \{\mathcal{S}_n\mid \lpar{s_1\ \vec{p}_1\ \vec{p}_2 }\rightarrow^n \mathcal{S}_n\}
\end{align*}
\item[(App)] Let $\lpar{\vec s_2}\subseteq \bigcup_{n>0} \{\mathcal{R}_n\mid \lpar{s_1}\rightarrow^n \mathcal{R}_n\}$, then
\begin{align*}
\lpar{\vec{r}}=\lpar{p\ \vec{s}_2}&= \lpar{p}\ \lpar{\vec{s}_2}\\
&\subsetast^\text{IH} \bigcup_{n>0} \{\lpar{p}\ \mathcal{R}_n\mid \lpar{s_1}\rightarrow^n \mathcal{R}_n\}\\
&\subseteq \bigcup_{n>0} \{\mathcal{S}_n\mid \lpar{p}\ \lpar{s_1}\rightarrow^n \mathcal{S}_n\}\\
&=\bigcup_{n>0} \{\mathcal{S}_n\mid \lpar{p\ s_1}\rightarrow^n \mathcal{S}_n\}
\end{align*}
\item[(App$_\V$)] Similar to (App).
\item[(Pair)] Let $\lpar{\vec s_2}\subseteq \bigcup_{n>0} \{\mathcal{R}_n\mid \lpar{s_1}\rightarrow^n \mathcal{R}_n\}$, then
\begin{align*}
\lpar{\vec{r}}=\lpar{\pair{\vec s_2}{p}}&= \lambda x.(x\ \lpar{\vec{s}_2}\ \lpar{p})\\
&\subsetast^\text{IH} \bigcup_{n>0} \{\lambda x.(x\ \mathcal{R}_n\ \lpar{p} )\mid \lpar{s_1}\rightarrow^n \mathcal{R}_n\}\\
&\subseteq \bigcup_{n>0} \{\mathcal{S}_n\mid \lambda x.(x\ \lpar{s_1}\ \lpar{p})\rightarrow^n \mathcal{S}_n\}\\
&=\bigcup_{n>0} \{\mathcal{S}_n\mid \lpar{\pair{s_1}{p}}\rightarrow^n \mathcal{S}_n\}
\end{align*}
\item[(Pair$_\V$)] Similar to (Pair).
\item[(Let$_+$)] Let $\vec t=(\tlet{x}{y}{s_1}{\vec p})                                                                                                                                                                                                                                                                                                                                                                                                                                                                                                                                                                                           $ such that $s_1\rightsquigarrow \vec s_2$. Then 
\begin{align*}
\lpar{\vec{r}}&=\lpar{\vec s_2}\ \lambda x.\lambda y.\lpar{\vec{p}}\\
&\subsetast^\text{IH} \bigcup_{n>0} \{ \mathcal{R}_n\ \lambda x.\lambda y.\lpar{\vec{p}}\mid \lpar{s_1} \rightarrow^n \mathcal{R}_n\}\\
&\subseteq \bigcup_{n>0}\{\mathcal{S}_n\mid \lpar{s_1}\ \lambda x.\lambda y.\lpar{\vec p}\rightarrow^n \mathcal{S}_n\}\\
& =  \bigcup_{n>0}\{\mathcal{S}_n\mid \lpar{s_1\ \lambda x.\lambda y.\vec p}\rightarrow^n \mathcal{S}_n\}
\end{align*}
\item[(Sup)] Let $\vec{t}=\sum_{i\in I} \alpha_i\cdot p_i + \sum_{j\in J} \beta_j\cdot v_j$ such that $\forall i \in I,\ p_i\rightsquigarrow s_i$. Then, by the induction hypothesis, we have that $\forall i \in I,\ \lpar{s_i}\subseteq \bigcup_{n>0} \{ \mathcal{S}_n \mid \lpar{p_i}\rightarrow^n \mathcal{S}_n \}.$ We can see that
\begin{align*}
\lpar{\vec{r}}&=\bigcup_{i\in I} \lpar{s_i} \cup \bigcup_{j\in J} \lpar{v_j}\\
&\subsetast^\text{IH} \bigcup_{i\in I}\bigg(\bigcup_{n>0} \{\mathcal{S}_n \mid \lpar{p_i} \rightarrow^n \mathcal{S}_n\}\bigg)  \cup \bigcup_{j\in J} \lpar{v_j}\\
&=\bigcup_{n>0} \{\mathcal{S}_n \mid \lpar{\vec{t}} \rightarrow^n \mathcal{S}_n\}. & \text{(by Lemma~\ref{lem:values})}
\end{align*}
\item[(Equ)] Since we have that $\vec{s_1}\equiv \vec{s_2}$ and $\vec{s_2}\rightsquigarrow \vec{s_3}$ and $\vec{s_3}\equiv \vec{s_4}$, by the induction hypothesis the following are true:
\begin{align*}
(\text{IH1})\quad & \lpar{\vec{s_1}}\eqast\lpar{\vec{s_2}},\\
(\text{IH2})\quad & \lpar{\vec s_3}\subsetast \bigcup_{n>0} \{\mathcal{S}_n\mid \lpar{\vec s_2}\rightarrow^n \mathcal{S}_n\},\\
(\text{IH3})\quad & \lpar{\vec{s_3}}\eqast\lpar{\vec{s_4}}.
\end{align*}
We may then see that:
\begin{align*}
\lpar{\vec{s_4}}\eqast^\text{IH3}\lpar{\vec{s_3}}\subsetast^{\text{IH2}}  \bigcup_{n>0} \{\mathcal{S}_n\mid \lpar{\vec s_2}\rightarrow^n \mathcal{S}_n\}\eqast^{\text{IH1}}  \bigcup_{n>0} \{\mathcal{S}_n\mid \lpar{\vec s_1}\rightarrow^n \mathcal{S}_n\}.
\end{align*}
\end{itemize}
This concludes the proof.
\end{proof}

\subsection{Proofs of Section~\ref{s:type checking}}
\label{app:type-checking}

In this section of the Appendix we consider the proofs and general points addressed in Section~\ref{s:type checking}.

\paragraph*{Restriction to the algebraic numbers} 
It is important to note the role that the complex coefficients in \punq{} play in the decidability of the typing system. In \punq{}, the restriction of amplitudes to the set $\Ct$ of algebraic numbers ensures that, for any amplitude $\alpha$, we are capable of checking if $\alpha=0$ and therefore, for some $\vec{t}$, if $\alpha\cdot \vec{t}\equiv\vec{0}$. This can be done efficiently when $\alpha\in\Ct$, as shown in~\cite[Proposition 2.2]{HHHJ05}. Allowing for the full set of complex numbers, on the other hand, allows for the encoding of undecidable problems, as shown in~\cite{ADH97}.

\paragraph*{Reducing ground types} In this section, we will be interested in the time complexity of reducing terms in order to check if they are orthogonal (see Section~\ref{ss:orthogonality}). In the program semantics (Figure~\ref{fig:sem}), we allow in rule (Sup) for one-step reductions to alter different parts of a superposition simultaneously, so as to capture a single-step evolution of a quantum machine. However, if type checking requires performing some reduction steps, it must do so in the \emph{classical} model, and therefore, we must consider the number of steps required by a classical Turing machine to perform the reduction.

From the typing rules (Figure~\ref{fig:type}), we are only interested in checking the orthogonality of terms to which we have successfully attributed some ground type $Q$. This limitation ensures a bounded dimension in the set of realizers, and a maximum number of base values that can be in a superposition. For such a $Q$, we define this number the \emph{cardinality of $Q$}.

\begin{definition}[Ground type dimension] For any ground type $Q$, we define the \emph{cardinality of $Q$}, written $\card(Q)$, inductively as:
\[\card(\B)\triangleq 2,\quad \card(\sharp Q)\triangleq \card(Q),\quad \card(\S Q)\triangleq \card(Q),\quad \card(Q\times R)\triangleq \card(Q)\card(R).\]
Using the unitary semantics of Section~\ref{s:unitarity}, we can, in an equivalent way, define $\card(Q)$ as the cardinality of $\flat \llbracket Q\rrbracket_\emptyset$.
\end{definition}

\begin{restatable}[]{lemma}{lemmaclassreduct}\label{lemma:class-reduct}
Let $;\vdash \vec{t}:Q$ such that $\vec{t}\rightsquigarrow^* \vec{v}$ for $\vec{v}\in\V$. Then $\vec{v}$ can be computed from $\vec{t}$ classically in time at most $O(\card(Q)\,\textnormal{poly}(|\vec{t}|))$.
\end{restatable}

\begin{proof}
Since $\vec{t}\in\llbracket Q\rrbracket_\emptyset$, by the unitarity semantics, we have that $\vec{t}\equiv\sum_{i=1}^n \alpha_i\cdot t_i$, such that $\forall i, t_i\in\flat\llbracket Q\rrbracket_\emptyset$ and $n\leq\dim(Q)$. The treatment of each $t_i$ in a reduction step is at most quadratic on $|t_i|$ (e.g. checking if $t_i$ is a value, performing variable substitution, converting to normal form, \dots). This is done for each term $t_i$. Since there are at most $n$ such $t_i$, where $n$ is a constant fixed by the type $Q$, the number of reduction steps in the classical model remains polynomial.
\end{proof}

\lemmachecksincluded*
\begin{proof}
The fact that $\CHK_\emptyset\subseteq \CHK_\times$ is trivial. Similarly, $\ortho_\forall \subseteq \ortho$ since the only difference is that we perform an orthogonality check over all values and not the subset of values that correspond to the correct type. We also have that $\ortho_\times\subseteq \ortho_\forall$, since any variable substitution will still lead to orthogonal terms since the closed subterm which ensures orthogonality will not be altered.
\end{proof}

\lemmaorthoscomplexity*

\begin{proof} We analyze each case separately.
\begin{itemize}
\item $(\ortho_\emptyset\in\mathsf{PTIME})$ Since $;\vdash \vec{t},\vec{s}:Q$, by Theorem~\ref{thm:soundness}, there exists a polynomial $P_Q$ such that $\vec{t}$ and $\vec{s}$ reduce to values, say $\vec{t}\rightsquigarrow^\ast \vec{v}$ and $\vec{s}\rightsquigarrow^\ast \vec{w}$ in at most a number $P_Q(\max(|\vec{t}|,|\vec{s}|))$ of reduction steps. Notice that, since $\vec{t}$ and $\vec{s}$ have type $Q$, in each step of the reduction they are represented by a superposition of at most $\dim(Q)$ terms, therefore the classical complexity of calculating the reduction differs only by a constant factor from the number of reduction steps. Finally, verifying that $\langle \vec{v}|\vec{w}\rangle=0$ is also done in polynomial time since $|\vec{v}|,|\vec{w}|\leq P_Q(\max(|\vec{t}|,|\vec{s}|))$ and we only need to reduce at most a constant $\dim(Q)$ number of terms in superposition.
\item $(\ortho_\times\in \mathsf{PTIME})$ Checking if a term contains free variables can be done in linear time on its size, i.e. in time $O(\max(|\vec{t}|,|\vec{s}|))$. If the terms are closed, computing $\ortho_\emptyset(\vec{t},\vec{s})$ can be done in polynomial time. If the terms are open, we proceed inductively on the structure of the term until we find (or not) matching subterms that are closed and from which we can prove orthogonality. This can easily be done in polynomial time.
\item $(\ortho_\mathsf{untyped}\in \Pi_1^0)$ Since $\V_c$ is a countable set, it suffices to show that $\mathsf{PERP}$ is computable. To compute $\mathsf{PERP}$, we only need to reduce the terms to values $\vec{w}_1$ and $\vec{w}_2$. Notice that, since the terms which we are using in the substitution may not have the correct type, this implies that the term we are reducing might not be typable in \punq{}, which renders Theorem~\ref{thm:soundness} inapplicable and therefore we cannot place a (in particular, finite) bound on its number of reduction steps. However, by Theorem~\ref{thm:soundness}, we know that we need only apply a polynomial $P_Q$ number of reduction steps and, if the term is not yet a value (i.e., by Theorem~\ref{thm:progress}, it does not reduce), we know that it does not correspond to the candidate type $Q$ in \punq{}, and therefore we do not need to consider it.

If the terms do reduce to values $\vec{w}_1$ and $\vec{w}_2$, then we may check if they belong to $\llbracket Q\rrbracket_\emptyset$, the set of realizers of $Q$. Notice that, since the grammar of ground types $Q$ does not include applications nor polymorphisms, we can check that a term belongs in $\llbracket Q\rrbracket_\emptyset$ purely syntactically. Finally, we can check $\ortho_\emptyset$ in polynomial time.\qedhere
\end{itemize}
\end{proof}

\end{document}